\documentclass[11pt]{article}

\usepackage{amstext,amsmath,amssymb,amsfonts,bbm}
\usepackage{hyperref}
\usepackage{authblk}
\usepackage{caption}
\usepackage{subcaption}
\usepackage[latin1]{inputenc}
\usepackage{epsfig}
\usepackage{amsthm}
\usepackage{color}
\usepackage{graphicx}
\usepackage{slashed}
\usepackage{float}
\usepackage{enumitem}

\usepackage[lmargin=80pt,rmargin=80pt,tmargin=80pt,bmargin=80pt]{geometry}

\usepackage{comment}

\newcommand{\be}{\begin{equation}}
\newcommand{\ee}{\end{equation}}

\newcommand{\f}{\frac}
\newcommand{\p}{\partial}

\newcommand{\Tr}{{\rm Tr}}

\newcommand{\la}{\langle}
\newcommand{\ra}{\rangle}

\newtheorem{proposition}{Proposition}
\newtheorem{definition}{Definition}
\newtheorem{lemma}{Lemma}
\newtheorem{theorem}{Theorem}
\newtheorem{claim}{Claim}

\theoremstyle{remark}
\newtheorem{remark}{Remark}

\def\uN{\mathrm{U}(N)}
\def\oD{\mathrm{O}(D)}
\def\Ne{\mathrm{N{}_e}}
\def\No{\mathrm{N{}_o}}

\let\a=\alpha \let\b=\beta  \let\g=\gamma  \let\d=\delta
       \let\k=\kappa \let\l=\lambda
\let\m=\mu    \let\n=\nu           \let\om=\omega
 \let\t=\tau     
    
\let\G=\Gamma     \let\X=F
         
  \let\eps=\epsilon

\newcommand{\cC}{\mathcal{C}}
\newcommand{\cD}{\mathcal{D}}
\newcommand{\cF}{\mathcal{F}}
\newcommand{\cG}{\mathcal{G}}

\newcommand{\cL}{\mathcal{L}}
\newcommand{\cM}{\mathcal{M}}

\newcommand{\cS}{\mathcal{S}}
\newcommand{\cT}{\mathcal{T}}

\begin{document}

\title{\bf Multiple scaling limits of $\uN^2 \times \oD$\\ multi-matrix models}

\author[1]{Dario Benedetti}
\author[2]{Sylvain Carrozza}
\author[3]{Reiko Toriumi}
\author[4]{Guillaume Valette}

\affil[1]{\normalsize\it CPHT, CNRS, Ecole Polytechnique, Institut Polytechnique de Paris,\authorcr\it Route de Saclay, 91128 Palaiseau, France. \hfill}
\affil[2]{\normalsize\it  Perimeter Institute for Theoretical Physics, 31 Caroline St N, Waterloo, ON N2L 2Y5, Canada. \hfill}
\affil[3]{\normalsize\it Okinawa Institute of Science and Technology Graduate University, 1919-1, Tancha, Onna, Kunigami District, Okinawa 904-0495, Japan. \hfill}
\affil[4]{\normalsize\it Service de Physique Th\'eorique et Math\'ematique,
Universit\'e Libre de Bruxelles (ULB) and The International Solvay Institutes,
Campus de la Plaine, CP 231, B-1050 Bruxelles, Belgium. \authorcr
Emails: {\rm\url{dario.benedetti@polytechnique.edu}, \url{scarrozza@perimeterinstitute.ca}, \url{reiko.toriumi@oist.jp}, \url{gvalette@ulb.ac.be}.} \authorcr \hfill }

\date{}

\maketitle

\hrule\bigskip

\begin{abstract}

\noindent We study the double- and triple-scaling limits of a complex multi-matrix model, with $\mathrm{U}(N)^2\times \mathrm{O}(D)$ symmetry. The double-scaling limit amounts to taking simultaneously the large-$N$ (matrix size) and large-$D$ (number of matrices) limits while keeping the ratio $N/\sqrt{D}=M$ fixed. The triple-scaling limit consists in taking the large-$M$ limit while tuning the coupling constant $\lambda$ to its critical value $\lambda_c$ and keeping fixed the product $M(\lambda_c-\lambda)^\alpha$, for some value of $\alpha$ that depends on the particular combinatorial restrictions imposed on the model. 
Our first main result is the complete recursive characterization of the Feynman graphs of arbitrary genus which survive in the double-scaling limit.
Next, we classify all the dominant graphs in the triple-scaling limit, which we find to have a plane binary tree structure with decorations. Their critical behavior belongs to the universality class of branched polymers.
Lastly, we classify all the dominant graphs in the triple-scaling limit under the restriction to three-edge connected (or two-particle irreducible) graphs. 
Their critical behavior falls in  the  universality  class  of Liouville  quantum  gravity  (or,  in  other  words,  the  Brownian sphere).

\end{abstract}

\hrule\bigskip

\tableofcontents

\section{Introduction}

Field theories in zero spacetime dimensions are in principle nothing but ordinary integrals, yet they are of great interest in physics and combinatorics. Their perturbative expansion in Feynman diagrams requires no integrals over spacetime, hence it reduces to counting certain classes of diagrams: standard field-theoretic objects, such as the partition function, the free energy, and so on, are generating functions of counting sequences for various combinatorial classes.
In physics, such theories are often viewed as toy models, but for certain models of multi-variable integrals the relevant representation of Feynman diagrams has a natural interpretation in terms of geometrical objects, thus opening the possibility of using such models to construct theories of random geometry, which is of interest for quantum gravity and string theory.

Matrix and tensor models, distinguished by different symmetry groups or different representations of the same group, are clear examples of such an application. Their Feynman diagrams are dual to piecewise flat manifolds of dimension two (for matrices) and higher (for tensors), hence their free energy is the generating function of connected piecewise flat manifolds, which in the continuum limit is expected to define a functional integral for Euclidean quantum gravity  \cite{DiFrancesco:1993nw,EDTAmbjorn}.
However, such generating functions are only a formal construction, as the perturbative expansion is divergent due to the large proliferation of Feynman diagrams. A powerful organizational principle that allows to construct well-defined generating functions is the large-$N$ limit, where $N$ is the dimension of the (real or complex) vector space on which the matrices or tensors act as (multi-)linear maps.
For matrix models, the large-$N$ limit leads to an expansion in $1/N$ indexed by the genus of the piecewise flat surfaces \cite{'tHooft:1973jz}, and for each fixed genus the generating function is a convergent series. For tensor models, there is also a $1/N$ expansion, indexed in this case by a non-topological number, known as the degree \cite{expansion1,expansion2,expansion3,expansion4,Dartois:2013he,Carrozza:2015adg,Gurau:2017qya,Benedetti:2017qxl,Carrozza:2018ewt}. The leading order (vanishing degree) of such an expansion is well understood, and it is dominated by melonic graphs \cite{critical}. For the higher orders, it is known that there is only a finite number of schemes of fixed degree \cite{GurSch,Fusy:2014rba}, where each scheme corresponds to a summable class of graphs. However, except for the lowest orders (small degree) \cite{Kaminski:2013maa,Bonzom:2017pqs, Bonzom:2019yik, Bonzom:2019moj}, little is known about the classification of graphs of arbitrary degree.

Given a summable class of Feynman diagrams, we are interested in their continuum limit and continuum probability distribution, the functional integral. For the class of diagrams selected by the large-$N$, such a continuum limit is possible, and it leads to two well-defined and ubiquitous continuum probabilistic models: the continuous random tree (branched polymers) \cite{aldous1991, aldous} and the Brownian sphere \cite{legall2013, miermont2013}.
Typical matrix models lead to the latter, while typical tensor models to the former \cite{melbp}, but special models can be built with exchanged continuum limits \cite{Das:1989fq,Bonzom:2015axa}.
Finding alternatives to these two limits remains an open problem, of particular relevance to quantum gravity, as neither branched polymers nor the Brownian spheres approach a classical geometry at large distances.
It would therefore be desirable to find other variants of large-$N$ limits selecting different classes of Feynman diagrams that could lead to new continuum models.\footnote{The inverse problem could also be interesting. For example, higher-dimensional generalizations of the Brownian map have been introduced in \cite{Lioni:2019tdb}, but at the moment it is not known how to obtain them from a field theory.}

One interesting variant of large-$N$ limit was introduced by Ferrari in \cite{Ferrari:2017ryl}, and further developed and generalized in \cite{Ferrari:2017jgw, Azeyanagi:2017mre}. The fundamental variables can be described either as $D$ different $N\times N$ matrices, or as a tensor describing a bilinear map from a vector space of dimension $N$ to one of dimension $D$ (i.e.\ a tensor whose array forms a rectangular cuboid rather than a cube).
Then, one has two parameters to play with, and different limits can be considered.
In particular, the large-$N$ limit of such models leads again to a genus expansion, but on top of that one can perform the large-$D$ limit as well, thus introducing a new expansion at each fixed genus.
In a typical model, sending both $N$ and $D$ to infinity leads back to the large-$N$ limit of cubic tensors, with melonic dominance, and hence to a branched polymer model.
However, higher orders of such double expansion have not been explored yet.

In this paper we consider a complex multi-matrix model, with $\mathrm{U}(N)^2\times \oD$ symmetry. The $\uN^2$ group acts on the single matrices, while the $\oD$ group acts as a mixing of the matrices. The interaction vertex is taken to be of a tetrahedral type, with two pairs of matrices appearing in a trace of order four. This results in a nice geometric interpretation of the Feynman graphs. First, as in similar one-matrix models, the ribbon structure generated by the propagation of $\mathrm{U}(N)^2$ indices is dual to quadrangulated orientable surfaces of arbitrary genus. Second, as the result of the $\oD$ symmetry, these surfaces are decorated by specific patterns of cycles, referred to as $\oD$-loops, which only intersect at vertices. In this combinatorial space, the statistical properties of the genus are controlled by the large-$N$ limit, while the proliferation of $\oD$-loops -- captured by a second integer number \cite{Ferrari:2017ryl} that we call the grade -- is directly tied to the large-$D$ limit. This already suggests that, by suitably correlating the two limits, one might be able to balance the interactions between topological (genus) and combinatorial ($\oD$-loops) aspects, in such a way that the model is driven to different universality classes in the continuum. With this idea in mind, we address three questions in the present manuscript:
\begin{enumerate}
    \item Given a fixed value of the genus (that is, at fixed order in the $1/N$ expansion), what are the Feynman graphs that survive in the large-$D$ limit?
    \item Are there interesting double-scaling limits, in which $N$ and $D$ are sent to infinity in a correlated way, which allow to resum all the dominant fixed-genus contributions into a single generating function? 
    \item What universality classes of continuum geometry do these double-scaled generating functions lead to at criticality? 
\end{enumerate}
We finally note that, while the formulation and resolution of the last two questions is tied to the random-geometric context, the general characterization of higher-genus leading-order Feynman graphs we will provide might be of broader interest. In higher dimension, it might for instance provide an opportunity to embed the melonic regime of SYK-like tensor/matrix quantum-mechanical models \cite{Witten:2016iux,Klebanov:2016xxf,Azeyanagi:2017drg,Carrozza:2018psc} and large-$N$ tensor quantum field theory \cite{Giombi:2017dtl,Prakash:2017hwq,Benedetti:2017fmp,Giombi:2018qgp,Benedetti:2018ghn,Benedetti:2019eyl,Benedetti:2019ikb,Benedetti:2019rja}, into a genus expansion tractable enough to allow explicit computations (of e.g. quantum corrections to operator dimensions).

\paragraph{Main results.} Our first main result will be the complete recursive characterization -- in Proposition~\ref{propo:g1} and Theorem~\ref{thm:induction} -- of Feynman graphs with vanishing grade and arbitrary genus, which are precisely the ones that survive in a double-scaling limit where $N$ and $D$ are sent to infinity while keeping the 
ratio
$N^2/D$ finite. 

We will then restrict our attention to a particular subclass of graphs -- the \emph{dominant} graphs -- which turn out to govern the critical regime of this theory. Proposition \ref{propo:dominant-schemes} will establish that the latter have a plane binary tree structure. As a result, their partition function (which we obtain through a triple-scaling of the multi-matrix model) has a critical point dominated by large trees, which however describe orientable surfaces with large genus. The expectation value of the genus (or equivalently the size of the trees) diverges at criticality, and even though its samples look naively quite different from a tree, this ensemble converges in the continuum limit to the universality class of branched polymers.

Besides the free energy, which can be viewed as the generating function of vacuum diagrams with $v$ vertices, it is quite natural in combinatorics to consider other generating functions, with further restrictions. 
In Sec.~\ref{sec:2PI}, we will restrict to the class of 3-edge connected graphs, also known in physics as two-particle irreducible (2PI) graphs. This in particular forbids tadpoles and triple edges.
On top of being an interesting class of diagrams from the combinatorial point of view, the 2PI restriction is also natural in the context of two-dimensional quantum gravity, where tadpoles and multiple edges are viewed as dual to degenerate quadrangulations \cite{EDTAmbjorn}. 
Surprisingly, with such a restriction in the triple-scaling limit, we will find a very different critical behavior, falling in the universality class of Liouville quantum gravity (or, in other words, the Brownian sphere) \cite{David:1984tx, duplantier2011liouville, Duplantier:2014daa, David:2014aha}. This change of universality class resulting from the 2PI condition is an interesting new feature, not shared by the usual large-$N$ limit, which in contrast displays universality under such type of change of ensembles \cite{Brezin:1977sv}.

\paragraph{Plan of the paper.} The paper is organized as follows. In section \ref{sec:model}, we introduce the multi-matrix model, review what was previously known about it, and describe our general strategy towards the definition of non-trivial double- and triple-scaling limits. We then move on to the algorithmic characterization of Feynman graphs with vanishing grade, which is the subject of section \ref{sec:recursion}, and culminates in Theorem~\ref{thm:induction}. 
Most of the combinatorial concepts and tools we rely on in the paper are introduced in this section. In particular, the key notion of \emph{scheme}, which describes equivalence classes of graphs defined up to insertion of melonic or ladder subgraphs, is thoroughly reviewed \cite{GurSch,Fusy:2014rba}. In section \ref{sec:dominantSchemes}, we specialize our analysis to the subclass of schemes which govern the dominant singularities of the partition function. A more precise (and non-inductive) characterization of the dominant schemes in terms of plane binary trees is proposed in Proposition~\ref{propo:dominant-schemes}, which in turn allows to pinpoint that the model flows to the universality class of branched polymers at criticality. Finally, section \ref{sec:2PI} focuses on the more interesting phase generated by the 2PI generating function. After showing how to directly impose the 2PI restriction at the level of the matrix integral, we explain that its main virtue is to remove the two classes of subgraphs which typically proliferate in tensor models, and are responsible for the tree-like structure of the dominant schemes: namely, melonic subgraphs and broken ladders. As a result, the dominant singularities are determined by a richer family of geometries, which we call \emph{2PI-dominant schemes}. Their allowed combinatorial structures are elucidated in subsection~\ref{sec:combinatorics_2PI} (Proposition~\ref{propo:dominant-2PI}) and mapped in subsection~\ref{sec:Ising} to Ising states on cubic planar maps. As a result, and quite remarkably, the generating function of 2PI-dominant schemes is equivalent to an Ising model on a certain family of random spheres. This model is investigated in some details by means of an effective two-matrix model in its planar limit (subsection~\ref{sec:effective-matrix}), and with the help of known map enumeration results (subsection~\ref{sec:2PI-enumeration}). The two methods yield consistent results, and allow to conclude that the triple-scaling of the 2PI model lies in the universality class of two-dimensional quantum gravity. In particular, the expectation value of the genus of 2PI-dominant schemes (or equivalently the number of nodes in their cubic map representation) remains finite, but its variance diverges at the critical point.

\section{The model and the general idea}\label{sec:model}

We consider an $\oD$-invariant complex matrix model in zero dimension. The basic degrees of freedom are given by $D$ complex matrices $X_\m$ of size $N\times N$: $(X_\mu)_{ab} = X_{\mu ab}$ with $1\leq \mu\leq D$ and $1\leq a,b\leq N$. Remark that writing the matrices $X_\mu$ in terms of their components $X_{\mu ab}$ makes it evident that we can think of them as the components of a rank-$3$ tensor with indices having different ranges. 
A fundamental quantity to be determined is the free energy
\be \label{eq:free-en-def}
\cF(\l) = \log \int [dX] \, e^{-S[X,X^\dagger]} \,,
\ee
with an action  $S[X,X^\dagger]$ to be specified, and a measure $[dX] = \prod_{\mu,a,b} d{\rm Re} (X_\mu)_{ab}\, d{\rm Im}(X_\mu)_{ab}$.
The global symmetry group, under which the  action is invariant, is assumed to be $U(N)^2\times O(D)$, with the following transformation law 
\be \label{eq:transfLaws}
X_\mu \rightarrow X'_\mu = O_{\mu \mu'} U_{(L)} X_{\mu'} U_{(R)}^{\dagger}\, ,
\ee
where $O$ is an orthogonal matrix in $O(D)$, while $U_{(L)}$ and $U_{(R)}$ are two independent unitary matrices in two distinct copies of the group $U(N)$, which we call left and right, respectively. As a result, the two matrix indices (which we omitted in Eq.\eqref{eq:transfLaws} and in the following, as standard matrix multiplication is assumed) are distinguishable because they transform with respect to two distinct $U(N)$ groups.

Models of this type have been studied in \cite{Ferrari:2017ryl, Azeyanagi:2017drg, Ferrari:2017jgw, Azeyanagi:2017mre, Ferrari:2019ogc}. In this paper, we focus on the following invariant action:\footnote{Other quartic interaction terms compatible with the symmetries, namely $\Tr\bigl[ X^\dagger_\m X_\m X^\dagger_\n X_\n \bigr] $ and $\Tr\bigl[ X^\dagger_\m X_\m \bigr] \Tr\bigl[  X^\dagger_\m X_\n \bigr] $, could be added to the action, but they are of lesser interest to us, as they do not form melonic diagrams at leading order, and thus we will not consider them.}
\be \label{eq:actionND}
S[X,X^\dagger] = ND \Bigl(  \Tr\bigl[ X^\dagger_\m X_\m \bigr] - \frac{\l}{2}\sqrt{D} \, \Tr\bigl[ X^\dagger_\m X_\n X^\dagger_\m X_\n \bigr] \Bigr)\,,
\ee
where the interaction term is known as the tetrahedral interaction and $\lambda$ is the corresponding coupling constant. 
In this action, the coupling constant has been scaled in such a way that it is kept fixed as $N,D \rightarrow + \infty$. Indeed, this is the right scaling so as to obtain well-defined large $N$ and large $D$ expansions \cite{Ferrari:2017ryl}, as further detailed below. 

The perturbative expansion in $\l$ of the free energy $\cF(\l)$ admits a graphical representation in terms of Feynman graphs. These Feynman graphs can be represented in three equivalent ways:
\begin{itemize}

\item as connected $4$-regular directed orientable maps (with self-loops/tadpoles and multiple edges allowed) such that: each vertex has two outgoing and two ingoing half-edges, and furthermore, the two outgoing (resp.\ two ingoing) half-edges appear on opposite sides of the vertex (see Figure \ref{fig:vertex_propa}, left panel);

\item as connected $4$-regular directed orientable stranded graphs, obtained from the above representation upon replacement of each edge by a triple of parallel strands: two external and one internal, as illustrated in the top right corner of Figure \ref{fig:vertex_propa}. The external strands carry the indices of the two $\uN$ symmetry groups, and can therefore be distinguished. An external strand is called \emph{left} (resp.\ \emph{right}) if it is on the left (resp.\ right) side, with respect to the orientation, of an edge connecting two half-edges. Besides, the internal strand corresponds to the $\oD$ symmetry group. The contraction pattern of the three types of strands at each vertex follows from the structure of the tetrahedral interaction (see Figure \ref{fig:vertex_propa}, right panel). It is such that a strand of a given type (left, right, or internal) is always connected to another strand of the same type. In a Feynman graph, the strands are closed into loops. The loops made out of external strands correspond to the faces of the underlying $4$-regular map, and together form a ribbon graph; we call them \emph{L- or R-faces}, depending on whether they are constituted of left or right strands. As for the loops made out of internal strands, we will call them \emph{straight faces} or \emph{$\oD$-loops}.

\item as connected 4-colored graphs, obtained in the usual way as in tensor models \cite{RTM}.

\end{itemize}

\begin{remark}\label{rem:surfaces}
The Feynman graphs correspond to graphs embedded on Riemann surfaces and are thus dual to discretized surfaces. From this perspective, the $\oD$-loops can be thought of as loops drawn on these discretized surfaces.
In the spirit of tensor models, the Feynman graphs can also be viewed as dual to discretizations of three-dimensional pseudo-manifolds, in general non-orientable \cite{RTM,Tanasa:2015uhr}, but we will not rely on this interpretation in the present paper.
\end{remark} 

\begin{figure}[htb]
    \centering
    \includegraphics[scale=.9]{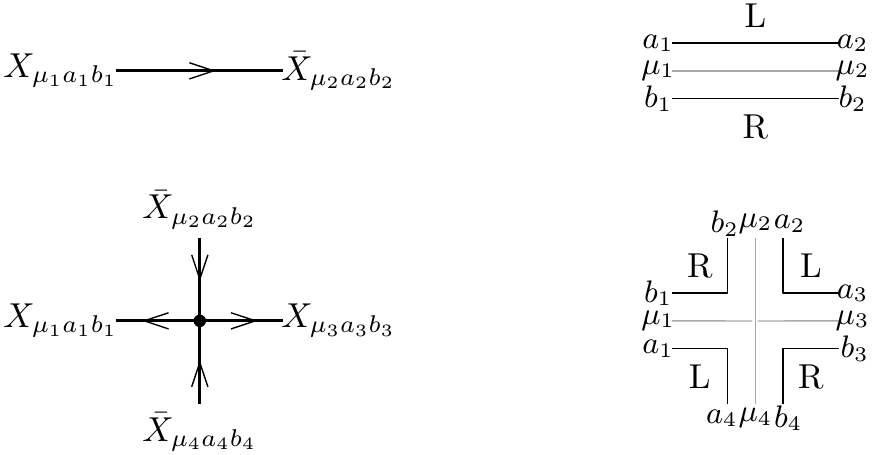}
    \caption{\small Propagator and vertex in the map (left) and stranded graph (right) representations. For simplicity, the edge orientations are left implicit in the stranded representation (note however that fixing the orientation is equivalent to choosing R and L sides).}
    \label{fig:vertex_propa}
\end{figure}

\medskip

As first shown in \cite{Ferrari:2017ryl}, the free energy has a double expansion in $1/N$ and $1/\sqrt{D}$, reading:
\be \label{eq:free-en}
\cF(\l) = \sum_{g\in \mathbb{N}} N^{2-2g} \sum_{\ell\in \mathbb{N}} D^{1+g-\frac{\ell}{2}} \cF_{g,\ell}(\l)\, .
\ee
In this expression, $g\in\mathbb{N}$ is the \emph{genus} of the Feynman graphs, which corresponds to the genus of the $4$-regular maps, or equivalently, to the genus of the corresponding $U(N)^2$ ribbon graphs. It is defined through Euler's relation
\be \label{eq:g}
2-2g = -e+v+f_L+f_R = -v +f \,,
\ee
where $e$ is the number of edges or propagators, $v$ is the number of vertices ($e=2v$ since the maps are $4$-regular), $f_L$ (resp.\ $f_R$) is the number of L-faces (resp.\ R-faces) and $f=f_L+f_R$. The quantity $\ell\in\mathbb{N}$ is another parameter associated with the Feynman graphs; it is related to the \emph{index} (see below), which was defined in full generality in \cite{Ferrari:2017jgw}. In the present case, the parameter $\ell$ is given by:
\be \label{eq:ell}
\frac{\ell}{2} = 2 +  v - \frac{1}{2}f  - \varphi \,,
\ee
where $\varphi$ is the number of straight faces or $\oD$-loops. It can also be expressed, using Eq.\ \eqref{eq:g}, as
\be \label{eq:ell2}
\frac{\ell}{2} = 1+g+ \frac{1}{2}v-\varphi \,.
\ee
As seen from Eq.\ \eqref{eq:free-en}, the parameter $\ell$ introduces an extra grading in the standard genus expansion of matrix models. We therefore refer to it as the \emph{grade}.\footnote{Notice that the parameter $\ell$ was called index in  \cite{Ferrari:2017ryl}, while this name was used for a different quantity in \cite{Ferrari:2017jgw}. In order to avoid confusion, here we introduce a new name for it.}

The fact that the grade is non-negative is made evident by rewriting it as
\be \label{eq:ell3}
\f{\ell}{2} = g_L + g_R \,,
\ee
where $g_L$ (resp. $g_R$) is the genus of the ribbon graph obtained from a Feynman graph in the stranded representation, by deletion of the $L$ (resp. $R$) strands.
Since $g_i \in\frac{\mathbb{N}}{2}$ for $i=L,R$ (the corresponding ribbon graphs are not necessarily orientable), it follows that $\ell\in\mathbb{N}$.

\medskip

Another important combinatorial quantity is the \emph{degree} \cite{Dartois:2013he, Carrozza:2015adg}:
\be \label{eq:omega}
\omega = g+\frac{\ell}{2} \,,
\ee
a close relative of the \emph{Gurau degree} \cite{expansion1, expansion2}, also known as the \emph{index} in the more general context of \cite{Ferrari:2017jgw}. For $D=N$, we recover the large $N$ structure of  $\mathrm{U}(N)^2 \times \mathrm{O}(N)$ \cite{Dartois:2013he} and $\mathrm{O}(N)^3$ \cite{Carrozza:2015adg} tensor models
\be \label{eq:free-enND}
\cF(\l) = \sum_{\om \in \frac{\mathbb{N}}{2}} N^{3-\om}  \cF_{\om}(\l)\, .
\ee
Using Eqs.\ \eqref{eq:g} and \eqref{eq:ell}, the degree can also be written as
\be \label{eq:omega2}
\omega = 3+\frac{3}{2}v-f-\varphi \,,
\ee
which does not refer to two-dimensional topology and coincides with the familiar expression found in the tensor models literature. 

\medskip

We wish to reorganize \eqref{eq:free-en} as
\be \label{eq:free-en-2}
\cF(\l) = \sum_{g\in \mathbb{N}} \left(\f{N}{\sqrt{D}}\right)^{2-2g} \sum_{\ell \in \mathbb{N}} D^{2-\frac{\ell}{2}} \cF_{g,\ell}(\l)\, ,
\ee
from which it is evident that if we keep
\be \label{eq:ratioND}
M:= \f{N}{\sqrt{D}} 
\ee
fixed as we take $N\to\infty$ and $D\to\infty$, we obtain
\be \label{eq:free-en-0}
\lim_{ \substack{N,D\to\infty \\ M<\infty} } \f{1}{D^2} \cF(\l) = \sum_{g\geq 0} M^{2-2g}  \cF_{g,0}(\l) \equiv \cF^{(0)}(M,\l)\, .
\ee
In other words, by allowing $D\neq N$, but keeping the ratio \eqref{eq:ratioND} fixed, we have a \emph{double-scaling limit} that selects Feynman graphs with $\ell=0$, but of arbitrary genus. Since such graphs are much less than all the possible graphs, they might lead to a summable series.

As usual, we are interested in determining the critical point $\l_c$, which we do not expect to depend on $g$, and the critical exponent $\g(g)$, associated to a non-analytic behavior of the free energy such as $\cF_{g,0}(\l)_{\rm crit}\sim (\l-\l_c)^{2-\g(g)}$. Determining the critical properties of the model is interesting both from the combinatorial and physical point of view: from the critical point and critical exponent we can infer the asymptotic number of graphs for large number of vertices, which is a standard objective in combinatorics \cite{flajolet2009analytic};
and from the physical point of view, the critical model determines the continuum limit of the geometrical objects dual to the Feynman graphs, as in the limit $\lambda\to\lambda_c$ the average number of vertices typically diverges \cite{DiFrancesco:1993nw}.

In principle, we could use $\cF^{(0)}(M,\l)$ as a generating function for $\cF_{g,0}(\l)$, and use the latter to define a continuum limit at fixed $g$. However, since we expect  $\l_c$ to be genus-independent, we can also find a combination of $M$ and $\l-\l_c$ to keep fixed for a \emph{triple-scaling limit}. More precisely, if $\g(g)=a+b g$, then we have
\be \label{eq:cont-lim}
\lim_{ \substack{M\to\infty \\ \l\to \l_c^- } } \f{1}{M^2 (\l-\lambda_c)^{2-a}} \cF^{(0)}(M,\l) = \sum_{g \in \mathbb{N}} \k^{2g} f_g \,,
\ee
with $\k^{-1} = M (\l-\l_c)^{b/2}$ fixed.
Since the large-$D$ limit selects for each genus $g$ a subset of diagrams, we expect that the series in Eq.~\eqref{eq:cont-lim} will have an improved convergence with respect to the usual double-scaling limit of matrix models, which is not even Borel summable \cite{DiFrancesco:1993nw}. In fact, we will see in Sec.~\ref{sec:dominantSchemes} that the triple-scaling limit leads to a series with a finite radius of convergence.

For practical reasons, it will turn out to be more convenient to work with graphs having a marked edge, which can equivalently be thought as two-point graphs, i.e.\ graphs in the perturbative expansion of the two-point function:
\be \label{eq:2pt-def}
\left\la \Tr\left[ X^\dagger_\m X_\m \right]\right\ra 
 = \frac{ \int [dX] \, e^{-S[X,X^\dagger]}\Tr\left[ X^\dagger_\m X_\m \right] }{\int [dX] \, e^{-S[X,X^\dagger]}}\,.
\ee 

\section{Recursive characterization of graphs with vanishing grade}\label{sec:recursion}

In this section, we perform a detailed combinatorial analysis of Feynman graphs with vanishing grade, culminating in the complete characterization of Proposition \ref{propo:g1} and Theorem \ref{thm:induction}. 

\subsection{Combinatorial structures}

Most of the following structures are identical to ones previously introduced in the tensor models literature, and more particularly in \cite{GurSch, Fusy:2014rba}.

\subsubsection{Melon-free graphs}\label{sec:MelonFree}

For convenience, we work with \emph{rooted} Feynman graphs, which are connected Feynman graphs with a marked edge, called the \emph{root-edge}. This choice will play an important role in some of our combinatorial constructions. The root-edge can also be thought of as a book-keeping device allowing to analyze what is really the two-point generating function, while still summing over vacuum graphs. In addition, it is also convenient to represent the root-edge by a \emph{root-vertex} of degree 2, inserted in the middle of the root-edge, with two attached oriented edges. This is of course nothing but the $\Tr\left[ X^\dagger_\m X_\m \right]$ insertion in \eqref{eq:2pt-def}.
In the the rest of the paper, we will only work with rooted Feynman graphs, even if we sometimes keep it implicit.

 We introduce a particular rooted Feynman graph, called the rooted \emph{cycle graph}, which corresponds to an oriented edge that connects the root-vertex to itself; see Figure \ref{fig:cycle-graph}. By convention, it is characterized by $v=0$, $f=2$ and $\varphi=1$. In particular, it has $g=\ell=\omega=0$.

\begin{figure}[htb]
    \centering
    \includegraphics[scale=.8]{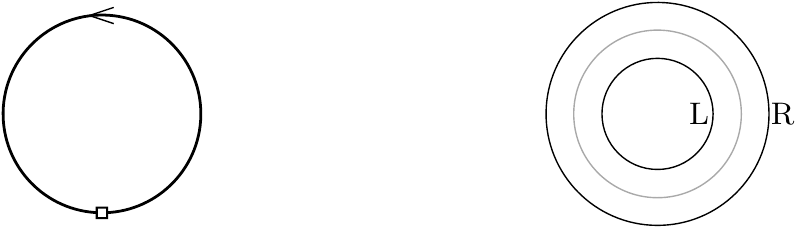}
    \caption{\small The rooted cycle graph (left) and its stranded representation (right). On the right panel, the gray line represents the $\mathrm{O}(D)$-loop or straight face, and the black lines the L- and R-faces.}
    \label{fig:cycle-graph}
\end{figure}

The class of melonic Feynman graphs plays a central role in tensor models \cite{critical, RTM}. In their rooted version, they are defined as follows \cite{GurSch, Fusy:2014rba}.

In a rooted Feynman graph, we first define an \emph{elementary melonic 2-point subgraph} as the subgraph represented on the right side of Figure \ref{fig:melon}, where the two external half-edges or legs are distinct. Note that by definition, an elementary melonic $2$-point subgraph does not contain the root-vertex. We then define a \emph{melonic insertion} on an edge $e$ of a rooted Feynman graph as the replacement of $e$ by an elementary melonic $2$-point subgraph, respecting the orientation (see Figure \ref{fig:melon}). We call the reverse operation a \emph{melonic removal}. Finally, a rooted Feynman graph is called \emph{melonic} if it can be reduced to the rooted cycle-graph by successive removals of elementary melonic $2$-point subgraphs. 

\begin{figure}[htb]
    \centering
    \includegraphics[scale=.8]{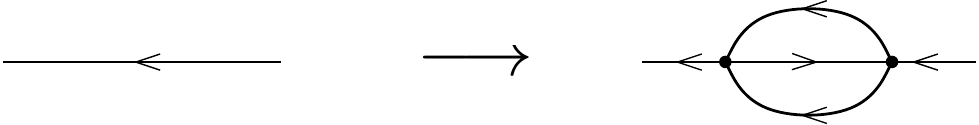}
    \caption{\small Insertion of an elementary melonic $2$-point subgraph.}
    \label{fig:melon}
\end{figure}

It is well-known that \cite{critical,Dartois:2013he}:

\begin{claim} The melonic rooted Feynman graphs are exactly the rooted Feynman graphs of degree $\omega=0$; that is, of genus $g=0$ and of grade $\ell=0$.
\end{claim}

\begin{claim} A melonic insertion or removal preserves the genus, grade and degree.
\end{claim}

\begin{claim} The generating function $T(\lambda)$ of rooted melonic Feynman graphs obeys the closed equation
\be \label{eq:melonEq}
T(\lambda)=1+\lambda^2 T(\lambda)^4 \, .
\ee
\end{claim}

 A rooted Feynman graph is \emph{melon-free} if it does not contain any elementary melonic $2$-point subgraph. By definition, the rooted cycle-graph is the only melon-free rooted Feynman graph with $\omega=g=\ell=0$. We can restrict ourselves to the study of melon-free rooted Feynman graphs because of the following result \cite{GurSch, Fusy:2014rba}.

\begin{claim}\label{cl:melon} Any rooted Feynman graph $G$ can be uniquely obtained from a melon-free rooted Feynman graph $H$, called the core of $G$, by a sequence of successive melonic insertions. In particular, the core of any melonic rooted Feynman graph is the rooted cycle-graph. 
\end{claim}

This result can also be formulated as follows. We define the insertion of a rooted Feynman graph $G_2$ on an edge $e_1$ of another rooted Feynman graph $G_1$, denoted by $(G_1,e_1)\#G_2$, as the replacement of $e_1$ by the $2$-point subgraph obtained from $G_2$ by removing its root-vertex; see Figure \ref{fig:connected_sum}. It is clear that any rooted Feynman graph $G$ can be obtained from its (unique) core $H$ by inserting melonic rooted Feynman graphs on the edges of $H$. 

\begin{figure}[htb]
    \centering
    \includegraphics[scale=.8]{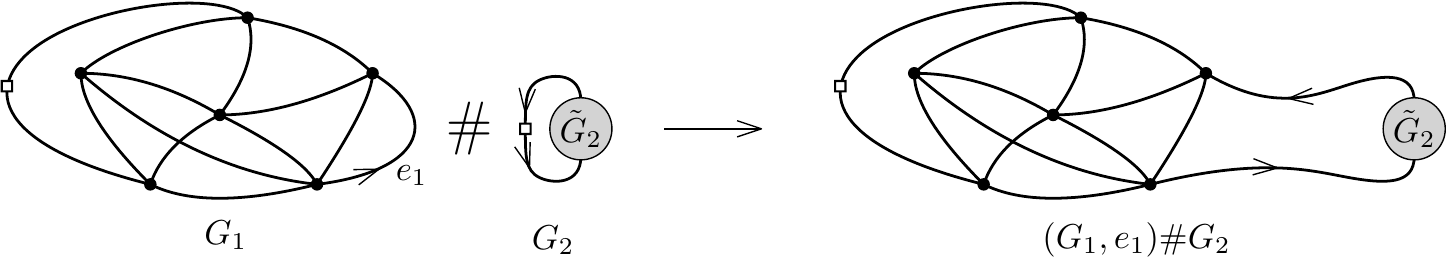}
    \caption{\small Given two rooted graphs $G_1$ and $G_2$, one may construct the graph $(G_1,e_1)\#G_2$, obtain by insertion of a two-point function on the edge $e_1$.}
    \label{fig:connected_sum} 
\end{figure}

\subsubsection{Schemes}\label{sec:schemes}

We define a \emph{dipole} as a two-edge subgraph on two vertices (excluding the root-vertex) which contains a face of length two. There are three types of dipoles (see Figure \ref{fig:03dipoles}):
\begin{itemize}
\item non-planar dipoles or N-dipoles, which contain a length-two straight face;
\item left dipoles or L-dipoles, which contain a length-two L-face;
\item right dipoles or R-dipoles, which contain a length-two R-face.
\end{itemize}

Note that a N-dipole has the topology of a cylinder\footnote{The letter "N" stands for "non-planar", and refers to the fact that the length-two face of a N-dipole may be non-contractible.}, whereas an L- or R-dipoles has the topology of a disk. This is illustrated in Figure \ref{fig:N-dipole_topology}.

The four half-edges incident to a dipole, i.e.\ its external legs, can be canonically partitioned into two pairs as depicted on Figure \ref{fig:03dipoles}, where one pair is represented on the left side of the dipole and the other pair on the right side. This distinction will be useful to compose dipoles into ladders, which we now define.
\begin{figure}
    \centering
    \includegraphics[scale=.8]{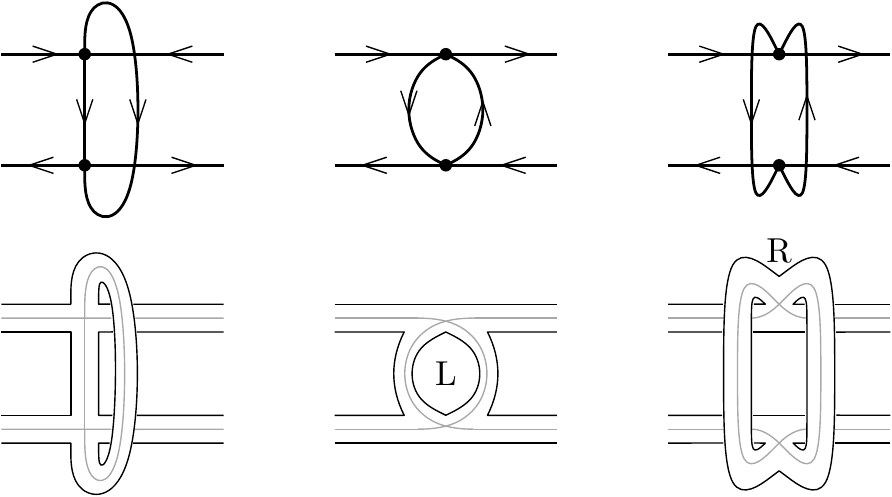}
    \caption{\small From left to right: a N-dipole, a L-dipole and a R-dipole. The Feynman graphs are shown in the upper part, their associated stranded structures are illustrated in the lower part. A L-dipole (resp.\ R-dipole) is defined as containing an internal L-face (resp.\ R-face).}
    \label{fig:03dipoles}
\end{figure}

\begin{figure}[htb]
    \centering
    \includegraphics[scale=.9]{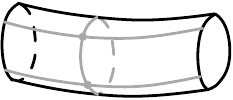}
    \caption{\small Topological structure of a N-dipole: the dipole is represented in grey, and the manifold it is canonically embedded into in black.}
    \label{fig:N-dipole_topology} 
\end{figure}

A \emph{ladder}\footnote{Our notion of ladder is equivalent to the notion of \emph{proper chain} in \cite{Fusy:2014rba}.} is a sequence of $n\geq2$ dipoles $(d_1, ..., d_n)$ such that two consecutive dipoles $d_i$ and $d_{i+1}$ are connected by two edges that involve two half-edges on the same side of $d_i$ and two half-edges on the same side of $d_{i+1}$. The dipoles in a ladder will be referred to as \emph{rungs}. Furthermore, the edges that connect them split into two \emph{rails}. Note that by definition, the root-vertex cannot appear on a rung or on a rail of a ladder.

A ladder is said to be \emph{broken}, or a B-ladder, if it contains rungs of different types; it is \emph{unbroken} otherwise. Accordingly, there are three types of unbroken ladders: N-, L- and R-ladders, which respectively only contain rungs of type N, L and R. To keep track of the external face structure of ladders, we furthermore need to split N-ladders into two subfamilies: $\Ne$-ladders which have an even number of rungs, and $\No$-ladders which have an odd number of rungs. This is summarized in Figure \ref{fig:0chainvertexes}.

A ladder is \emph{maximal} if it is not included into a larger ladder. We have the following property \cite{GurSch, Fusy:2014rba}.

\begin{claim}\label{claim:vertex-disjoint} Any ladder can be uniquely extended into a maximal ladder. In addition, two distinct maximal ladders are vertex-disjoint.\footnote{Note that for these statements to be true, it is important that Feynman graphs are rooted and that ladders contain at least two rungs, as it is the case with our conventions.}
\end{claim}

In the following, we will replace maximal ladders by new 4-point vertices which we call \emph{ladder-vertices}.\footnote{They are called chain-vertices in \cite{Fusy:2014rba}.} There are five types of ladder-vertices; we call them N${}_e$-, N${}_o$-, L-, R- and B-vertices. Whenever we do not need to distinguish N${}_e$- and N${}_o$-vertices, we simply call them N-vertices. This is illustrated in Figure \ref{fig:0chainvertexes}. 

\medskip

\begin{remark}
The external legs on the two sides of each ladder-vertex are fixed by convention as in Figure \ref{fig:0chainvertexes}. In the case of a B-vertex, this may require `twisting' the two rails at one end of the corresponding B-ladder. 
\end{remark}

In the rest of the paper, we will rely heavily on \emph{schemes}, which characterize equivalence classes of Feynman graphs, defined up to melon and ladder insertions. 
\begin{figure}[htb]
    \centering
    \includegraphics[scale=.8]{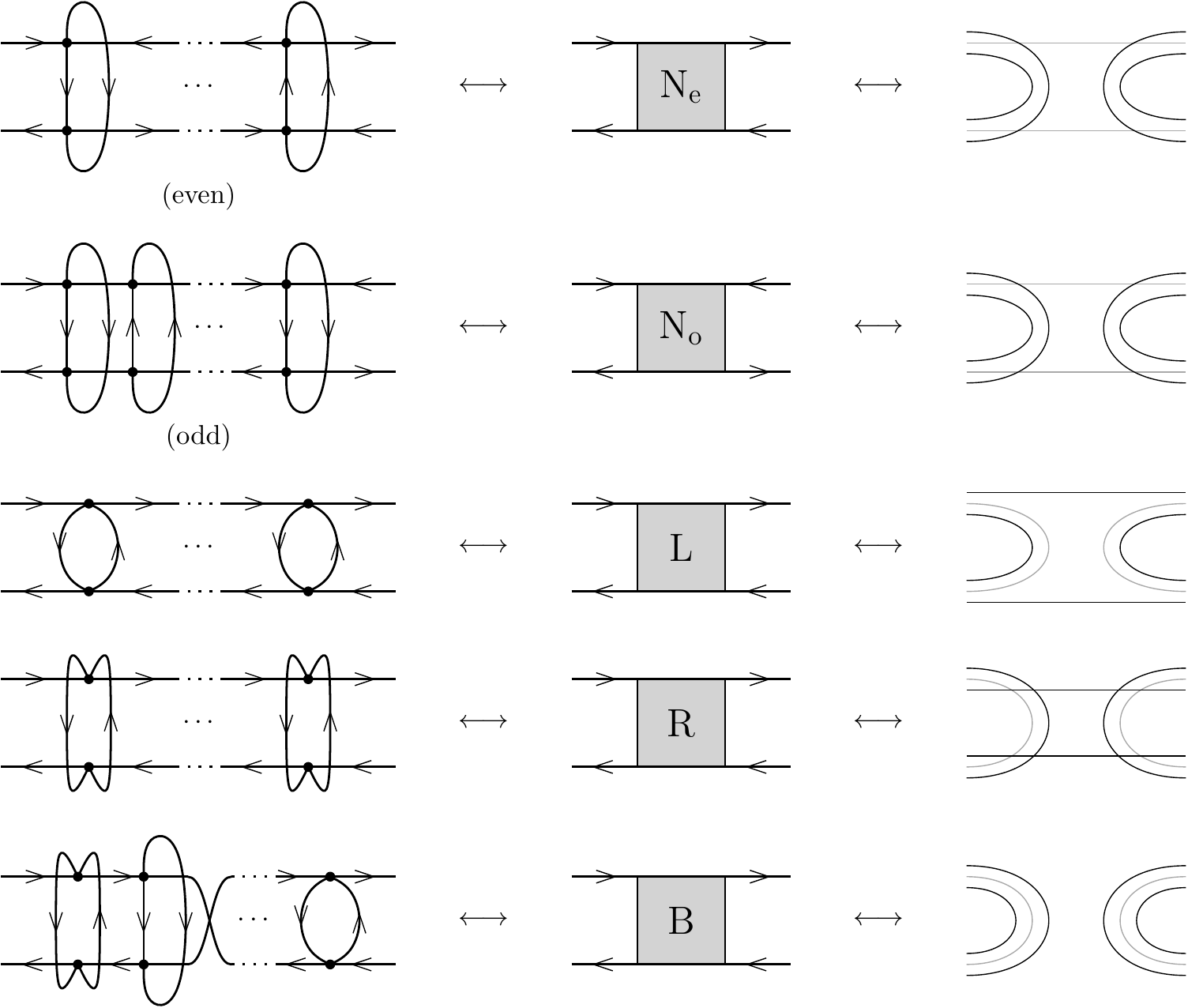}
    \caption{\small Maximal ladders and their associated ladder-vertices. For any type of ladder-vertex (center), we have represented: the parent maximal ladders they represent (left panel), and the structure of their external faces (right panel).}
    \label{fig:0chainvertexes}
\end{figure}

\begin{definition}
Let $G$ be a connected, melon-free and rooted Feynman graph. The \emph{scheme} $S_G$ of $G$ is the graph obtained by replacing any maximal ladder by the ladder-vertex of the corresponding type.
\end{definition}

More generally, we will sometimes consider the larger family of \emph{Feynman graphs with ladder-vertices}, which consists of \emph{all} connected and rooted graphs built out of edges, standard vertices and ladder-vertices (see Figure \ref{fig:0vertexes}). In the class of Feynman graphs with ladder-vertices, an elementary melonic $2$-point subgraph is defined in the same way as for rooted Feynman graphs. Furthermore, the definition of ladder is extended so that rungs may correspond to dipoles or ladder-vertices. By construction, a scheme is a Feynman graph with ladder-vertices. However, the converse is not always true, as made explicit in the following remark.

\begin{figure}[htb]
    \centering
    \includegraphics[scale=.8]{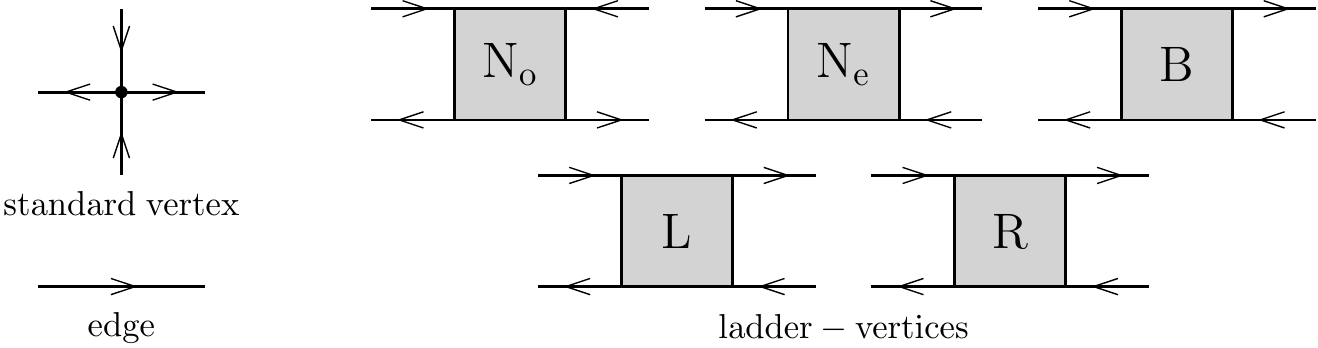}
    \caption{\small Feynman graphs with ladder-vertices and schemes are made out of: edges, standard vertices and ladder-vertices.}
    \label{fig:0vertexes}
\end{figure}

\medskip

\begin{remark}
 A scheme is a Feynman graph with ladder-vertices which \emph{cannot} contain any of the following subgraphs:
\begin{enumerate} 
\item an elementary melonic $2$-point subgraph;
\item an edge that connects the two external legs on the same side of a ladder-vertex;
\item a ladder.
\end{enumerate}
A Feynman graph with ladder-vertices is said to be \emph{melon-free} (resp.\ \emph{ladder-free}) if it obeys the first two (resp.\ the third) conditions.\footnote{A Feynman graph with ladder-vertices that obeys the three conditions is called a \emph{reduced scheme} in \cite{Fusy:2014rba}.} Therefore, a scheme is a melon-free and ladder-free Feynman graph with ladder-vertices. This is illustrated in Figure \ref{fig:schemeNotAllowed}. Remark on the other hand that the subgraph of Figure \ref{fig:schemeAllowed} is allowed in a scheme.
\end{remark}

\medskip

\begin{figure}[htb]
    \centering
    \includegraphics[scale=.8]{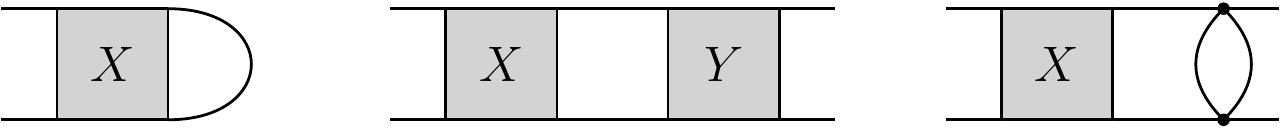}
    \caption{\small Examples of subgraphs which cannot be realized in a scheme, where $X$ and $Y$ are arbitrary ladder-vertices and the edge orientations are left implicit.}
    \label{fig:schemeNotAllowed}
\end{figure}

\begin{figure}[htb]
    \centering
    \includegraphics[scale=.8]{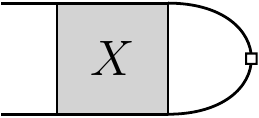}
    \caption{\small A subgraph which may be found in melon-free Feynman graphs with ladder-vertices, where $X$ is a ladder-vertex of any type. Note the crucial presence of the root-vertex.}
    \label{fig:schemeAllowed}
\end{figure}

It is possible and convenient to extend the definitions of genus, degree and grade to the class of Feynman graphs with ladder-vertices. Given $G$ a Feynman graph with ladder-vertices, one can construct an ordinary Feynman graph by replacing each ladder-vertex in $G$ with a ladder of the appropriate type (for instance, a N${}_\mathrm{o}$-vertex is replaced by a N${}_\mathrm{o}$-ladder, etc.). One can then show 
that the genus, grade and degree of the resulting Feynman graph does not depend on the replacement details, that is, on the length of the new ladders or on the structure of the new B-ladders. 
This provides a consistent prescription for the genus, grade or degree of a Feynman graph with ladder-vertices, which is useful to compute these quantities in practice.
However, to shortcut the details of this construction, we will simply rely on the following result \cite{Fusy:2014rba}.

\begin{claim} Let $G_1$ and $G_2$ be two connected, melon-free and rooted Feynman graphs. If $S_{G_1} = S_{G_2}$, then:
\be
g(G_1) = g(G_2) \,, \qquad \omega(G_1) = \omega(G_2) \qquad \mathrm{and} \qquad \ell(G_1) = \ell(G_2)\,. 
\ee
\end{claim}
\noindent In other words, the genus, degree and grade are constant on any equivalence class of graphs defined by a scheme. This consistently extends the definition of these three quantities to schemes, and in a second step, to Feynman graphs with ladder-vertices. Indeed, it is easy to see that a Feynman graph with ladder vertices $G$ can itself be mapped to a unique scheme $S_G$ (obtained by consistent replacement of melon two-point functions by propagators, and ladders by their corresponding ladder-vertices), which allows to define: $g(G) := g(S_G)$, $\omega(G) := \omega(S_G)$ and $\ell(G) := \ell(S_G)$.

Finally, it will be useful to distinguish between two types of Feynman graphs with ladder-vertices\footnote{We recall that the family of Feynman graphs with ladder-vertices contains the family of connected and rooted Feynman graphs.}: whether they are \emph{two-particle reducible} (2PR) or \emph{two-particle irreducible} (2PI). We say that a Feynman graph with ladder-vertices $G$ is 2PR if it contains a \emph{two-edge-cut}, that is, a pair $(e,e')$ of edges in $G$ whose removal disconnects $G$ and such that $e$ and $e'$ are not both incident to the root-vertex. Otherwise, we say that $G$ is 2PI. 

We prove the following result for connected, melon-free and rooted Feynman graphs:

\begin{lemma}\label{lem:schemes2Pi}
Let $G$ be a connected, melon-free and rooted Feynman graph. $G$ is 2PI if and only if $S_G$ is 2PI.
\end{lemma}
\begin{proof}
The following implication is immediate: if $S_G$ is 2PR, then $G$ is also 2PR. Conversely, let us assume that $G$ is 2PR. We can then find a pair $(e,e')$ of edges in $G$ forming a two-edge-cut. We distinguish two subcases.
\begin{itemize}

\item[i)] If neither $e$ nor $e'$ are contained in maximal ladders in $G$, then both are realized as edges in $S_G$. They therefore constitute a two-edge-cut in $S_G$.

\item[ii)] If $e$ (or, equivalently, $e'$) is contained in a maximal ladder $\mathcal{L} \in G$, it is then straightforward to realize that: $e$ must lie on one of the two rails of the ladder; and furthermore, $e'$ must lie on the opposite rail. More precisely, $e$ and $e'$ must connect the same two rungs, as shown on the left part of Figure \ref{fig:02PRproof}. Let us call $e_R$ and $e_R'$ (resp. $e_L$ and $e_L'$) the two external legs on the right (resp. left) side of $\mathcal{L}$. The fact that $(e,e')$ is a two-edge-cut in $G$ implies that: in $S_G$, both $(e_R , e_R')$ and $(e_L , e_L')$ are connected by $2$-point subgraphs, as illustrated on the right part of Figure \ref{fig:02PRproof}. At least one of the latter is non-empty (by the melon-free condition) and does not reduce to the root; therefore, at least one of the two pairs of edges constitutes a two-edge-cut in $S_G$.  
\end{itemize}
This concludes the proof.
\end{proof}

\begin{remark}
More generally, the 2PR/2PI property is transitive under replacement of a maximal ladder by a ladder-vertex in the class of melon-free rooted Feynman graphs with ladder-vertices.
\end{remark}

\begin{figure}[htb]
\centering
\includegraphics[scale=.7]{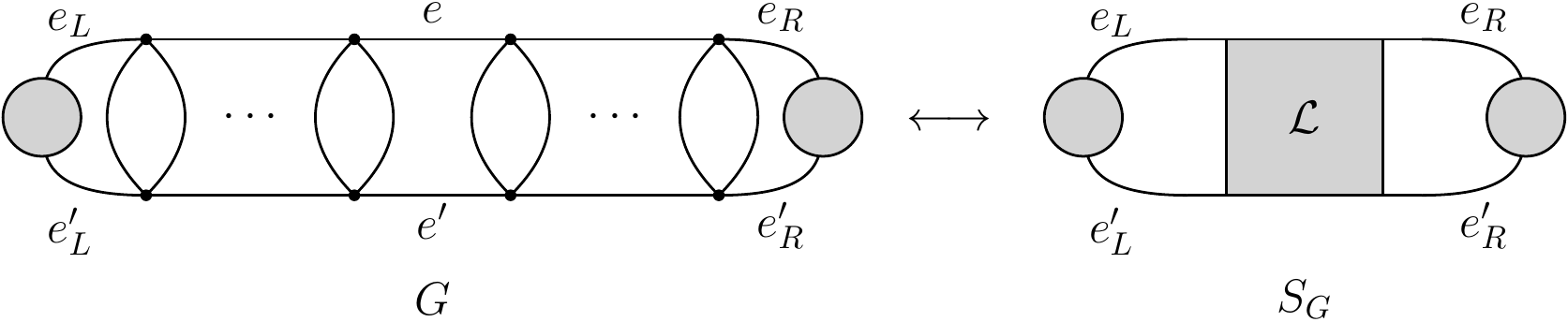}
\caption{\small A two-edge-cut $(e,e')$ inside a maximal ladder of $G$ necessarily translates into a two-edge-cut $(e_L ,e_L')$ or $(e_R ,e_R')$ in the scheme $S_G$, where $2$-point subgraphs are represented as shaded disks.}
\label{fig:02PRproof}
\end{figure}

\subsubsection{Combinatorial moves on Feynman graphs with ladder-vertices}\label{sec:CombMoves}

We now introduce a set of local operations on Feynman graphs with ladder-vertices and we study their effect on the genus, the grade and the degree of these graphs. 

We define a dipole or a ladder-vertex \emph{contraction} as the operation which consists in: 1) removing the dipole or the ladder-vertex, and 2) reconnecting the two half-edges on each side of the dipole or the ladder-vertex. The reverse operation is called a dipole or ladder-vertex \emph{insertion}. This is illustrated in Figure \ref{fig:0movesdipole}.

\begin{figure}[htb]
\centering
\includegraphics[scale=.8]{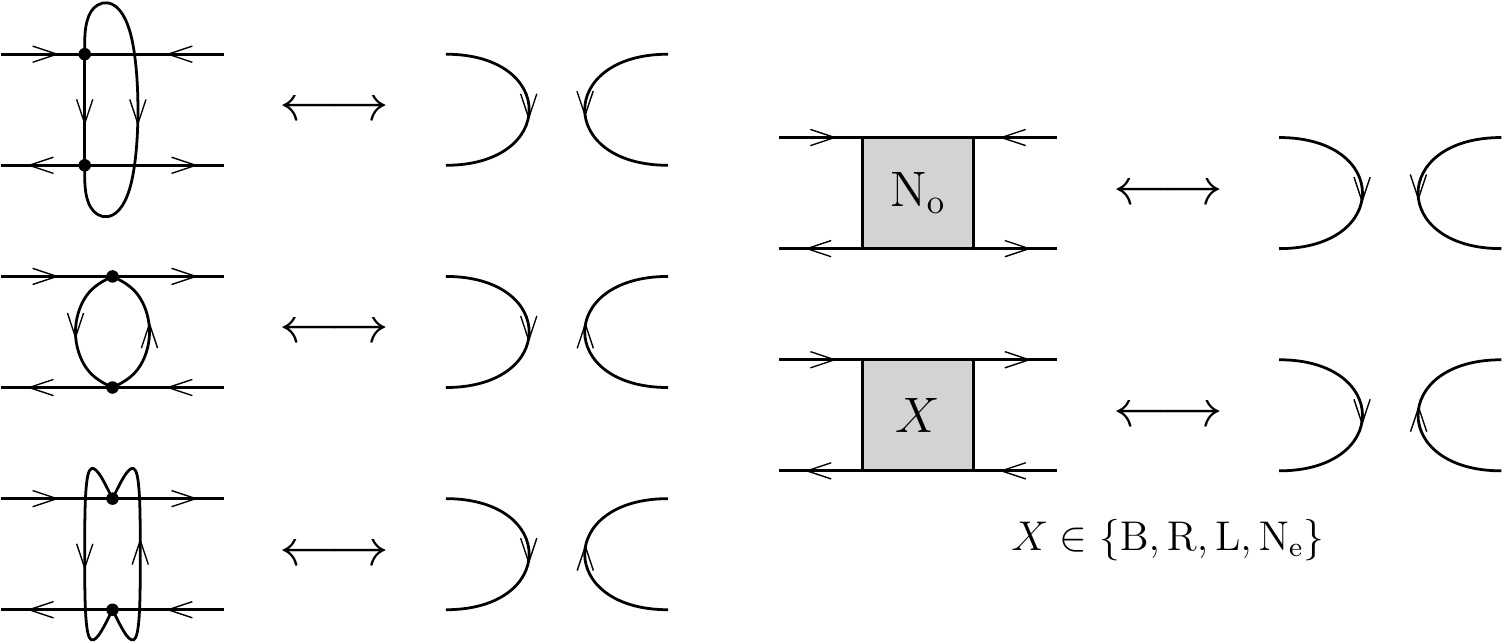}
\caption{\small Contraction/insertion of dipoles (left panel) and ladder-vertices (right panel).}
\label{fig:0movesdipole}
\end{figure}

Note that the contraction of a dipole or a ladder-vertex may disconnect the graph. In addition, it is easy to check that the number of connected components can increase by at most one. A dipole or a ladder-vertex is said to be \emph{separating} if its contraction increases the number of connected components by one, and \emph{non-separating} otherwise. 

\medskip

\begin{remark}
A separating (resp. non-separating) N-dipole is also separating  (resp.\ non-separating) in the topological sense of the term: the cycle it constitutes separates (resp.\ fails to separate) the discretized Riemann surface encoded in the Feynman graph into two disconnected regions. 
\end{remark} 

\medskip

For later convenience, we want to work with connected and rooted graphs. The contraction of a non-separating dipole or ladder-vertex yields a graph in this class. However, we need a prescription so that the contraction of a separating dipole or ladder-vertex yields two connected and rooted graphs instead of one rooted graph with two connected components. Let $G$ be a Feynman graphs with ladder-vertices and suppose that there exists a separating dipole or ladder-vertex in $G$. As a result, $G$ necessarily has a structure similar to that depicted in the left panel of Figure \ref{fig:0separating} (up to a choice of dipole or ladder-vertex type), where $\tilde{G}_1$ and $\tilde{G}_2$ are two connected $2$-point subgraphs. Since $G$ is rooted, the root-vertex is necessarily contained in $\tilde{G}_1$ or $\tilde{G}_2$. Note that it may be adjacent to the separating dipole or ladder-vertex. Suppose that it is contained in $\tilde{G}_1$, as in Figure \ref{fig:0separating} (the case of $\tilde{G}_2$ is similar). Then, when we contract the separating dipole or ladder-vertex, we add a root-vertex in the middle of the edge that reconnects the two half-edges of $\tilde{G}_2$. As a result, the contraction gives rise to two Feynman graphs with ladder-vertices $G_1$ and $G_2$, as depicted in Figure \ref{fig:0separating}, which are both connected and rooted. 

\medskip

\begin{remark}
The above prescription implies that the insertion of a separating dipole or ladder-vertex in between two Feynman graphs with ladder-vertices $G_1$ and $G_2$ must involve the root-vertex of either $G_1$ or $G_2$. This specification is left implicit when it does not play an important role. Otherwise, we specify that the insertion is performed with respect to the root-vertex of $G_1$ or $G_2$. 
\end{remark} 

\medskip

\begin{figure}[htb]
    \centering
    \includegraphics[scale=.8]{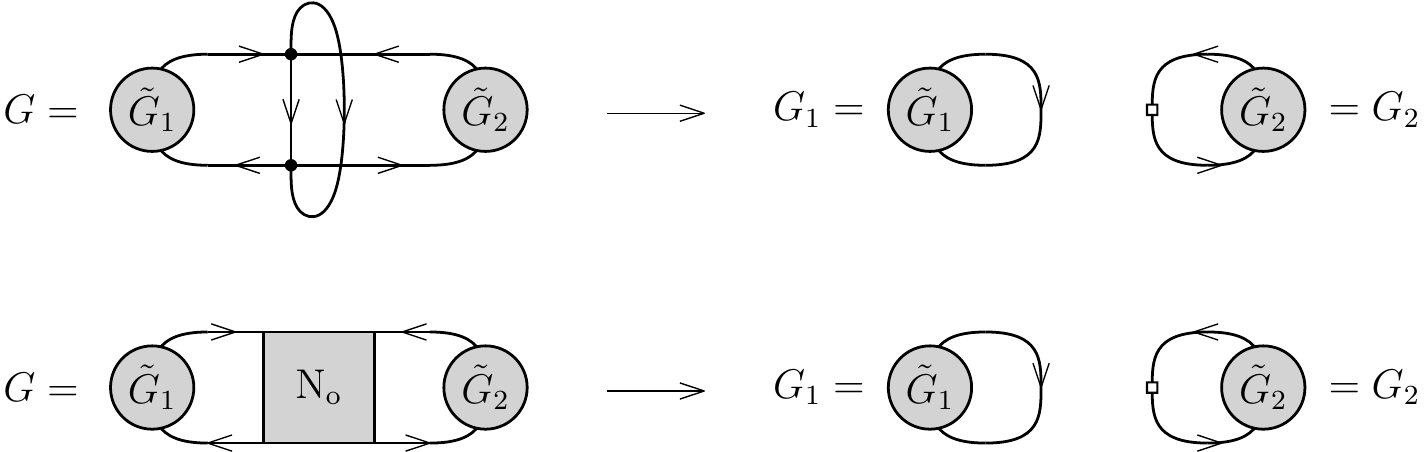}
  \caption{\small Top: a separating dipole and its contraction. Bottom: a separating ladder-vertex and its contraction. In both cases, we have assumed that the root-vertex of $G$ is contained in $\tilde G_1$. In the top figure, the N-dipole also constitutes a separating cycle on the discretized Riemann surface represented by the Feynman graph. Remark here that the opposite operation is referred to as an insertion.}
\label{fig:0separating}
\end{figure}

Let $G$ be a Feynman graph with ladder-vertices and suppose that it contains a dipole. If this dipole is separating, we denote by $G_1$ and $G_2$ the two Feynman graphs with ladder-vertices obtained after contracting this dipole; while if it is non-separating, we denote by $G'$ the resulting Feynman graph with ladder-vertices. Clearly, we have in both cases $v(G_1)+v(G_2)=v(G)-2$ and $v(G')=v(G)-2$, respectively. In order to study the effect of the contraction on the genus, the grade and the degree, one needs to analyze how the total number of faces and $\oD$-loops is affected. We have the following cases:
\begin{itemize}
    \item \textit{Separating N-, L-, or R-dipole}. The contraction deletes one internal $\oD$-loop in the case of a N-dipole and one internal face in the cases of L- and R-dipoles. Furthermore, in the case of a N-dipole, the number of external\footnote{As is customary in the literature, a face is said to be \emph{internal} to a given subgraph (such as e.g. a dipole) if it has only supports on edges of the said subgraph, and \emph{external} otherwise.} faces is unaffected (the external face structure is the same before and after the contraction) and there is one additional external $\oD$-loop which is created due to the separating nature of the dipole. In the case of a L-dipole (resp.\ R-dipole), it is the number of external $\oD$-loops and R-faces (resp.\ L-faces) that  are unaffected while there is one additional external L-face (resp\ R-face) being created. As a result, in the three cases, we have $f(G_1)+f(G_2)=f(G)$ and $\varphi(G_1)+\varphi(G_2)=\varphi(G)$. Using Eqs.\ \eqref{eq:g}, \eqref{eq:ell} and \eqref{eq:omega2}, we thus obtain
    \be \label{eq:contSep}
    g(G_1)+g(G_2)=g(G)\,, \quad \ell(G_1)+\ell(G_2)=\ell(G) \quad \mathrm{and} \quad \omega(G_1)+\omega(G_2)=\omega(G) \, ;
    \ee
    
    \item \textit{Non-separating N-dipole}. The contraction removes one internal $\oD$-loop and the number of external faces is unaffected. Besides, there is one additional external $\oD$-loop which is either created or deleted. Hence, $\varphi(G')=\varphi(G)-1+\sigma$ with $\sigma=\pm1$, and we obtain
    \be \label{eq:contNonSepN}
    g(G')=g(G)-1\,, \quad \ell(G')=\ell(G)-2(\sigma+1) \quad \mathrm{and} \quad \omega(G')=\omega(G)-(\sigma+2) \, ;
    \ee
    
    \item \textit{Non-separating L- or R-dipole}. The contraction removes one internal L- or R-face while the number of external $\oD$-loops and R- or L-faces remains the same, respectively. As for the external L- or R-faces, there is an additional one which is either created or deleted; therefore, in both cases, $f(G')=f(G)-1+\sigma$ with $\sigma=\pm1$. As a result, we have
    \be \label{eq:contNonSepLR}
    g(G')=g(G)-\frac{1}{2}(\sigma+1)\,, \quad \ell(G')=\ell(G)-(\sigma+3) \quad \mathrm{and} \quad \omega(G')=\omega(G)-(\sigma+2) \, .
    \ee
    
\end{itemize}

We now study the effect of contracting a ladder-vertex in a Feynman graph with ladder-vertices $G$. We still denote by $G_1$ and $G_2$ (resp.\ $G'$) the two graphs (resp.\ the graph) obtained after contracting a separating (resp.\ non-separating) ladder-vertex in $G$. As explained in Section \ref{sec:schemes}, one way to analyze how the genus, grade and degree change when contracting a ladder-vertex is to first replace it with a ladder of dipoles of the corresponding type. This ladder can be of arbitrarily length and arbitrary structure in the case of a B-vertex. Then, it is straightforward to see that the contraction of the ladder-vertex is equivalent to the combination of two moves: 1) the contraction of one dipole in the corresponding ladder, whose analysis has been given above; and 2) the removal of up to two melonic $2$-point subgraphs, which may have been generated by the first move due to the presence of other dipoles in the initial ladder. Since the second step preserves the genus, degree and grade (see Section \ref{sec:MelonFree}), it requires no further discussion. We distinguish the following cases: 
\begin{itemize}
    \item \textit{Separating B-, N-, L- or R-vertex}. We first replace the ladder-vertex with a B-, N-, L- or R-ladder, respectively. Because of the separating nature of the ladder-vertex, the dipoles in the corresponding ladder are necessarily separating as well. Hence, we can use the result of Eq.\ \eqref{eq:contSep} for the contraction of a separating dipole. Furthermore, since melonic removals do not change the genus, grade and degree of a Feynman graph, we obtain
    \be \label{eq:LVcontSep}
    g(G_1)+g(G_2)=g(G)\,, \quad \ell(G_1)+\ell(G_2)=\ell(G) \quad \mathrm{and} \quad \omega(G_1)+\omega(G_2)=\omega(G) \, ;
    \ee
    
    \item \textit{Non-separating N-vertex}. We replace the N-vertex with a N-ladder, made out of non-separating N-dipoles. Using the result of Eq.\ \eqref{eq:contNonSepN} and the properties of melonic removal, we have
    \be \label{eq:LVcontNonSepN}
    g(G')=g(G)-1\,, \quad \ell(G')=\ell(G)-2(\sigma+1) \quad \mathrm{and} \quad \omega(G')=\omega(G)-(\sigma+2) \, ,
    \ee
    with $\sigma=\pm1$;
    
    \item \textit{Non-separating L- or R-vertex}. The same reasoning as in the previous case leads, using Eq.\ \eqref{eq:contNonSepLR}, to
    \be \label{eq:LVcontNonSepLR}
    g(G')=g(G)-\frac{1}{2}(\sigma+1)\,, \quad \ell(G')=\ell(G)-(\sigma+3) \quad \mathrm{and} \quad \omega(G')=\omega(G)-(\sigma+2) \, ,
    \ee
    with $\sigma=\pm1$;
    
    \item \textit{Non-separating B-vertex}. We can replace the B-vertex with a B-ladder of arbitrary length and structure, all of its dipoles being necessarily non-separating. Furthermore, because of the structure of a B-ladder, one can check that the contraction of a non-separating N-dipole always yields one additional external $\oD$-loop (case of Eq.\ \eqref{eq:contNonSepN} with $\sigma=+1$); and the contraction of a non-separating L- or R-dipole leads to one additional external face (case of Eq.\ \eqref{eq:contNonSepLR} with $\sigma=+1$). Hence, we obtain in this final case
    \be \label{eq:LVcontNonSepB}
    g(G')=g(G)-1\,, \quad \ell(G')=\ell(G)-4 \quad \mathrm{and} \quad \omega(G')=\omega(G)-3 \, .
    \ee
    
\end{itemize}

\medskip

\begin{remark}
The above discussion on the contraction of a ladder-vertex in a Feynman graph with ladder-vertices naturally extends to the contraction of a ladder, not necessarily maximal, in a connected and rooted Feynman graph.
\end{remark} 

\medskip

Finally, we introduce yet another local operation that will be useful in order to analyze 2PR Feynman graphs, in the same spirit as in \cite{Bonzom:2019yik}. Let $G$ be a Feynman graph with ladder-vertices and suppose that it is 2PR. $G$ therefore contains a two-edge-cut $(e,e')$ and has the structure depicted on the left of Figure \ref{fig:flipOp}, where $\tilde G_1$ and $\tilde G_2$ are, by definition, two connected $2$-point subgraphs that are non-empty and distinct from the root-vertex alone. We define a \emph{flip} on $(e,e')$ as the operation that consists in: 1) cutting the two edges $e$ and $e'$, and 2) reconnecting the four half-edges two by two, as shown on the right side of Figure \ref{fig:flipOp}. A flip necessarily disconnects $G$ into two connected components. As we discussed for separating dipoles and ladder-vertices, we add a root-vertex in the middle of the reconnected edge on the side of $\tilde G_2$ (resp.\ $\tilde G_1$) if the root-vertex of $G$ is contained within $\tilde G_1$ (resp.\ $\tilde G_2$). By doing so, a flip generates two Feynman graphs with ladder-vertices $G_1$ and $G_2$, as indicated in Figure \ref{fig:flipOp}. Besides, we call the reverse operation of a flip a \emph{two-edge-connection insertion}, which must be performed with respect to the root-vertex of either $G_1$ or $G_2$.

\begin{figure}[htb]
\centering
\includegraphics[scale=1.2]{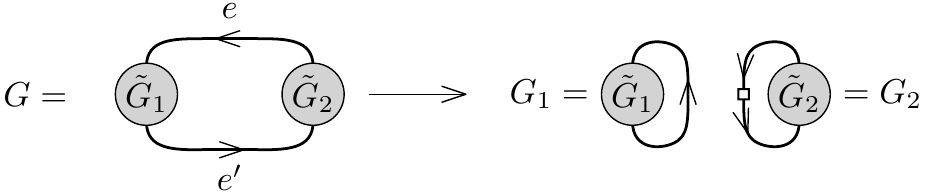}
\caption{\small Flip operation on a 2PR Feynman graph with ladder-vertices, where we assumed that the root-vertex of $G$ is contained in $\tilde G_1$.}
\label{fig:flipOp}
\end{figure}

The following result summarizes the effect of a flip operation on the genus, grade and degree of the graphs:
\begin{lemma}\label{lem:2PR} Let $G$ be a Feynman graph with ladder-vertices and suppose that it contains a two-edge-cut $(e,e')$. Then, the flip operation on $(e,e')$ generates two Feynman graphs with ladder-vertices $G_1$ and $G_2$ satisfying
\be \label{eq:flipOp}
g(G_1)+g(G_2)=g(G)\,, \quad \ell(G_1)+\ell(G_2)=\ell(G) \quad \mathrm{and} \quad \omega(G_1)+\omega(G_2)=\omega(G) \, .
\ee
\end{lemma}
\begin{proof} The result follows from a direct inspection of the flip operation in Figure \ref{fig:flipOp}. Indeed, the total number of vertices remains the same whereas the total number of faces increases by two (one additional L-face and one additional R-face) and the total number of $\oD$-loops increases by one. Eq.\ \eqref{eq:flipOp} then follows from Eqs.\ \eqref{eq:g}, \eqref{eq:ell} and \eqref{eq:omega2} applied to $G$, $G_1$ and $G_2$.
\end{proof}


\subsection{Melon-free Feynman graphs with vanishing grade}
\label{sec:FGVanishingGrade}

We want to characterize the connected, melon-free and rooted Feynman graphs with $\ell=0$ and the corresponding schemes. A first trivial observation is that they always have an even number of (standard) vertices, as made explicit in Eq.\ \eqref{eq:ell2}. The following two lemmas provide a more detailed characterization.

\begin{lemma}\label{lem:N_dip}
Let $G$ be a connected, melon-free and rooted Feynman graph and $S_G$ its scheme. If $\ell(G) = 0$, then there exists a N-dipole in $G$. 
In particular, there exists a N-dipole, a N-vertex or a B-vertex in $S_G$.
\end{lemma}
\begin{proof}
It is convenient to decompose the $\oD$-loops in $G$ with respect to their length, that is, the number of vertices they pass through, counted with multiplicity. Due to the structure of the Feynman graphs, it is straightforward to see that any $\oD$-loop has even length. We thus denote by $\varphi_{2n}$ the number of $\oD$-loops of length $2n$ ($n\in\mathbb{N}^*$) in $G$ and we write 
\be
\varphi=\sum_{n\in\mathbb{N}^*} \varphi_{2n} \,.
\ee
In addition, due to the $\oD$-structure of the Feynman graph vertices, we have the following constraint:
\be
\sum_{n\in\mathbb{N}^*} 2n\varphi_{2n}=2v \,.
\ee
Plugging these two equations in the expression \eqref{eq:ell2} for the grade, we deduce
\be
\ell(G) = 2 + 2g(G) + \frac{1}{2} \sum_{n \in \mathbb{N}^*} (2 n - 4) \varphi_{2n} \geq 2 + 2g(G) - \varphi_2\, .
\ee
By assumption, $\ell(G) = 0$ so that $\varphi_2 \geq 2g(G) + 2 > 0$, which means that there is at least one $\oD$-loop of length two in $G$. This $\oD$-loop is necessarily contained within a N-dipole in $G$. If this dipole is in a maximal ladder of $G$, the latter translates into a N-vertex or a B-vertex in $S_G$.
\end{proof}

We call \emph{connecting} a N-dipole or a N-vertex which `connects' two distinct external $\oD$-loops, that merge into a single $\oD$-loop upon contraction. 
More precisely, we have one of the subgraphs represented in Figure \ref{fig:connecting}, with two distinct $\oD$-loops running through the two rails of the N-dipole or N-vertex.

\begin{figure}[htb]
    \centering
    \includegraphics[scale=.8]{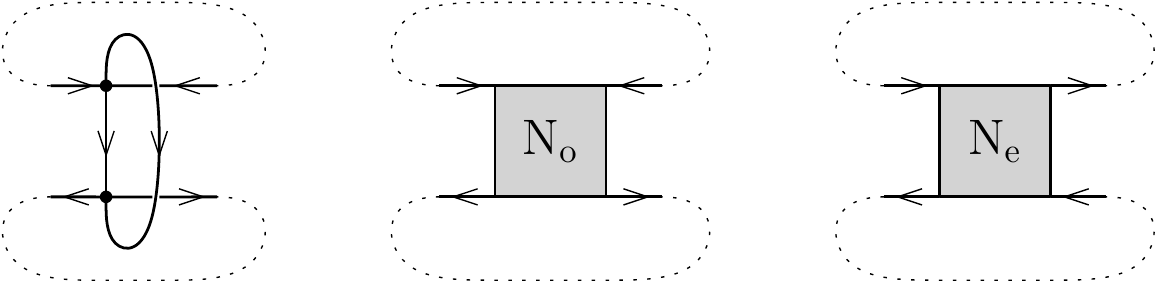}
    \caption{\small Connecting N-dipole and N-vertices. The dotted lines illustrate the external structure of the $\oD$-loops: two distinct $\oD$-loops run through the rails of the dipole or the ladder-vertex. Upon contraction of such a connecting N-dipole or N-vertex, the two distinct $\oD$-loops merge into a single $\oD$-loop.}
    \label{fig:connecting}
\end{figure}

\medskip

\begin{remark}A connecting N-dipole or N-vertex is automatically non-separating.
\end{remark}

\medskip

\begin{lemma}\label{lem:sep-conn}
Let $G$ be a connected, melon-free and rooted Feynman graph and $S_G$ its scheme. If $\ell(G) = 0$, any N-dipole in $G$ is separating or connecting; any other dipole is separating.  
Furthermore, in $S_G$: any N-dipole or N-vertex is separating or connecting; any other dipole or ladder-vertex is separating.
\end{lemma}
\begin{proof}
This follows from the variation of the grade under a dipole or a ladder-vertex contraction (see Section \ref{sec:CombMoves}) and from the non-negativity of the grade. Consider a N-dipole in $G$ and suppose it is non-separating. Performing a dipole contraction yields another connected and rooted Feynman graph $G'$ such that, by Eq.\ \eqref{eq:contNonSepN}, $\ell(G')=\ell(G)-2(\sigma+1)$ with $\sigma=\pm1$. Since $\ell(G)=0$ and $\ell\geq0$, it implies $\sigma=-1$. One can check that this can only occur if the N-dipole is connecting. 

Consider now a L- or R-dipole in $G$ and suppose it is non-separating. By Eq.\ \eqref{eq:contNonSepLR}, contracting this dipole yields another connected and rooted Feynman graph $G'$ such that $\ell(G')<\ell(G)$, which is impossible because $\ell(G) = 0$ and $\ell \geq 0$.

The same reasoning applies for any dipole or ladder-vertex in $S_G$ using a dipole or a ladder-vertex contraction and Eqs.\ \eqref{eq:LVcontNonSepN}-\eqref{eq:LVcontNonSepB}.
\end{proof}

At this point, it is clear that we can manipulate $\ell = 0$ connected, melon-free and rooted Feynman graphs\footnote{To be more succinct, we assume in the following that the Feynman graphs are always connected and rooted, unless otherwise stated.} and their schemes by successive insertions or contractions of: 1) connecting N-dipoles or N-vertices; and 2) separating dipoles or ladder-vertices. We will in fact prove that theses graphs can be generated inductively at arbitrary genus, and in an entirely constructive manner, from the ones with genus one.\footnote{Genus one is peculiar because our Feynman graphs are rooted. Had we worked with non-rooted Feynman graphs, we could have started our inductive construction at genus $0$. Treating $g=1$ separately (in Proposition \ref{propo:g1}) will also allow us to introduce the ingredients of the general proof of Theorem \ref{thm:induction} in a progressive manner.}

To prepare the ground for this construction, it is convenient to enumerate the situations in which the contraction of a dipole or a ladder in a $\ell=0$ melon-free Feynman graph generates a melonic $2$-point subgraph. Let $G$ be such a graph and suppose that it contains a dipole or ladder $X$. If we assume that $X$ is non-separating, it is necessarily a connecting N-dipole or N-ladder by Lemma \ref{lem:sep-conn}. Contracting $X$ yields another $\ell=0$ Feynman graph $G'$, which is \emph{not} melon-free in exactly three situations:
\begin{enumerate}
    \item $X$ is connected on one side to another dipole or ladder $Y$ (of arbitrary type), as illustrated in Figure \ref{fig:create-melon1}. In particular, the root-vertex of $G$ is not in between $X$ and $Y$; 
 
    \item $X$ is inserted in between two edges of an elementary $2$-point melon, as illustrated in Figure \ref{fig:create-melon2}. Note that the choice of pair of edges on which $X$ is inserted, and therefore their orientation, is fixed by the requirement that $X$ is connecting. As a result, $X$ must be a N-dipole or a N$\mathrm{_o}$-ladder;

    \item $X$ is forming a $2$-point subgraph on one side of a dipole or ladder $Y$ (of arbitrary type), as illustrated in Figure \ref{fig:create-melon3}. Note that the orientation of the edges is again fixed by the requirement that $X$ is connecting.  As a result, $X$ must be a N$\mathrm{_e}$-ladder. Furthermore, the root-vertex of $G$ is not incident to $X$.
\end{enumerate} 

If we assume instead that $X$ is separating, it can be of any type by Lemma \ref{lem:sep-conn}, and its contraction yields two $\ell=0$ Feynman graphs $G_1$ and $G_2$. One can check that the only configuration which may generate a melonic $2$-point subgraph in $G_1$ or $G_2$ is the one illustrated in Figure \ref{fig:create-melon1}, where the root-vertex of $G$ is not: 1) in between $X$ and $Y$; 2) on the side of $X$ opposite to $Y$.

\begin{figure}[htb]
     \centering
     \begin{subfigure}[b]{.3\textwidth}
         \centering
         \includegraphics[scale=.4]{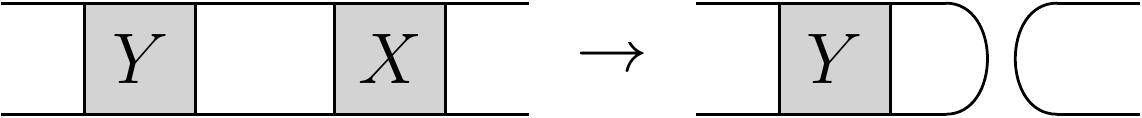}
         \caption{$X$ is connecting or separating.}\label{fig:create-melon1}
     \end{subfigure}
     \begin{subfigure}[b]{0.3\textwidth}
         \centering
         \includegraphics[scale=.4]{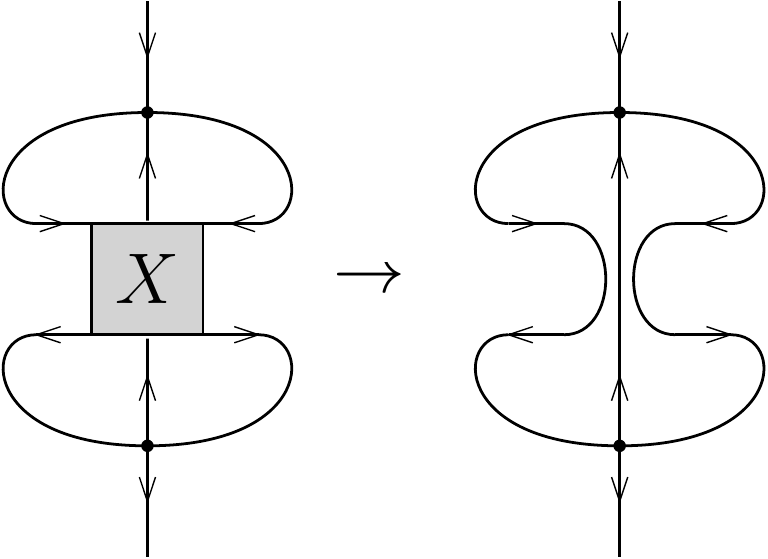}
         \caption{$X$ is connecting.}\label{fig:create-melon2}
     \end{subfigure}     %
     \begin{subfigure}[b]{0.3\textwidth}
         \centering
         \includegraphics[scale=.4]{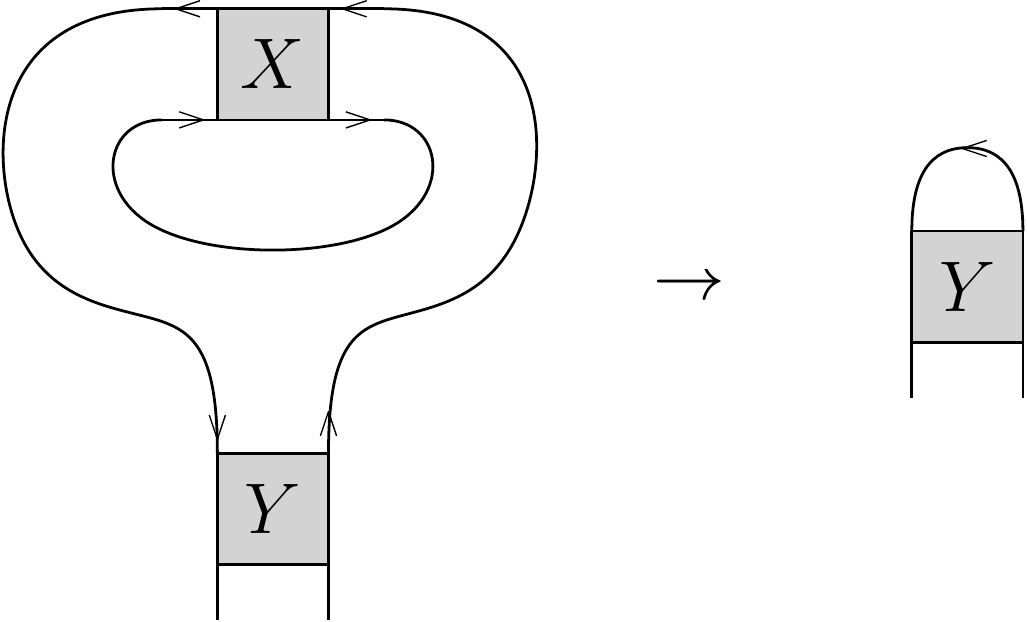}
         \caption{$X$ is connecting.}\label{fig:create-melon3}
     \end{subfigure}        \caption{\small Configurations in which the contraction of a dipole or ladder $X$ in a $\ell=0$ melon-free Feynman graph, assumed to be connecting or separating, generates a melonic $2$-point subgraph (where $Y$ is itself a dipole or a ladder).}
        \label{fig:generate-melon}
\end{figure}

\medskip 

We now study in detail the structure of the $\ell=0$ melon-free Feynman graphs of arbitrary genus. Recall that at genus zero, there is a single $\ell=0$ melon-free Feynman graph, namely the rooted cycle graph. The following lemma provides all the $\ell=0$ melon-free Feynman graphs of genus one by specifying their schemes. 

\begin{proposition}\label{propo:g1}
There are two 2PI schemes of genus one:

\begin{center}
\includegraphics[scale=.55]{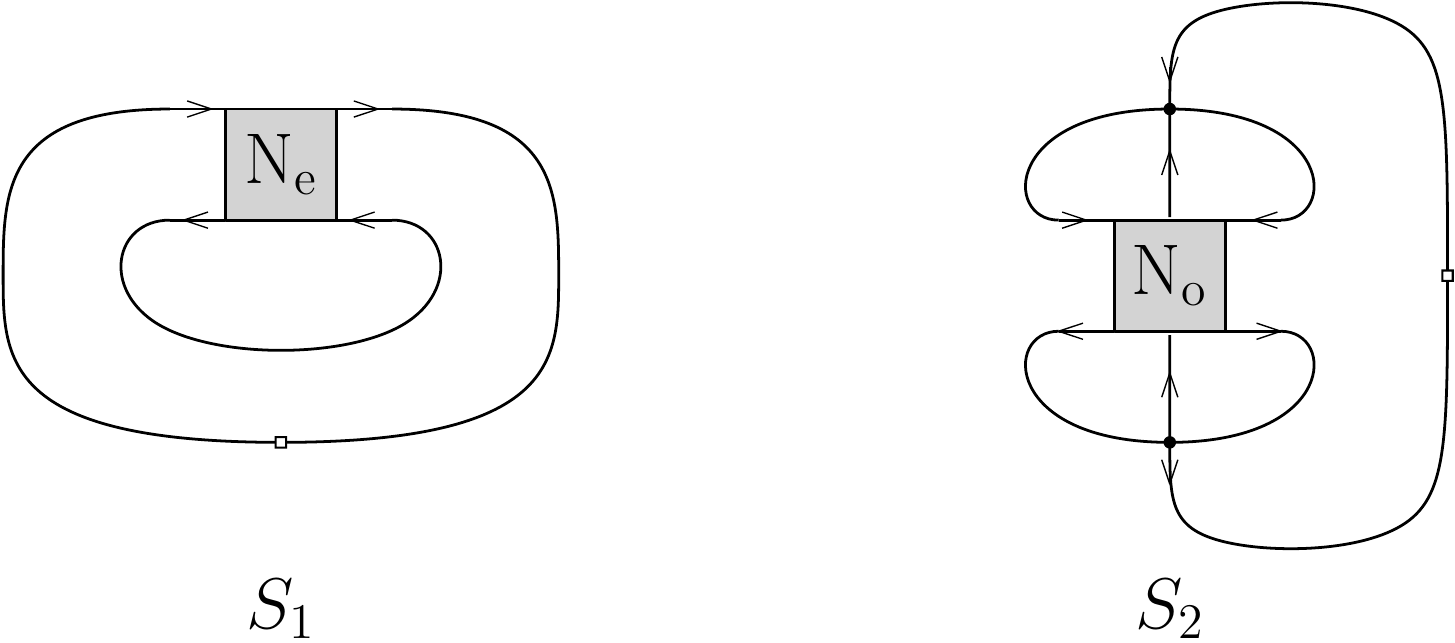}
\end{center}

\medskip

\noindent There are $16$ 2PR schemes of genus one:
\begin{center}
\includegraphics[scale=.55]{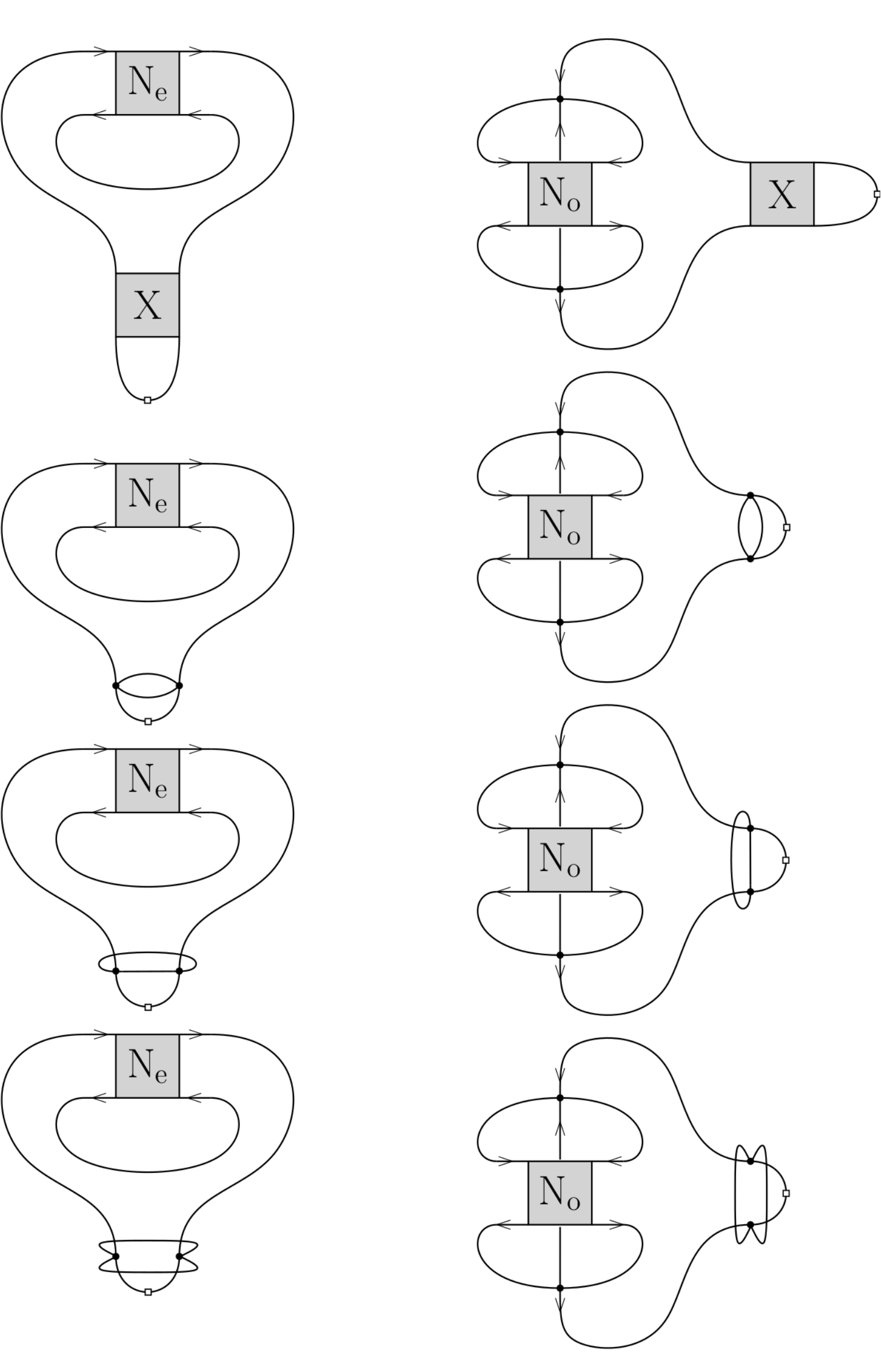}
\end{center}
where $\mathrm{X} \in \{ \mathrm{N_e}, \mathrm{N_o}, \mathrm{L}, \mathrm{R}, \mathrm{B}\}$.
\end{proposition}
\begin{proof}
Let $G$ be a $\ell=0$ melon-free Feynman graph of genus one and $S_G$ its scheme.

We first assume that $G$ is 2PI. By Lemmas \ref{lem:N_dip} and \ref{lem:sep-conn}, there necessarily exists a connecting N-dipole in $G$ (a separating one would break the 2PI condition). Furthermore, the maximal extension of this N-dipole in $G$ is either the N-dipole itself or a maximal N-ladder (a B-ladder is necessarily separating by Lemma \ref{lem:sep-conn}), which translates into a N-dipole or a N-vertex in $S_G$, respectively. Using Eq.\ \eqref{eq:contNonSepN} or \eqref{eq:LVcontNonSepN} with $\sigma=-1$ (by the connecting property), the contraction of the maximal extension of this connecting N-dipole yields a $\ell=0$ Feynman graph $G'$ of genus zero. We distingusih two cases. 

If $G'$ is melon-free, it must be the rooted cycle graph; and the orientation of the edges imposes that the contraction involved a N$_\mathrm{e}$-ladder. As a result, $S_G = S_1$. 

On the other hand, if $G'$ is not melon-free, we must be in one of the configurations shown in Figures \ref{fig:create-melon1}, \ref{fig:create-melon2}, and \ref{fig:create-melon3}, where $X$ represents the maximally extended N-dipole to be contracted. However, configuration \ref{fig:create-melon1} is excluded by maximality of $X$, while configuration \ref{fig:create-melon3} is incompatible with the 2PI character of $G$. We are thus left with configuration \ref{fig:create-melon2}, where $X$ is a N-dipole or a N$_\mathrm{o}$-ladder:
\begin{center}
\includegraphics[scale=.6]{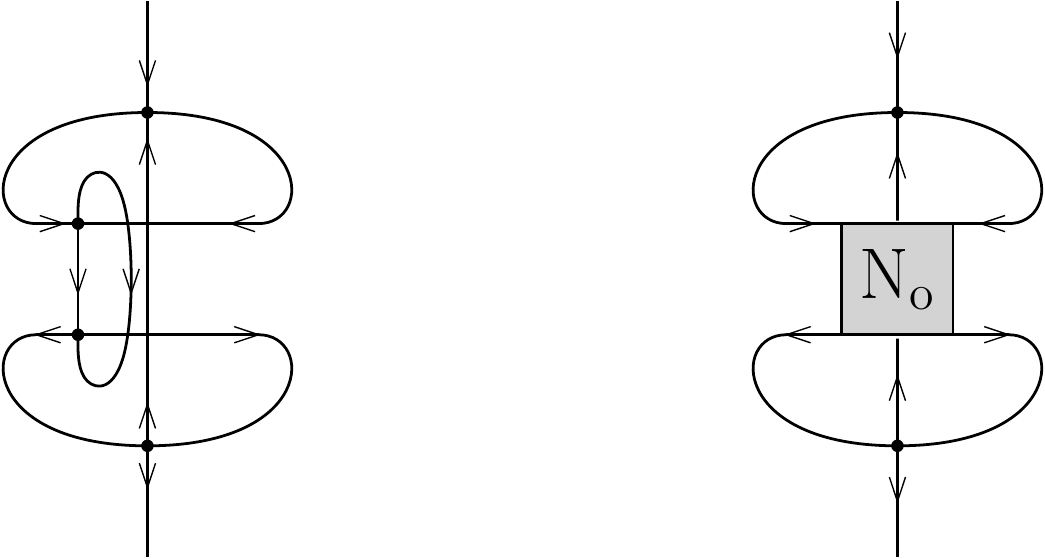}
\end{center}
But then, the 2PI constraint further imposes that these $2$-point subgraphs close onto the root-vertex. In the first case, we find:
\begin{center}
\includegraphics[scale=.6]{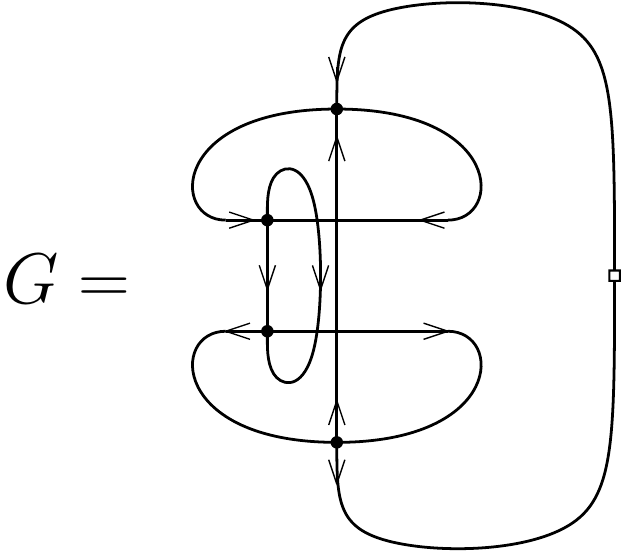}
\end{center}
whose scheme is again $S_1$. In the second case, $G$ has scheme $S_2$.

\medskip

We now assume that $G$ is 2PR. We work in two steps: 1) we first prove that $G$ necessarily contains a separating dipole; and 2) we use that information to construct all the 2PR schemes from the 2PI ones.

Suppose that $G$ does not contain a separating dipole. It must still contain a two-edge-cut $(e,e')$, so that we can perform a flip on $(e,e')$ as in Lemma \ref{lem:2PR}. One can check that the two resulting $\ell=0$ Feynman graphs $G_1$ and $G_2$ cannot contain melonic subgraphs. Indeed, it would otherwise mean that there were a separating dipole in $G$. Since $g(G)=1$, the conditions $g(G_1) \leq g(G)$, $g(G_2) \leq g(G)$, and $g(G_1) + g(G_2) = g(G)$ of Lemma \ref{lem:2PR} further imply that (say) $g(G_1)=0$ and $g(G_2) = 1$. Therefore, $G_1$ must be the rooted cycle graph. But this yields a contradiction because then, by definition, $(e,e')$ cannot be a two-edge-cut.

As a result, $G$ must contain a separating dipole. The maximal extension of this dipole in $G$ is either the dipole itself or a maximal ladder, which we denote by $X$. Contracting $X$ yields two $\ell=0$ Feynman graphs $G_1$ and $G_2$, which cannot contain melonic subgraphs because the situation of Figure \ref{fig:create-melon1} is excluded by maximality of $X$. Using $g(G)=1$ and Eq.\ \eqref{eq:contSep} or \eqref{eq:LVcontSep}, these graphs further obey $g(G_1) \leq 1$, $g(G_2) \leq 1$, and $g(G_1) + g(G_2) = 1$. This implies that (say) $g(G_1)=0$ and $g(G_2) = 1$, so that $G_1$ is the rooted cycle graph. We thus have the following structure: 
\begin{center}
\includegraphics[scale=.8]{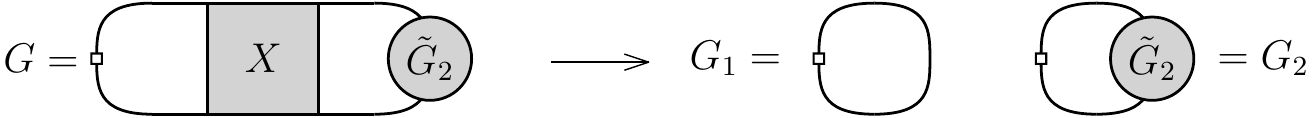}
\end{center}
In particular, we have shown that the maximal extension $X$ of a separating dipole in $G$ is necessarily adjacent to the root-vertex, as represented above. 

We have also shown that $G_2$ is itself a $\ell=0$ melon-free Feynman graph of genus one. In addition, one can verify that $G_2$ cannot be 2PR. Otherwise, by the arguments of the previous paragraph, it would contain a separating dipole whose maximal extension is adjacent to the root-vertex, which contradicts the maximality of $X$ in $G$. As a result, $G_2$ is 2PI and has scheme $S_1$ or $S_2$. 

Finally, by replacing $X$ in $G$ by an arbitrary separating dipole or ladder, which translates into a separating dipole or ladder-vertex in $S_G$, we obtain the 16 2PR schemes we were after, which concludes the proof.  
\end{proof}

The following proposition then gives a way of constructing all the $\ell = 0$ melon-free Feynman graphs of genus $g$, from the family of $\ell = 0$ melon-free Feynman graphs of genus $g' < g$.  

\begin{theorem}\label{thm:induction}
Let $G$ be a connected, rooted and melon-free Feynman graph, such that $\ell(G) = 0$ and $g(G) = g \geq 2$.

Suppose first that $G$ is 2PI. Then it can be obtained by insertion of a connecting N-dipole or N-ladder into a: $\ell = 0$, connected, rooted and melon-free Feynman graph of genus $g-1$.

Suppose instead that $G$ is 2PR. Then one of these two conditions holds:
\begin{enumerate}[label=(\roman*)]
    \item $G$ can be obtained by insertion of a separating dipole, a separating ladder, or a two-edge-connection in-between: $\ell = 0$, connected, rooted and melon-free Feynman graphs $G_1$ and $G_2$, such that $g(G_1) < g$, $g(G_2) <g$, and $g(G_1) + g(G_2) = g$.
    
    \item $G$ can be obtained by insertion of a separating dipole or a separating ladder in-between: the rooted cycle graph and a $\ell=0$, connected, rooted and melon-free Feynman graph $G_2$, such that $g(G_2) = g$ and $G_2$ is 2PI. Furthermore, the insertion is performed with respect to the root-vertex of $G_2$.
\end{enumerate}
\end{theorem}
\begin{proof} 
We follow the same strategy as in the proof of Proposition \ref{propo:g1}.

\medskip

If $G$ is 2PI, there exists a connecting N-dipole in $G$, which can be maximally extended. Contracting the resulting N-dipole or N-ladder $X$ yields a $\ell=0$ Feynman graph $G'$ of genus $g(G') = g-1$. Furthermore, $G'$ is necessarily melon-free because the three configurations of Figure \ref{fig:generate-melon} are all excluded: configuration \ref{fig:create-melon1} because $X$ is maximal, configuration \ref{fig:create-melon2} because $G$ is 2PI and $g(G)\neq 1$, and configuration \ref{fig:create-melon3} because $G$ is 2PI. This proves the first part of the proposition.

\medskip

If $G$ is 2PR, we prove that the negation of \textit{(i)} implies \textit{(ii)}, by generalizing the arguments of Proposition \ref{propo:g1}. 

Suppose first that $G$ does not contain a separating dipole. Since $G$ is 2PR, it must still contain a two-edge-cut $(e,e')$. By Lemma \ref{lem:2PR}, performing a flip on $(e,e')$ yields two Feynman graphs $G_1$ and $G_2$ such that $g(G_1) \leq g$, $g(G_2) \leq g$, and $g(G_1) + g(G_2) = g$. As in Proposition~\ref{propo:g1}, $G_1$ and $G_2$ are necessarily melon-free. The negation of \textit{(i)} then implies that (say) $g(G_1)=0$ and $g(G_2) = g$, that is, $G_1$ is the rooted cycle graph, which leads to the same contradiction as in Proposition \ref{propo:g1}.

It must therefore be that $G$ contains a separating dipole. As in Proposition \ref{propo:g1}, we proceed to contracting its maximal extension $X$, itself a separating dipole or a separating ladder. The two resulting $\ell=0$ Feynman graphs $G_1$ and $G_2$ are again melon-free and obey: $g(G_1) \leq g$, $g(G_2) \leq g$, and $g(G_1) + g(G_2) = g$. Together with the negation of \textit{(i)}, it implies that (say) $g(G_1)=0$ and $g(G_2) = g$. We thus have the same structure as in the proof of Proposition \ref{propo:g1}, namely: 
\begin{center}
\includegraphics[scale=.8]{g1_2PR_proof.pdf}
\end{center}
In particular, we have shown that when $G$ does not satisfy \textit{(i)}, the maximal extension of any separating dipole in $G$ is adjacent to the root-vertex.

Furthermore, we have also shown that $G_2$ is itself a $\ell=0$ melon-free Feynman graph of genus $g$. Suppose that $G_2$ is 2PR. One can check that $G_2$ cannot obey condition \textit{(i)} (otherwise, $G$ itself would). Hence, by the arguments of the previous paragraph, $G_2$ must contain a separating dipole, whose maximal extension is adjacent to its root-vertex. But this is in contradiction with the maximality of $X$ in $G$. As a result, $G_2$ must be 2PI. Finally, it is clear from the above structure for $G$ that the insertion of $X$ is performed with respect to the root-vertex of $G_2$, which concludes the proof.
\end{proof}

\begin{remark}
Even though it is not obvious from the proof of Theorem \ref{thm:induction} itself, the situation \textit{(i)} can always be achieved, for any values of $g$, $g(G_1)$ and $g(G_2)$ allowed by the Theorem.
\end{remark}

\medskip

\paragraph{Algorithmic construction of all $\ell=0$ graphs.}
The Theorem \ref{thm:induction} is a key result of this paper. It provides a constructive way of generating all the $\ell =0$ Feynman graphs, order by order in the genus and starting at genus one, by application of the following algorithm:
\begin{enumerate}
    \item Assume that the set $\hat{E}_{g'}$ of $\ell = 0$ melon-free Feynman graphs of genus $g'\leq g-1$ has been constructed; 
    
    \item Construct the set $\hat{E}^\mathrm{2PI}_{g}$ of 2PI $\ell = 0$ melon-free Feynman graphs of genus $g$ by inserting a connecting N-dipole or N-ladder into any element of $\hat{E}_{g-1}$, in any possible way that yields a 2PI graph. Note that there exist situations in which the graph before the insertion is itself 2PR, it is therefore important to start from elements of $\hat{E}_{g-1}$ as opposed to $\hat{E}^\mathrm{2PI}_{g-1}$. Topologically, this step increases the genus by adding handles on lower-genus Riemann surfaces; 
    
    \item Obtain a first class of 2PR contributions of genus $g$ by inserting a separating dipole or ladder in between the rooted cycle graph and any element of $\hat{E}^\mathrm{2PI}_{g}$ (as illustrated in the examples of Figure \ref{fig:0separating}, with $G_1$ the rooted cycle graph). This step does not change the topology, and is only required because we work with rooted graphs;
    
    \item Obtain a second class of 2PR contributions of genus $g$ by inserting a separating dipole, a separating ladder or a two-edge-connection in between any element of $\hat{E}_{g_1}$ and any element of $\hat{E}_{g_2}$ such that $g_1 + g_2 = g$, in any possible way. Topologically, this step increases the genus by taking the connected sum of topologically non-trivial Riemann surfaces;
    
    \item Construct the set $\hat{E}^\mathrm{2PR}_{g}$ of 2PR $\ell = 0$ melon-free Feynman graphs of genus $g$ by taking the union of the sets obtained in the previous two steps;
    
    \item Obtain the set $\hat{E}_{g} = \hat{E}^\mathrm{2PR}_{g} \cup \hat{E}^\mathrm{2PI}_{g}$ of all $\ell = 0$ melon-free Feynman graphs of genus $g$;
    
    \item Finally, insert arbitrary melonic rooted Feynman graphs on the edges of any element in $\hat{E}_{g}$ to obtain the set $E_g$ of $\ell=0$ Feynman graphs (see Claim \ref{cl:melon}).
\end{enumerate}

This constructive algorithm is to be contrasted with the characterization of the Feynman graphs of fixed degree performed in the MO model \cite{Fusy:2014rba} or in a $\mathrm{O}(N)^3$ tensor model \cite{Bonzom:2019yik}. Indeed, in these cases, one first needs to determine all the dipole-free schemes of fixed degree, which is in general a challenging task. In contrast, the non-trivial algorithmic simplification we have achieved here is a direct consequence of the double-scaling limit, which allows us to restrict to Feynman graphs with vanishing grade.

In Appendix \ref{app:algorithmex}, we illustrate how $\ell = 0$ graphs of genus two can be constructed with the help of the previously described algorithm. As it will become clear in the next sections, only a subset of the structures generated in this manner are relevant to the analysis of the continuum limit. We will therefore focus on this particular subclass of $\ell = 0$ graphs, which as we will explain, can be conveniently described as collections of ladder diagrams glued together along certain effective six-point vertices. As made apparent by some of the contributions listed in Appendix \ref{app:algorithmex}, we emphasize that this picture would need to be generalized in order to include all $\ell = 0$ subgraphs: in particular, one would need to include one new type of effective eight-point vertex along which ladders can be glued.

\subsection{Schemes of vanishing grade}\label{sec:schemesZeroGrade}

In the previous section, we derived an inductive algorithm to generate all the connected, rooted (and melon-free) Feynman graphs of vanishing grade, order by order in the genus. A natural question is whether a similar algorithm can be constructed in terms of the corresponding schemes of vanishing grade. In particular, we would like to know if such an algorithm involves a finite number of inductive moves. Since there is a finite number of $\ell=0$ schemes of genus one and they are of finite size (see Propoition \ref{propo:g1}), it would indeed imply that there is a finite number of $\ell=0$ schemes of arbitrary fixed genus. This is to be contrasted with the number of $\ell=0$ melon-free Feynman graphs of fixed genus, which is infinite due to the insertions of ladders of arbitrary length. The finiteness of the number of $\ell=0$ schemes of fixed genus plays an important role for computing generating functions, as explained in the next section. As a side remark, we note that this result can be obtained as a direct consequence of the finiteness of the number of schemes of fixed degree, which is derived in \cite{Fusy:2014rba}. However, in the present context, we can make use of the double-scaling limit to have a better idea of the structure of the $\ell=0$ schemes of fixed genus.

Let us discuss how Theorem \ref{thm:induction} can be adapted to $\ell=0$ schemes. This result provides an induction in the class of connected, rooted and melon-free Feynman graphs of vanishing grade. However, a direct extension of the arguments using the same type of insertion moves does not allow to close the induction on the class of $\ell=0$ schemes.

As an illustration, suppose that we replace $G$ in Theorem \ref{thm:induction} with its $\ell=0$ scheme $S_G$ of genus $g\geq2$. By Lemma \ref{lem:schemes2Pi}, the notion of 2PR and 2PI is the same for $G$ and $S_G$. Furthermore, it is natural to replace the insertions of ladders with insertions of ladder-vertices. With these replacements, one could carry out the same strategy of proof as in the proposition. Since we always contract the maximal extension of a dipole in $G$, it is equivalent to contracting the corresponding dipole or ladder-vertex in $S_G$. Then, the same reasoning shows that the contraction of a connecting (resp.\ separating) dipole or ladder-vertex in $S_G$ yields a $\ell=0$ Feynman graph with ladder-vertices $G'$ (resp.\ two $\ell=0$ Feynman graphs with ladder-vertices $G_1$ and $G_2$) which is (resp.\ are) melon-free. However, it is straightforward to see that the resulting graphs are not necessarily ladder-free; in other words, they are not necessarily schemes. As a result, a closed induction on $\ell=0$ schemes can not be directly performed in this way.

A possible alternative is to first analyze the situations in which the contraction of a dipole or a ladder-vertex in $S_G$ generates a ladder in $G'$, $G_1$ or $G_2$, and then perform further operations on these graphs so as to obtain schemes. Let us consider a dipole or a ladder-vertex $X$ in $S_G$. The contraction of $X$ gives rise to two edges $e$ and $e'$, which can be distinct or not:\footnote{If $X$ is separating, a root-vertex needs to be added in the middle of $e$ or $e'$: this cannot generate a ladder in the Feynman graph connected to that edge. Also, remark that the effect of a flip move is essentially identical to that of a separating dipole or ladder-vertex contraction; we focus on the latter for definiteness.}
\vspace{0.2cm}
\begin{center}
\includegraphics[scale=.8]{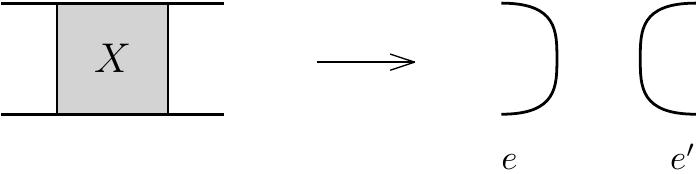}
\end{center}
Suppose that $G'$, $G_1$ or $G_2$ contains a ladder $\cL$. We distinguish the following cases:
\begin{enumerate}
    \item $e,e'\not\in\cL$;
    \item $e\in\cL\,, e'\not\in\cL$ (or $e'\in\cL\,, e\not\in\cL$);
    \item $e,e'\in\cL\, , e=e'$;
    \item $e,e'\in\cL\,, e\neq e'$.
\end{enumerate}
It is straightforward to see that the first case cannot happen because it would mean that there is a ladder in $S_G$, which is impossible. As for the remaining cases, one can check that $\cL$ necessarily has one of the following combinatorial structures:\footnote{We do not keep track of the embedding, edge orientations, or the position of the root-vertex, since these ingredients do not play any essential role in the argument. In particular, the dipoles appearing in these figures can be of any type, provided that it is consistent with the $\ell=0$ condition.}
\begin{center}
\includegraphics[scale=0.8]{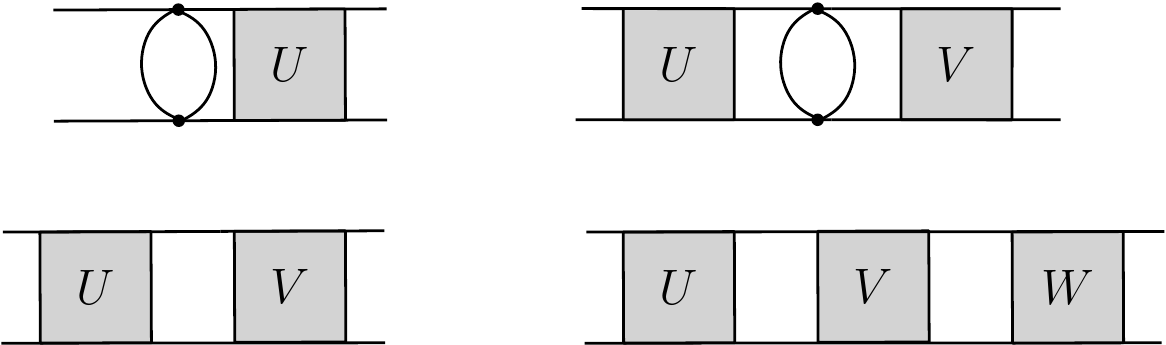}
\end{center}
where $e$ (and/or $e'$) appears as an edge in these diagrams, and  $U, V$ and $W$ are dipoles or ladder-vertices (of arbitrary type, consistent with $S_G$ being a $\ell=0$ scheme). Enumerating all the possible situations that generate one of the ladder structures represented above is quite cumbersome and in fact, it is not required. The main point is that there is only a finite number of possibilities, since there is a finite number of cases and a finite number of structures for $\cL$. In summary, we have argued that there is a finite number of configurations in which the contraction of a dipole or a ladder-vertex $X$ in $S_G$ generates a ladder $\cL$ in $G'$, $G_1$ or $G_2$. 

As a second step, we can replace the ladder $\cL$ in $G'$, $G_1$ or $G_2$ with a ladder-vertex of the consistent type. As a consequence of the definition of the genus and the grade of a Feynman graph with ladder-vertices (see Section \ref{sec:schemes}), this replacement does not affect the genus and the grade of $G'$, $G_1$ or $G_2$. Furthermore, it cannot generate melonic subgraphs. Hence, the resulting Feynman graphs with ladder-vertices, which we respectively denote by $\hat{G}', \hat{G}_1$ or $\hat{G}_2$, must correspond to $\ell=0$ schemes of genus $g(\hat{G}')=g(G')$, $g(\hat{G}_1)=g(G_1)$ or $g(\hat{G}_2)=g(G_2)$. 

We have thus shown that the algorithm of Theorem \ref{thm:induction} can be adapted to $\ell=0$ schemes by taking additional care of the configurations in which the contraction of a dipole or ladder-vertex $X$ in $S_G$ generates a ladder. In this way, the induction can be closed on the set of $\ell=0$ schemes. Furthermore, at each step of the induction, there is a finite number of operations. As explained earlier, since there is a finite number of $\ell=0$ schemes of genus one, which are of finite size, we conclude that there is a finite number of $\ell=0$ schemes of arbitrary fixed genus.

\subsection{Connected two-point function}\label{sec:2point}

We are now ready to express the double-scaled two-point function as a weighted sum of schemes.

Let us call $\cS_g$ the set of schemes of vanishing grade with genus $g$. For each $S \in \cS_g$, we first construct the generating function of the connected rooted melon-free Feynman graphs corresponding to such scheme:
\be
\hat\cG_S (u) = \sum_{n \in \mathbb{N}} \hat\cG_{S,n} u^{p+n}\,, 
\ee
where $2 p$  is the number of non-root standard vertices of the scheme $S$, and $\hat\cG_{S,n}$ is the number of connected rooted melon-free Feynman graphs of scheme $S$ with $2n$ vertices inside the ladders. 
If $S$ is the cycle graph, we have $\hat\cG_{S,n}=\delta_{n,0}$ and $p=0$.

It is convenient to introduce also a generating function $\cC_X$ for each type of  ladder-vertex $X$. These are easily evaluated as (sums of) geometric series, and one finds:
\be\label{eq:ne}
\cC_{\Ne}(u)=\frac{u^2}{1-u^2} \, ,
\ee
\be\label{eq:no}
\cC_{\No}(u)=\frac{u^3}{1-u^2} \, ,
\ee
\be
\cC_{\mathrm{L}}(u)=\cC_{\mathrm{R}}(u)=\frac{u^2}{1-u} \, ,
\ee
\be\label{eq:nb}
\cC_{\mathrm{B}}(u)=\frac{(3u)^2}{1-3u}-\cC_{\Ne}(u)-\cC_{\No}(u)-\cC_{\mathrm{L}}(u)-\cC_{\mathrm{R}}(u) = \frac{6u^2}{(1-3u)(1-u)} \, ,
\ee
where $u$ is again a parameter that counts half the number of vertices (or, equivalently, the number of rungs in a ladder). The generating function $\hat\cG_S (u)$ is then: 
\be\label{eq:generating-scheme}
\hat\cG_{S}(u)=  u^p \cC_{\Ne}(u)^{n_e}\cC_{\No}(u)^{n_o}\cC_{\mathrm{L}}(u)^{l+r}\cC_{{B}}(u)^b \,,
\ee
where $b$, $n_e$, $n_o$, $l$, and $r$ are the respective numbers of $\mathrm{B}$-, $\mathrm{N}_\mathrm{e}$-, $\mathrm{N}_\mathrm{o}$-, $\mathrm{L}$-, and $\mathrm{R}$-vertices in $S$.

To obtain the sum over all connected rooted Feynman graphs, including melonic decorations, we simply need to substitute $u$ by $U(\lambda):= \lambda^2 T(\lambda)^4$, and multiply the result by $T(\lambda)$ to account for the extra propagator associated to the root-vertex. 
The generating function of connected rooted Feynman graphs of genus $g$ and grade $\ell=0$ is then:
\be\label{eq:sum_schemes_g}
\cG_{g}(\lambda)= \sum_{S \in \cS{g}}  T(\lambda) \, \hat\cG_{S}(T(\lambda)-1) \, ,
\ee
where we have used the melonic equation \eqref{eq:melonEq} to write $U(\lambda) = T(\lambda) -1$. 

From the point of view of  field theory, the presence of a root in the graphs corresponds to studying the two-point function $\cG^{(0)}(\l)$, rather than the free energy \eqref{eq:free-en-0}.
The two are related by the Schwinger-Dyson equation
\be
\left\la \Tr\left[ X^\dagger_\m \f{\d S}{\d X^\dagger_\m}\right] \right\ra = N^2 D\,,
\ee
implying
\be \label{eq:2pt}
\cG^{(0)}(\l) \equiv \f{N}{D} \left\la \Tr\left[ X^\dagger_\m X_\m \right]\right\ra = M^2+2\l\p_\l \mathcal{F}^{(0)}(M,\lambda) \,.
\ee
Notice that with the choice of scaling in the action \eqref{eq:actionND}, the two-point function of the free theory is $\left\la \Tr\left[ X^\dagger_\m X_\m \right]\right\ra_{\l=0} =N$, hence $\cG^{(0)}(0)=M^2$, consistently with $\cG_{0}(0)=1$ and $\cG_{g>0}(0)=0$.

Therefore, the connected two-point function in the double-scaling limit has  a similar expansion to that of the free energy \eqref{eq:free-en-0}. Owing to our choice of combinatorial factors in the action \eqref{eq:actionND}, it has a very simple expression in terms of the generating functions $\cG_{g}(\lambda)$:
\be \label{eq:2pt-1/M}
 \cG^{(0)}(\l) = \sum_{g\in \mathbb{N}} \cG_{g}(\lambda) M^{2-2g} \, .
\ee

The melonic two-point function $T(\lambda)$, solving equation \eqref{eq:melonEq}, is well known. It has a dominant singularity at the critical value of the coupling constant $\lambda_c = \sqrt{3^3/4^4}$, with the following singular behavior \cite{critical, RTM}:
\be\label{eq:singT}
T(\l) \underset{\lambda \to \lambda_c^-}{\approx} \frac{1}{3} \Bigl(4 -\sqrt{\frac{8}{3}}\sqrt{1-\frac{\l^2}{\l_c^2}}\Bigr) \, .
\ee
When the melonic two-point function reaches criticality, the function $U(\lambda)$ approaches the value $u_c := T(\l_c) -1 = 1/3$ (from below), at which $\cC_{\mathrm{B}}(u)$ itself becomes critical. The other generating functions of ladder-vertices stay instead regular, as their dominant singularity is at $\vert u \vert=1 > u_c$. As a result of Eqs.\ \eqref{eq:generating-scheme} and \eqref{eq:sum_schemes_g}, the most singular part of $\cG_g (\lambda)$ in the limit $\lambda \to \lambda_c$ will be governed by schemes that maximize the number of B-vertices.

\section{Dominant schemes of vanishing grade}\label{sec:dominantSchemes}

In view of the preceding section, and following the nomenclature of \cite{GurSch, Fusy:2014rba}, we say that a scheme is \emph{dominant} if it contains a maximal number of B-vertices allowed by its genus. The dominant $\ell=0$ schemes of genus $g$ pick up the most singular contributions in the expansion \eqref{eq:sum_schemes_g} of $\mathcal{G}_g$; they therefore determine the behavior of the multi-matrix model in the critical limit $\l \to \l_c$.

We will first show that the dominant schemes have the combinatorial structure of decorated plane binary trees. This fact will allow us to explicitly resum them, and define a triple-scaling limit retaining contributions with arbitrary values of the genus. Similarly to the melonic limit, we will find that this new scaling limit admits a critical regime dominated by large trees. However, in contrast to melonic diagrams, the tree-like structure of dominant schemes encodes Riemann surfaces of non-zero genus, the expectation value of which diverges at the critical point. 

\subsection{One-to-one mapping to plane binary trees}

We first determine the maximal number of B-vertices in a $\ell=0$ scheme as a function of its genus.
\begin{lemma}\label{lem:maxBvertices}
Let $S$ be a $\ell=0$ scheme of genus $g\geq1$ with $b$ B-vertices. Then, $b\leq2g-1$ and the bound can be saturated.  
\end{lemma}

\begin{proof}
All the $\ell=0$ schemes of genus $g=1$ are identified in Proposition \ref{propo:g1} and we see explicitly that $b\leq2g-1=1$. There are furthermore two such schemes that saturate the bound.

Assume that $g>1$. If $b\leq2$, then $b\leq2<2g-1$. We can therefore suppose that $b\geq3$. Then, there exists at least one B-vertex in $S$ which is not adjacent to the root-vertex. Furthermore, by Lemma \ref{lem:sep-conn}, any B-vertex is necessarily separating. We can thus perform a flip on one side of this B-vertex, as illustrated in the following figure:
\begin{center}
\includegraphics[scale=.7]{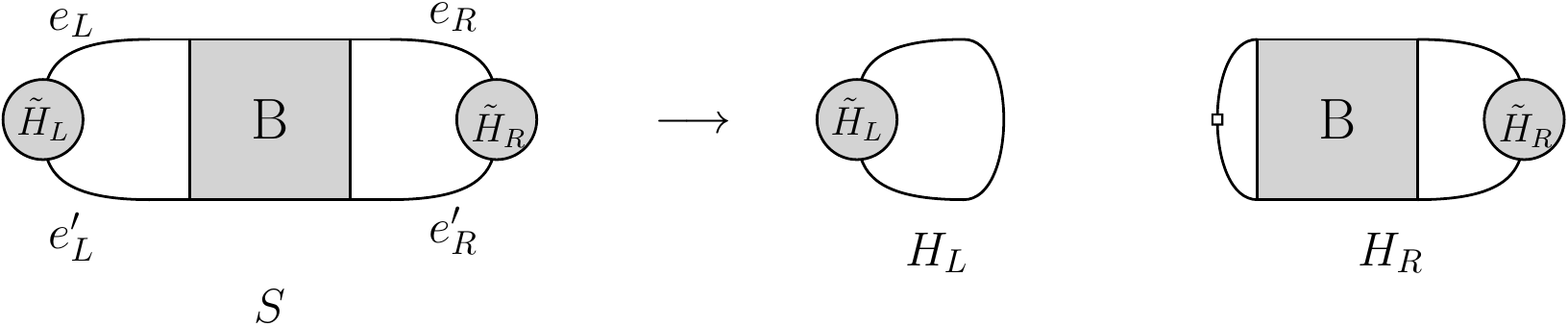}
\end{center}
where $e_L, e'_L, e_R$ nor $e'_R$ is adjacent to the root-vertex of $S$ while $\tilde{H}_L$ and $\tilde{H}_R$ are two connected $2$-point subgraphs which are non-empty (since $S$ is a scheme; hence, it is melon-free) and distinct from the root-vertex alone. In addition, we assume without loss of generality that the root-vertex of $S$ is in $\tilde{H}_L$. Using Eq.\ \eqref{eq:LVcontSep} with $\ell(S)=0$ and $g(S)=g$, the graphs $H_L$ and $H_R$ obtained after the flip are both $\ell=0$ Feynman graphs with ladder-vertices satisfying $g(H_R) +g(H_L) =g$. In addition, since we consider a B-vertex not adjacent to the root-vertex of $S$, we further have $g(H_L)\geq 1$ and $g(H_R)\geq 1$, which implies that $g(H_L)<g$ and $g(H_R)<g$. Finally, by the same argument as in the proof of Theorem \ref{thm:induction} for 2PR graphs, $H_L$ and $H_R$ are necessarily melon-free. However, as discussed in Section \ref{sec:schemesZeroGrade}, they do not necessarily correspond to schemes because they may contain a ladder. In this case, if we want to apply the induction hypothesis to both $H_L$ and $H_R$, we need to perform further operations so as to obtain $\ell=0$ schemes from them, in the same spirit as in Section \ref{sec:schemesZeroGrade}.

We first observe that because we add a root-vertex on the left side of the B-vertex in $H_R$, as illustrated in the above figure, the latter is necessarily ladder-free. It is therefore a $\ell=0$ scheme of genus $1\leq g(H_R)<g$ and by the induction hypothesis, it obeys $b(H_R)\leq 2g(H_R)-1$.

Next, if we assume that there is no ladder generated in $H_L$, then it also corresponds to a $\ell=0$ scheme of genus $1\leq g(H_L)<g$. Using the induction hypothesis, it thus verifies $b(H_L)\leq 2g(H_L)-1$. We then find that 
\be
b = b(H_L) + b(H_R) \leq 2 g(H_L) + 2 g(H_R) - 2 = 2 g - 2 < 2g-1\, .
\ee 
Hence, the bound on $b$ is verified but it is not saturated.

If we assume instead that there is a ladder $\cL$ generated in $H_L$, then this ladder necessarily contains the edge $e$ that reconnects $e_L$ and $e_L'$ in $H_L$. We are thus in the case 2.\ of Section \ref{sec:schemesZeroGrade}. By studying the possible structures for $\cL$ in this section, we deduce that $H_L$ has one of the following structures (up to embedding and edge orientations, which we ignore for the moment): 
\begin{center}
\includegraphics[scale=.5]{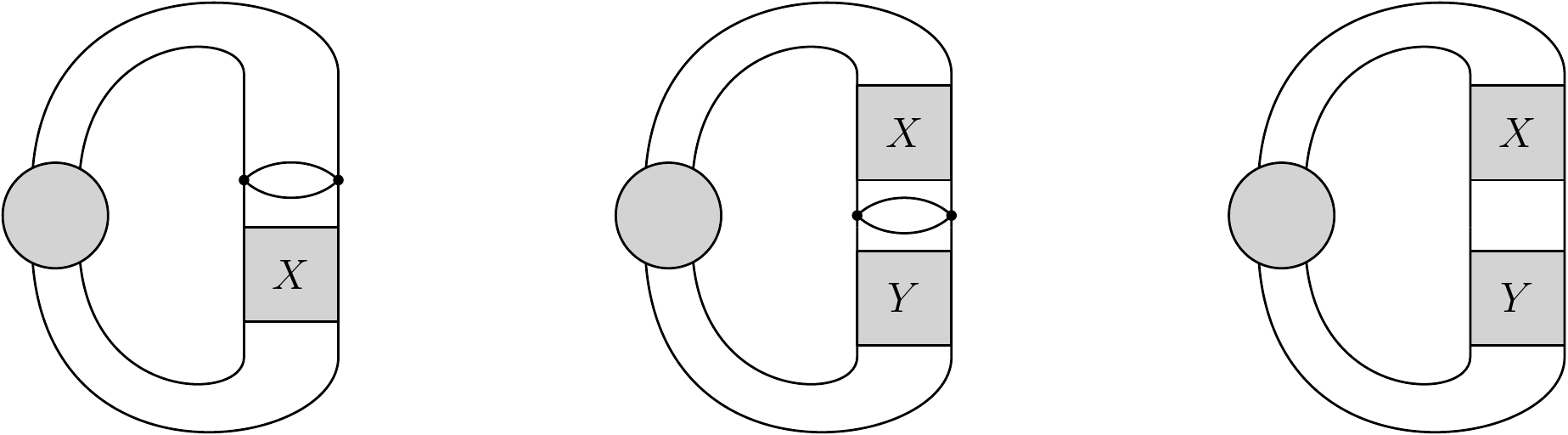}
\end{center}
where $X$ and $Y$ are two dipoles or ladder-vertices (of any type, consistent with $S$ being a $\ell=0$ scheme). In addition, the remaining part of $H_L$ is represented with a shaded disk, which contains the root-vertex and which can either be a connected $4$-point subgraph or factorize in two connected $2$-point subgraphs. 

We then replace the generated ladder $\cL$ in $H_L$ by a ladder-vertex of the consistent type, as explained in Section \ref{sec:schemesZeroGrade}. This yields a $\ell=0$ scheme $\hat{H}_L$ with the same genus as $H_L$. In addition, one can observe that at most two B-vertices in $H_L$ (possibly with a dipole) are replaced with a ladder-vertex in $\hat{H}_L$, and this ladder-vertex necessarily corresponds to a B-vertex itself. Hence, we have $b(H_L)\leq b(\hat{H}_L)+1$. We can then apply the induction hypothesis to $\hat{H}_L$ so that $b(\hat{H_L})\leq 2g(H_L)-1$, which ultimately leads to
\be
b = b(H_L) + b(H_R) \leq b(\hat{H}_L) + b(H_R) + 1 \leq 2 g(H_L) + 2 g(H_R) - 1 = 2 g - 1\, .
\ee 
Finally, this bound can be saturated if all the previous bound can also be saturated. On the one hand, the bounds for $b(\hat{H}_L)$ and $b(H_R)$ can be saturated by the induction hypothesis. On the other hand, the bound which relates $b(H_L)$ to $b(\hat{H}_L)$ can also be saturated if the ladder $\cL$ in $H_L$ contains two B-vertices. In particular, using the fact that B-vertices are necessarily separating when $\ell=0$, it means that $\tilde{H}_L$ has one of the structures given in Figure \ref{fig:HLstructure}, where the two external edges correspond to $e_L$ and $e_L'$, and the shaded disks correspond to connected $2$-point subgraphs. Note that we have restored relevant orientations, as well as embedding information in this Figure. This concludes the proof.
\end{proof}
\begin{figure}[htb]
\centering
\includegraphics[scale=.6]{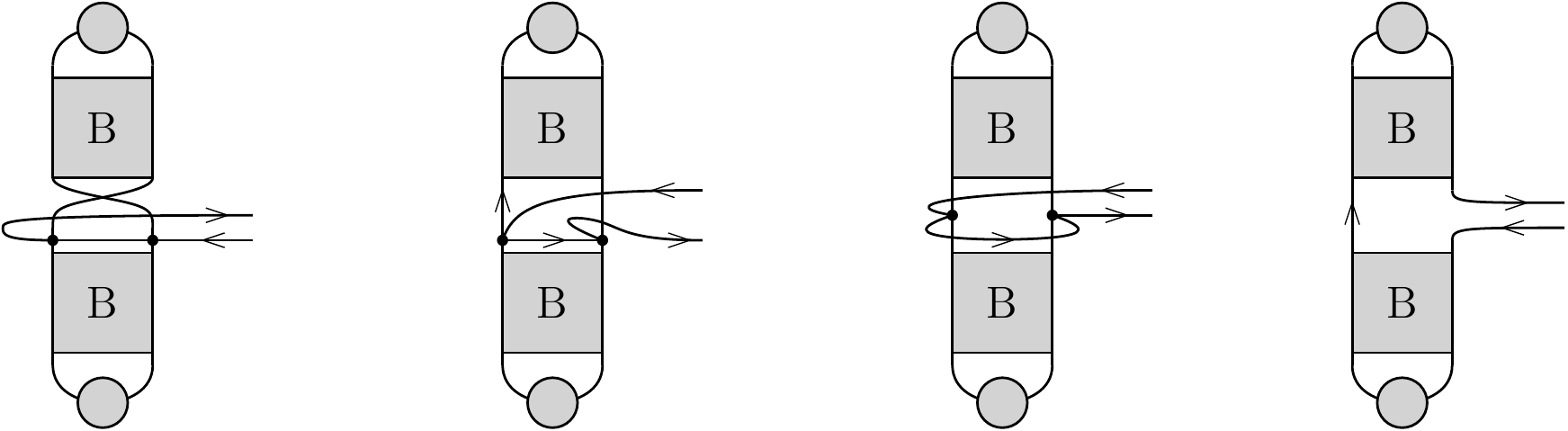}
\caption{\small Local branching structures that maximize the number of B-ladders, and therefore lead to dominant schemes. The twist on the leftmost figure has been introduced to comply with our embedding conventions.}\label{fig:HLstructure}
\end{figure}

The inductive construction used in the proof of the previous lemma immediately suggests how to saturate the upper bound on the number of B-vertices in a $\ell=0$ scheme; and therefore, how to obtain the dominant $\ell=0$ schemes. In fact, one must organize the B-vertices into a tree-like binary structure, whose leaves and vertices can be extracted from the combinatorial structures encountered previously. This is the purpose of the next proposition, which is the main result of this section.
\begin{proposition}\label{propo:dominant-schemes}
Let $S$ be a $\ell=0$ scheme of genus $g\geq1$. $S$ is dominant (i.e.\ $b(S)=2g-1$) if and only if it has the structure of a plane\footnote{We recall that a \emph{plane tree} is a tree embedded on the plane, i.e.\ a tree together with an ordering of the edges around each vertex, which has a distinguished valency-one vertex called the \emph{root}. In the present case, we emphasize that the trees associated with the dominant $\ell=0$ schemes are rooted because of the presence of a root-vertex.} binary tree with $b(S)=2g-1$ edges, $g$ leaves and $g-1$ inner vertices, such that
\begin{itemize}
    \item the root corresponds to the root-vertex:
\begin{center}
\includegraphics[scale=.6]{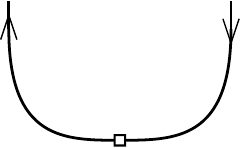}
\end{center}
    \item each edge corresponds to a B-vertex:
\begin{center}
\includegraphics[scale=.8]{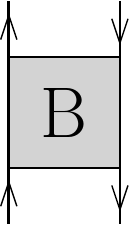}
\end{center}
    \item each leaf corresponds to one of the following types of $2$-point subgraphs:
\begin{center}
\includegraphics[scale=.6]{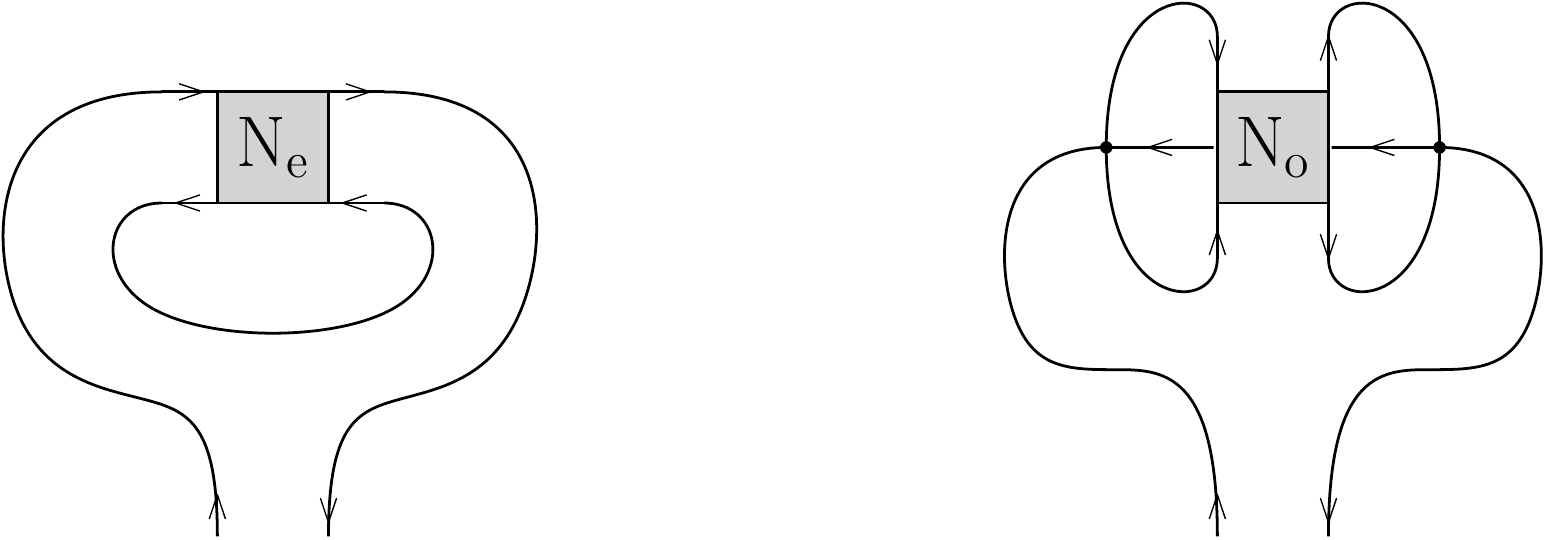}
\end{center}
    \item each inner vertex corresponds to one of the following types of $6$-point subgraphs:
\begin{figure}[H]
\centering
\includegraphics[scale=.75]{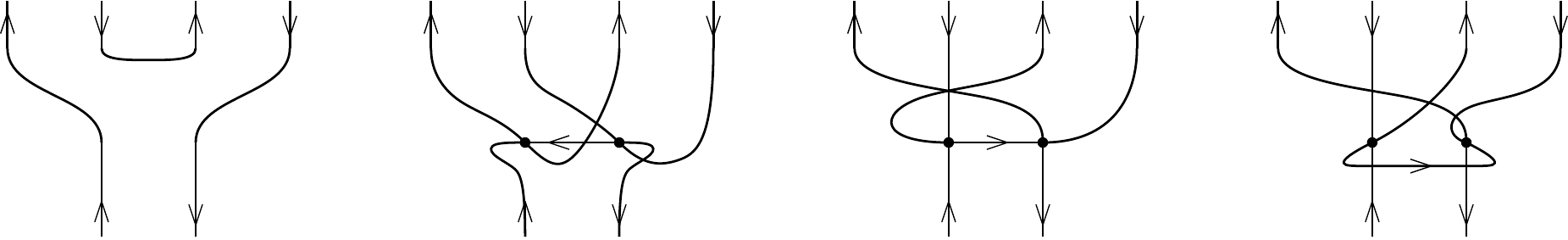}
\caption{}
\label{fig:tree_vertices}
\end{figure}
We refer to the leftmost structure as \emph{planar vertex}, and to the three others as \emph{contact vertices}. Note that we rely on an arbitrary but fixed convention for the embedding of the tree and its inner vertices, in which the root is in the tree component glued at the bottom.
\end{itemize} 
\end{proposition}
\begin{proof} This is a direct consequence of the proof of Lemma \ref{lem:maxBvertices}. For the sake of clarity, let us be more exhaustive about the structure of the trees. 

Let $S$ be a dominant $\ell=0$ scheme of genus $g\geq1$; hence with $b=2g-1$ B-vertices using Lemma \ref{lem:maxBvertices}. By saturating the bound on the number of B-vertices at each inductive step of the proof of Lemma \ref{lem:maxBvertices} (so as to obtain dominant schemes), it is clear that $S$ has the structure of a binary tree (see the structure of $\tilde{H}_L$ in Figure \ref{fig:HLstructure} for instance).

At some step $i\in\{1,2, \ldots, g-1\}$ of the induction, there will necessarily be a dominant $\ell=0$ scheme $\hat{H}^{(i)}_L$ of genus one with a single B-vertex. Identifying the two possibilities in Proposition \ref{propo:g1}, one observes that the two edges of the root-vertex are adjacent to one side of the B-vertex. Rewinding the induction up to $S$, this shows that the root of the tree is of the stated form.

Besides, the induction of Lemma \ref{lem:maxBvertices} performs a flip on one side of a B-vertex at each inductive step. As a result, the edges of the binary tree associated with $S$ correspond to B-vertices by construction.

Next, the structure of the leaves of the rooted binary tree associated with $S$ also follows from the induction in the proof of Lemma \ref{lem:maxBvertices}. Indeed, applying the corresponding arguments while saturating the bound on the number of B-vertices at each inductive step yields, starting from $S$, $g$ dominant $\ell=0$ schemes of genus one. Again, identifying the two possibilities in Proposition \ref{propo:g1} and rewinding the induction up to $S$ shows that the $g$ leaves of the tree associated with $S$ are of one of the two stated types.

Finally, let us study the inner vertices of the rooted binary tree associated with $S$. Their structure is also a consequence of the inductive construction of the dominant schemes deduced from the proof of Lemma \ref{lem:maxBvertices}, and in particular from the structure of $\tilde{H}_L$ in Figure \ref{fig:HLstructure}. In fact, a convenient way of analyzing the inner vertices can be obtained as follows. 
Starting from $S$ and using its tree-like structure described above, one can contract all the B-vertices originally present in $S$, one by one starting from the bottom of the tree. That is, we first contract the B-vertex closest to the root of the tree, then the B-vertex closest to the newly added root, and continue iteratively until all the B-vertices have been contracted. This amounts to deleting all the edges of the corresponding plane binary tree.
Several connected components are generated in the process, including: 1) the rooted cycle graph, which corresponds to the root of the tree; 2) $g$ 
2PI $\ell=0$ schemes of genus one, which correspond to the leaves of the tree; and 3) $g-1$ rooted connected $\ell=0$ Feynman graphs of genus zero, which correspond to the inner vertices of the tree. One can further verify (see Figure \ref{fig:HLstructure} for instance) that these $g$ rooted connected $\ell=0$ Feynman graphs of genus zero either correspond to the rooted cycle graph or the melonic rooted Feynman graph with two standard vertices. In the second case, the root-vertex marks the edge that corresponds, on the tree associated with $S$, to the edge that belongs to the (only) path connecting the inner vertex to the root of the tree. From this point of view, one can deduce the four types of inner vertices or contact 6-point subgraphs. The first planar type shown on the left of Figure \ref{fig:tree_vertices} in the Proposition \ref{propo:dominant-schemes} is obtained from the rooted cycle graph, while the other contact types are obtained from the melonic rooted Feynman graph with two standard vertices by opening up the root-vertex (which is glued to the tree component containing the root), together with two out of its three remaining edges, as illustrated in the following figure:
\begin{center}
    \includegraphics[scale=.8]{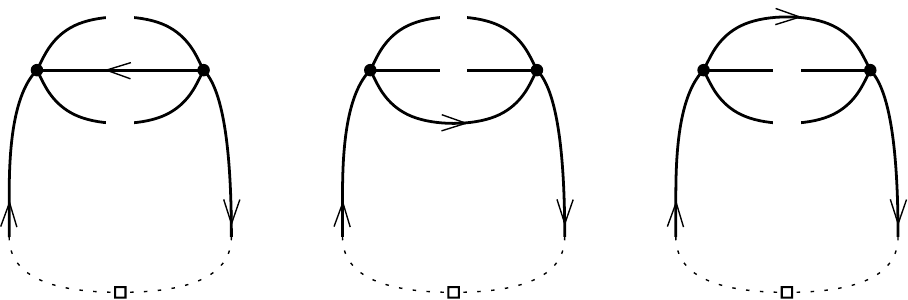}
\end{center}
By fixing the plane embedding of the tree and its inner vertices, we obtain the three remaining types of inner vertices in Figure \ref{fig:tree_vertices}. A simple way of distinguishing these three contact vertices is to examine what type of dipole is generated if one closes the two half-edges connected to the root onto themselves: from left to right, we obtain a N-, R-, or L-dipole, respectively. Likewise, for each of these cases, closing the pair of half-edges on the top-right or top-left corners yields two different types of dipoles. 

The last point shows that two branches of a tree emanating from a contact vertex can be unambiguously distinguished, and therefore ordered. Since this is also true for the planar vertex, which is by convention embedded in a clockwise manner, we conclude that our mapping is one-to-one provided that we work with plane trees. 
\end{proof}

\begin{figure}[htb]
\centering
\includegraphics[scale=.5]{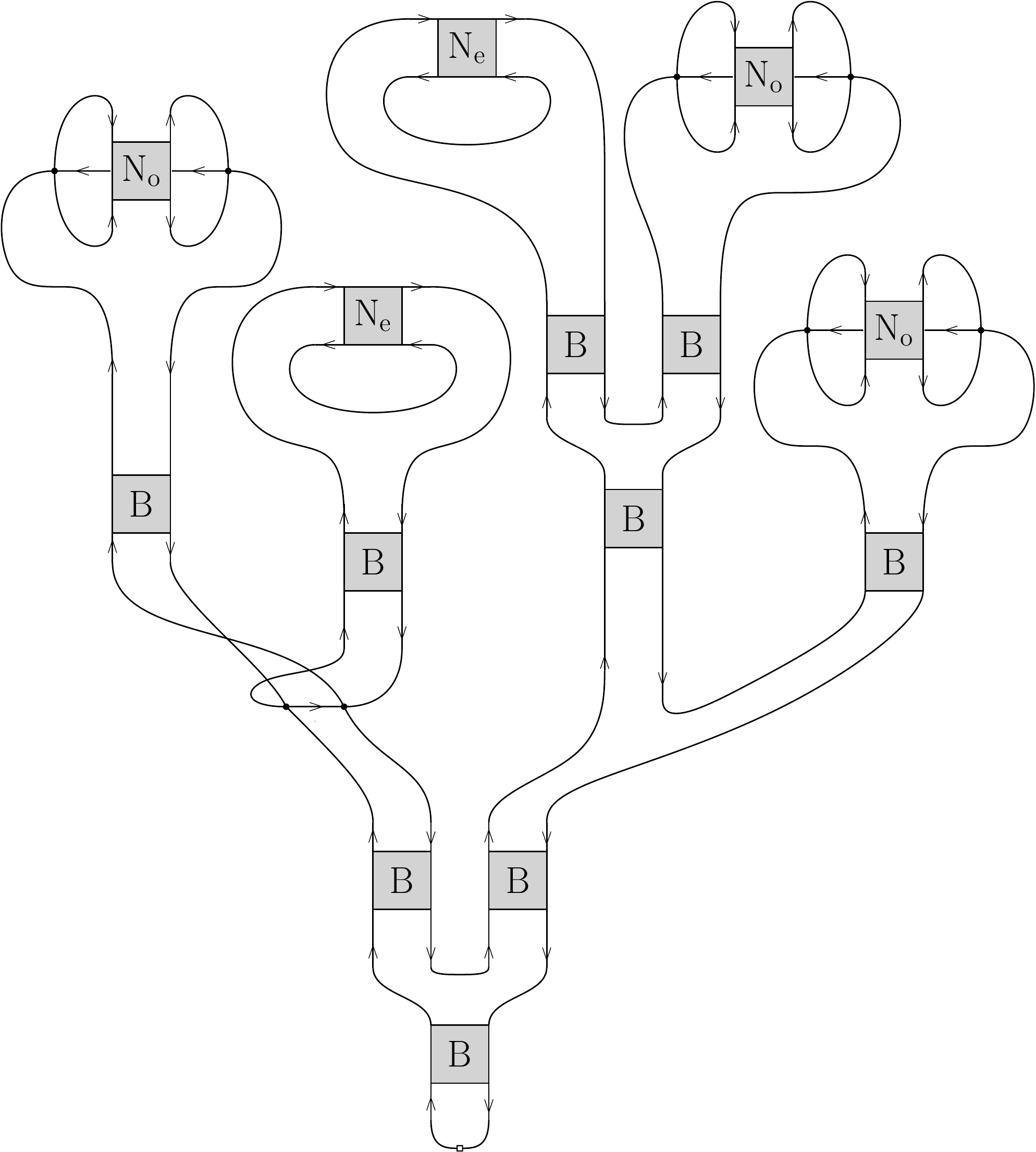}
\caption{\small A dominant scheme of genus $g=5$. It has the structure of a rooted binary tree with: $g=5$ leaves, $g-1 =4$ inner vertices, and $b=2g-1=9$ edges.}\label{fig:example_dominant}
\end{figure}

An example of plane binary tree associated with a dominant $\ell=0$ schemes of genus $g=5$ is given in Figure \ref{fig:example_dominant}. We note that the characterization of dominant schemes in our model is reminiscent of the one derived in \cite{GurSch} and \cite{Fusy:2014rba}. In \cite{GurSch}, the characterization is obtained in the context of colored tensor models in dimension three, which do not directly generate the same types of Feynman graphs (they in fact form a subfamily of the Feynman graphs studied in this paper). It is however similar to our characterization because in their case, the Gurau degree is an integer like the genus. On the other hand, the Feynman graphs of the tensor model studied in \cite{Fusy:2014rba} are the same as in our model. However, in that paper, the characterization of the dominant schemes is performed at fixed degree $\omega$, with no constraint on the grade $\ell$. In contrast, our double-scaling limit sets $\ell=0$, so that the degree $\omega$ reduces to the genus $g$ (see Eq.\ \eqref{eq:omega}). As a result, one can verify that all the dominant schemes identified in \cite{Fusy:2014rba} cannot contribute in our double-scaling limit and the relevant dominant schemes become the ones identified in Proposition \ref{propo:dominant-schemes}.

\subsection{Resummation and triple-scaling limit}

In this subsection, we first provide a direct resummation of the dominant $\ell=0$ schemes of fixed genus $g$. We then investigate a triple-scaling limit in which all genera can be resummed non-perturbatively.

From Eq.\ \eqref{eq:generating-scheme}, the generating function associated to a dominant $\ell=0$ scheme $S$ of genus $g\geq1$ takes the form: 
\be
\hat\cG_{S}(u)= u^{c_6+n_o} \cC_{\Ne}(u)^{n_e} \cC_{\No}(u)^{n_o} \cC_{\mathrm{B}}(u)^{2g-1} \, ,
\ee
where $c_6$ is the number of contact six-point functions, and we have used the fact that leaves with $\mathrm{N_o}$-vertices, as well as contact six-point functions, carry an extra factor of $u$ from their associated standard vertices.

As $\lambda$ approaches the critical value $\lambda_c$, the generating function of broken ladder-vertices picks up the following singular behavior, obtained by replacing $u$ with $U(\lambda) = T(\lambda)-1$ in \eqref{eq:nb} and using \eqref{eq:singT}:
\be
\cC_{\mathrm{B}}(U(\lambda)) \sim \left( \sqrt{\frac{8}{3}}\sqrt{1 - \frac{\lambda^2}{\lambda_c^2}} \right)^{-1}\,,
\ee
which is proportional to the inverse of the singular part of $T(\lambda)$. 
Therefore, the most singular contribution of $S$ to the sum \eqref{eq:sum_schemes_g} is:
\begin{align}
\cG_S^{\mathrm{(sing)}}(\lambda) &:= T(\lambda_c) {u_c}^{c_6+n_o} \cC_{\Ne}(u_c)^{n_e} \cC_{\No}(u_c)^{n_o} \left( \sqrt{\frac{8}{3}}\sqrt{1 - \frac{\lambda^2}{\lambda_c^2}} \right)^{1-2g} \\
&= {u_c}^{c_6+ 2 n_o} \frac{4}{3} \frac{1}{8^g} \left( \sqrt{\frac{8}{3}} \sqrt{1 - \frac{\lambda^2}{\lambda_c^2}} \right)^{1-2g}\,,
\end{align}
where we have used in the second equality that $n_e + n_o = g$.

Let us denote by $\cD_g$ the set of dominant schemes of genus $g$. Each element of $\cD_g$ is associated to a rooted binary plane tree with $g$ leaves. The number of such trees is given by the $(g-1)$-th Catalan number $\cT_g := \frac{1}{2g-1}\binom{2g-1}{g-1}$. Summing over all trees, we therefore obtain:
\begin{align}
\sum_{S \in \cD_g} \cG_S^{\mathrm{(sing)}}(\lambda) 
&= \frac{4}{3} \frac{\cT_g }{8^g} \left( \sqrt{\frac{8}{3}} \sqrt{1 - \frac{\lambda^2}{\lambda_c^2}} \right)^{1-2g} \left(1+ u_c^2 \right)^{g} \left(1+ 3 u_c \right)^{g-1} \\ 
&= \frac{2}{3} \sqrt{\frac{8}{3}} \cT_g \left( \frac{5}{48} \right)^g \left( \sqrt{1 - \frac{\lambda^2}{\lambda_c^2}} \right)^{1-2g}\,.
\end{align}
To justify the first line of this equation, first note that since each tree is planar, one can fix a canonical ordering of its vertices and leaves (for instance, by going around the tree in a clockwise manner starting from the root). The first three factors are common to all trees, while the last two depend on the decoration of its leaves and vertices. Each of the $g$ ordered leaves is either even (which brings no extra factor) or odd (which brings a factor $u_c^2$), and therefore contributes a factor $(1+ u_c^2)$. Similarly, each of the $g-1$ vertices is either planar (which brings no extra factor) or one of three contact interactions (each bringing a factor $u_c$), resulting in a factor $(1+3 u_c)$ per vertex.

Since the multi-matrix model in the double-scaling limit weighs any graph of genus $g$ with a factor $M^{2-2g}$ (see Eq.~\eqref{eq:2pt-1/M}), we can define a triple-scaling limit by sending $M$ to infinity and $\lambda$ to $\lambda_c$ while keeping the following ratio finite:
\be
\k^{-1} := M \left(1-\f{\l^2}{\l_c^2}\right)^{1/2} \,.
\ee
In this limit, we have
\be
\frac{\kappa}{M} \left( \cG^{(0)} (\lambda) - M^2 T(\lambda) \right) \sim \frac{2}{3} \sqrt{\frac{8}{3}} \sum_{g \geq 1} \cT_g \left( \frac{5}{48} \right)^g \kappa^{2g} =: \cD(\kappa) \,.
\ee
$\cD(\kappa)$ governs the deviation of the two-point function from its leading melonic behavior. A vanishing value of $\kappa$ corresponds to the purely melonic theory, obtained by first taking the $M \to + \infty$ limit, before sending $\lambda$ to its critical value. On the other hand, when $\kappa \neq 0$, $\cD(\kappa)$ does not vanish anymore, which means that the full two-point function deviates from $T(\lambda)$. Using the fact that $\underset{g\geq 1}{\sum} \cT_g x^g= x c(x)$, where $c(x) = (1- \sqrt{1 - 4x})/(2x)$
is the generating function of Catalan numbers, we can resum $\cD(\kappa)$ explicitly for $\kappa < \kappa_c := 2 \sqrt{\frac{3}{5}}$:
\be
\cD(\kappa) = \left(\frac{2}{3}\right)^{\f32} \left( 1 -  \sqrt{1 - \frac{5}{12} \kappa^2}\right)\,.
\ee
Near the critical value $\kappa_c$, the fluctuations are dominated by schemes of unbounded genus, or equivalently by large trees. In particular, the expectation value of the genus in the ensemble defined by $\cD(\kappa)$ diverges at the critical value:
\be
\langle g \rangle = \frac{1}{2} \kappa \partial_\kappa \ln\cD(\kappa) \simeq \frac{1}{2 \sqrt{1 - \kappa^2 / \kappa_c^2 } }\,.
\ee

We conclude this section by noting that the triple-scaling regime we have just derived shares its main features with the double-scaling limit of the multi-orientable tensor model \cite{Gurau:2015tua}. This should not overly surprise us, owing to the close combinatorial similarities between both models.

\section{2PI generating function}
\label{sec:2PI}

In the previous section, we focused on dominant $\ell=0$ schemes because they determine the critical behavior of the generating function for $\ell=0$ Feynman graphs of genus $g$. In particular, we have demonstrated that this class of $\ell=0$ schemes can be mapped to rooted binary trees. As a result, our model converges in the continuum limit to a branched-polymer phase \cite{Ambjorn:1990wp, Bialas:1996ya}, the same universality class of random geometry found in the standard melonic regime of tensor models \cite{melbp}. 

In this section, we study another class of $\ell=0$ schemes, which exhibits a richer structure and thus leads to a richer continuum limit. In view of the pole structure of the generating functions of ladder-vertices constructed in section \ref{sec:2point}, it is tempting to construct a model dominated by a critical value $u_c = 1$, instead of $u_c =1/3$. This can be achieved by making sure that neither melon subgraphs nor B-ladders can contribute to the $\ell = 0$ sector. Interestingly, by Lemma~\ref{lem:schemes2Pi}, Proposition \ref{propo:g1} and Theorem \ref{thm:induction}, both of these structures have in common that they can only occur in 2PR $\ell = 0$ graphs. Hence, we can eliminate them by restricting the sum over Feynman graphs to 2PI contributions. 

The generating function of 2PI graphs can be related to the original multi-matrix model as follows.
First, we modify the action \eqref{eq:actionND} by introducing a quadratic ``counterterm'', or a source for the quadratic invariant:
\be \label{eq:actionND-2PI}
S[X,X^\dagger;m] = ND \Bigl( (1-m) \Tr\bigl[ X^\dagger_\m X_\m \bigr] - \frac{\l}{2}\sqrt{D} \, \Tr\bigl[ X^\dagger_\m X_\n X^\dagger_\m X_\n \bigr] \Bigr)\,.
\ee
Next, we choose the parameter $m=m(\l)$, with $m(0)=0$, such that the full two-point function of the modified model is $\l$-independent, and in particular it coincides with the two-point function of the free theory (with $m=\l=0$):
\begin{equation} \label{eq:2PI-cond}
\left\la \Tr\left[ X^\dagger_\m X_\m \right]\right\ra_{m(\l)}
 = N \,, 
\end{equation}
where
\be \label{eq:2pt_m-def}
\left\la \Tr\left[ X^\dagger_\m X_\m \right]\right\ra_m 
 = \frac{ \int [dX] \, e^{-S[X,X^\dagger;m]}\Tr\left[ X^\dagger_\m X_\m \right] }{\int [dX] \, e^{-S[X,X^\dagger;m]}} = \f{1}{ND} \f{\p\cF(\l;m)}{\p m}\,,
\ee 
and we introduced also a modified free energy $\cF(\l;m)$, defined analogously to Eq.~\eqref{eq:free-en-def}.
We claim that $m(\l)$ is the generating function of rooted 2PI vacuum graphs  of \eqref{eq:actionND}, with free propagators on the edges.
In order to see that, we can follow the formalism of the 2PI effective action (see for example \cite{Cornwall:1974vz,Berges:2004yj}, and \cite{Benedetti:2018goh} for its application in tensor models). We perform a Legendre transform of the modified free energy with respect to $m$, thus defining an effective action $\G(\l;G)$: 
\begin{equation} \label{eq:2PI-Legendre}
    \G(\l;G) = - \cF(\l;\tilde{m}(\l;G)) +N^2 D\, G \, \tilde{m}(\l;G) \,,
\end{equation}
where $\tilde{m}(\l;G)$ is defined as the solution of
\be
 \f{\p\cF(\l;m)}{\p m}= N^2 D\,G \,.
\ee
Notice that this is equivalent to choosing $m$ such that $\left\la \Tr\left[ X^\dagger_\m X_\m \right]\right\ra_m =N G$.

By deriving \eqref{eq:2PI-Legendre} with respect to $G$ we also find:
\be 
 \f{\p\G(\l;G)}{\p G}= N^2 D\, \tilde{m}(\l;G) \,,
\ee
while following a computation similar to the one in \cite{Cornwall:1974vz,Berges:2004yj,Benedetti:2018goh}, the effective action takes the form:
\be
\G(\l;G) = N^2 D\,( G -  \ln G ) + \G^{\mathrm{2PI}}(\l;G) \,,
\ee
where $\G^{\mathrm{2PI}}(\l;G)$ is the sum of vacuum 2PI Feynman diagrams of the action \eqref{eq:actionND-2PI} with the substitution $(1-m)\to G^{-1}$.
Combining the last two equations we obtain:
\begin{equation}
   1 - G^{-1} +\f{1}{N^2 D} \f{\p\G^{\mathrm{2PI}}(\l;G)}{\p G} = \tilde{m}(\l;G) \,.
\end{equation}
Setting $G=1$ and $\tilde{m}(\l;1)\equiv m(\l)$, we obtain the claimed result, because deriving $\G^{\mathrm{2PI}}(\l;G)$ with respect to $G$ is equivalent to marking an edge in its diagrams. 

In the double-scaling limit,  similarly to Sec.~\ref{sec:2point}, we define the generating function of rooted 2PI Feynman graphs as:
\be
\cG^{(0)}_{\mathrm{2PI}}(\l) \equiv \lim_{ \substack{N,D\to\infty \\ M<\infty} } \f{N^2}{D} m(\l) = \sum_{g\in \mathbb{N}} \cG_{g}^{\mathrm{2PI}}(\lambda) M^{2-2g}\,,%
 \label{}
\ee
where
\be \label{eq:2PI-G_g}
\cG_{g}^{\mathrm{2PI}}(\lambda) = \sum_{S \in \cS_{g}^\mathrm{2PI}}  \, \hat\cG_{S}(\lambda^2) \, ,
\ee
in which $\cS_g^{\mathrm{2PI}}$ is the set of 2PI schemes of  genus $g$ and vanishing grade. 
Note that, since melon subgraphs do not contribute, propagators have weight $1$ rather than $T(\lambda)$, and consistently, the counting variable $u$ has been substituted by $\lambda^2$ instead of $\lambda^2 T(\lambda)^4$. 

By construction, the schemes that determine the singular behavior of $\cG_g^{\mathrm{2PI}}(\lambda)$ are distinct from the dominant schemes defined in Section \ref{sec:dominantSchemes}, since the latter are all 2PR. 
But we can reason similarly. The only ladder-vertices allowed in $\ell = 0$ 2PI schemes are of type $\Ne$ and $\No$, and from equations \eqref{eq:ne} and \eqref{eq:no}, both have a dominant simple pole at $u=1$. Hence, the genus $g$ and $\ell = 0$ schemes that govern the most singular part of $\cG_g^{\mathrm{2PI}}(\lambda)$ are those that maximize the number of N-vertices. In the same spirit as in Section \ref{sec:dominantSchemes}, we will therefore say that a 2PI scheme is \emph{2PI-dominant} if it contains a maximal number of N-vertices allowed by its genus.

We now proceed with the combinatorial characterization of this new family of schemes, before analyzing their properties in the continuum limit. 

\subsection{2PI-dominant schemes}\label{sec:combinatorics_2PI}

As a first step, we want to determine the maximal number of N-vertices in a 2PI $\ell=0$ scheme as a function of its genus. Remark that because we deal with 2PI $\ell=0$ schemes, they cannot contain any separating ladder-vertex (they would be 2PR otherwise) and therefore, by Lemma \ref{lem:sep-conn}, any ladder-vertex corresponds to a connecting N-vertex. We denote by $n(S)=n_e(S)+n_o(S)$ the number of N-vertices (even or odd) in a given 2PI $\ell=0$ scheme $S$. Besides, we say that two N-vertices are \emph{separated by the root-vertex} if they are in one of the two configurations shown in Figure \ref{fig:separated-root}.
\begin{figure}[htb]
\centering
\includegraphics[scale=.6]{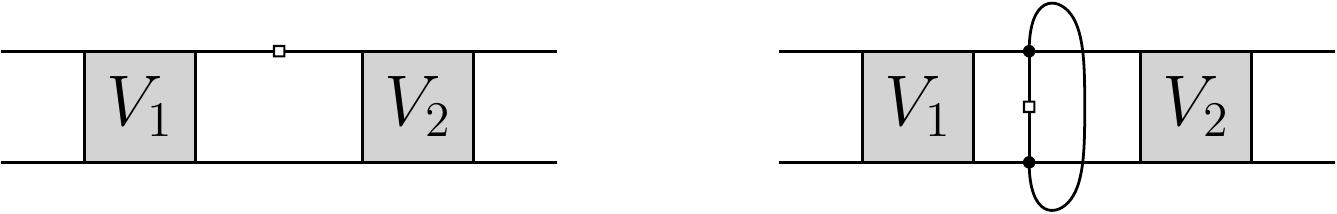}
\caption{\small Two N-vertices $V_1$ and $V_2$ \emph{separated by the root-vertex}.}\label{fig:separated-root}
\end{figure}

\begin{lemma}\label{lem:bound-cycles}
Let $S$ be a 2PI $\ell=0$ scheme of genus $g\geq1$ with $n$ N-vertices. Then, $n\leq3g-2$. Furthermore, if $g>1$ and no two N-vertices are separated by the root-vertex, then $n \leq 3g -3$.
\end{lemma}
\begin{proof}
All the 2PI $\ell=0$ schemes of genus $g=1$ are identified in Proposition \ref{propo:g1}. In particular, we see explicitly that $n=1$.

Let us therefore assume that $g>1$. To prove the upper bounds on the number $n$ of N-vertices in $S$, it is convenient to adopt a topological point of view, and regard the N-vertices in $S$ as (equivalence classes of) simple loops drawn on the (discretized) Riemann surface represented by $S$  (see for instance \cite{hatcher1980presentation, mohar2001graphs}). This point of view was already advocated in Remark \ref{rem:surfaces} for N-dipoles, and naturally extends to N-vertices. Indeed, the latter correspond to ladders of N-dipoles, and all the length-two $\oD$-loops in a N-ladder belong to the same homotopy class. In this topological language, we prove that N-vertices in $S$, regarded as simple loops, are:
\begin{enumerate}
\item pairwise disjoint;
\item non-separating (hence, non-contractible);
\item pairwise non-homotopic, unless possibly if they are separated by the root-vertex.
\end{enumerate}

The first property is ensured by the fact that ladder-vertices in $S$ represent maximal ladders, and maximal ladders are pairwise vertex-disjoint (see Claim \ref{claim:vertex-disjoint}).

The second property is a direct consequence of the fact that $S$ is a 2PI $\ell=0$ scheme and therefore only contains non-separating (more precisely connecting) N-vertices. The contraction of any N-vertex in $S$ thus preserves the connectedness of the corresponding discretized Riemann surface. Besides, a non-separating simple loop is necessarily non-contractible.

To prove the third property, let us assume that there exist two homotopic N-vertices $V_1$ and $V_2$ in $S$. Since any N-vertex in $S$ is non-separating, it implies that $S$ has the structure represented on the left panel of Figure \ref{fig:non-homotopic}, where $\tilde{H}_0$ and $\tilde{H}_1$ are $4$-point subgraphs and one of them (say $\tilde{H}_0$) has the topology of a $2$-sphere (with two punctures). Contracting $V_1$ and $V_2$ then yields two $\ell=0$ Feynman graphs with ladder-vertices: $H_0$ of genus $0$, and $H_1$ of genus $g-1$. If $\tilde{H}_0$ does not contain the root-vertex of $S$, then $H_0$ must be melon-free. Indeed, contracting (say) $V_1$ first cannot create a melonic subgraph; otherwise, $S$ would contain a ladder. Then, contracting $V_2$ does not create a melonic subgraph in $H_0$ either, because a root-vertex is added in the middle of the edge that closes onto $H_0$ after the contraction of $V_2$.
As a result, $H_0$ must be the rooted cycle graph. But then, it means that $V_1$ and $V_2$ form a ladder in $S$, which yields a contradiction. 
On the other hand, if $\tilde{H}_0$ contains the root-vertex of $S$, we have two cases. 
If $H_0$ is melon-free, it must correspond to the rooted cycle graph. But then, $V_1$ and $V_2$ are separated by the root-vertex in $S$ (case on the left of Figure \ref{fig:separated-root}). 
Otherwise, $H_0$ is a melonic rooted Feynman graph with at least two standard vertices. It is straightforward to check that if $H_0$ contains more than two standard vertices, then $S$ must either be 2PR or contain a ladder, leading to a contradiction.
$H_0$ must therefore be the melonic rooted Feynman graph with two vertices, meaning once again that $V_1$ and $V_2$ are separated by the root-vertex (case on the right of Figure \ref{fig:separated-root}).
\begin{figure}[htb]
\centering
\includegraphics[scale=.5]{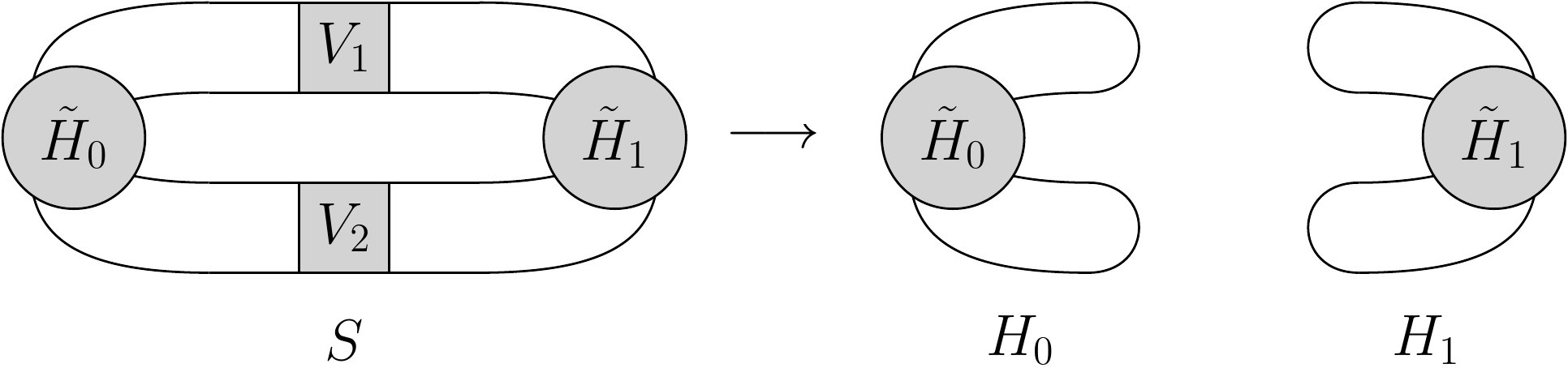}
\caption{\small A 2PI $\ell=0$ scheme with two homotopic N-vertices $V_1$ and $V_2$ (left panel): after contraction of $V_1$ and $V_2$ (right panel), one of the two resulting Feynman graphs (say $H_0$) must have vanishing genus.}
\label{fig:non-homotopic}
\end{figure}

Now, it is well known that, on an orientable Riemann surface of genus $g>1$, one can draw a maximum of $3g - 3$ pairwise disjoint, non-contractible and pairwise non-homotopic simple loops (see e.g. Proposition 4.2.6.\ in \cite{mohar2001graphs}). Since at most one pair of N-vertices can be separated by the root-vertex, in which case they are homotopic, this leads to $n \leq 3g - 3 + 1 = 3g -2$.     
\end{proof}

\begin{figure}[htb]
\centering
\includegraphics[scale=.5]{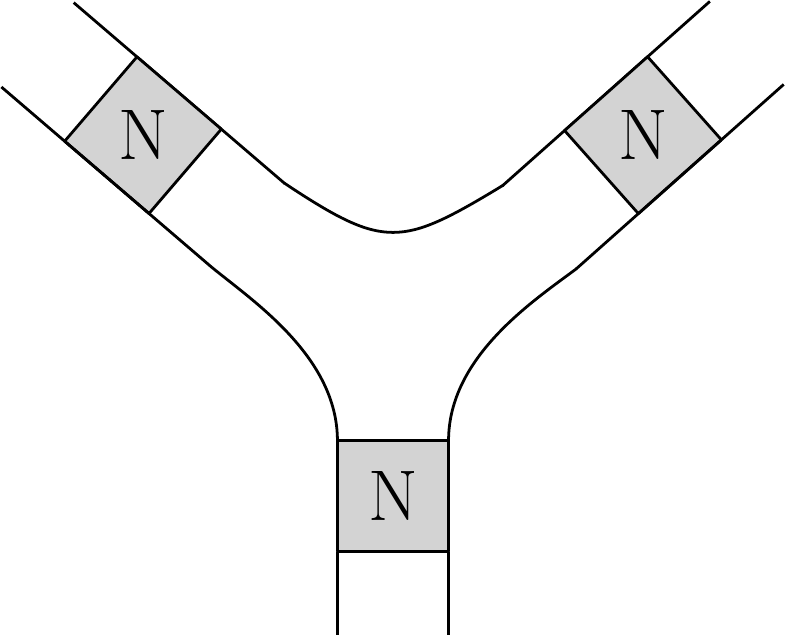}
\caption{\small Combinatorial structure of a planar 6-point subgraph connecting a triplet of N-vertices in a 2PI-dominant scheme. Depending on the type of N-vertices meeting at such a vertex ($\Ne$ or $\No$), twists may need to be added to respect our embedding conventions.}\label{fig:6-point}
\end{figure}

 One can prove that the bounds of Lemma \ref{lem:bound-cycles} are tight, and as a result, that the 2PI-dominant schemes are characterized by $n(S) = 3 g(S) -2$. We postpone this discussion to the next subsection, which will provide an explicit description of all 2PI-dominant schemes.

\begin{proposition}\label{propo:dominant-2PI}
Let $S$ be a 2PI $\ell=0$ scheme of genus $g>1$. $S$ is 2PI-dominant (i.e.\ $n(S) = 3g -2$) if and only if it has the following structure:
\begin{itemize}
    \item two N-vertices in $S$ are separated by the root-vertex, as in Figure \ref{fig:separated-root};
    \item the remaining N-vertices in $S$ are connected by the planar 6-point subgraph shown in Figure \ref{fig:6-point}. 
\end{itemize}
\end{proposition}
\begin{proof}
The first condition follows from Lemma \ref{lem:bound-cycles}: if no two N-vertices are separated by the root-vertex, then $n(S) = 3g - 3 < 3g -2$ and therefore, $S$ cannot be 2PI dominant.

Now that we have dealt with the root-vertex, it is convenient to replace by a single N-vertex the subgraph made out of: the two homotopic N-vertices separated by the root-vertex, the root-vertex itself, and possibly the two standard vertices the latter is connected to. 
In other words, we perform one of the following two replacements:\footnote{To avoid a tedious enumeration of cases, we keep the edge orientations and N-vertex parities implicit, but we note that there is always a unique substitution consistent with these features. For instance, in the situation shown on the right, the sum of the parities of the two N-vertices in the top figure must be opposite to the parity of the N-vertex in the bottom figure.}
\begin{center}
\includegraphics[scale=.6]{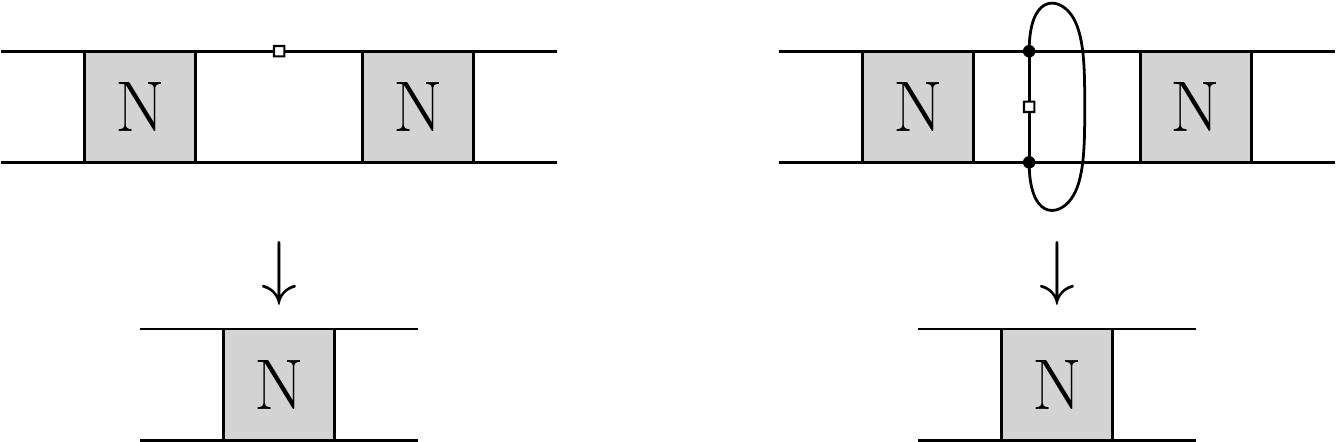}
\end{center}
As a result, we obtain a non-rooted embedded graph $S'$ of genus $g$, with $3g-3$ pairwise disjoint, non-separating and pairwise non-homotopic N-vertices, regarded as simple loops. In topological terms, because we have a maximal number of such loops drawn on the discretized (orientable) Riemann surface represented by $S'$, it implies that (see again Proposition 4.2.6.\ in \cite{mohar2001graphs}): cutting $S'$ along these $3g-3$ loops yields $2g-2$ connected components $c_1\,, \ldots \,, c_{2g-2}$, which all have genus zero and exactly three boundaries. In other words, these loops correspond to the \emph{cuffs} of a \emph{pants decomposition} of $S'$ \cite{hatcher1980presentation}. Translated in terms of the corresponding N-vertices in $S'$, it means that: contracting the $3g-3$ N-vertices in $S'$ gives rise to $2g-2$ connected components $C_1\,, \ldots \,, C_{2g-2}$ such that $g(C_i)=0 \; (i=1,\ldots,2g-2)$ and each connected component originates from a $6$-point subgraph connected to three distinct N-vertices in $S'$.
For instance, at genus $2$ or $3$, we are in one of the situations illustrated in Figure \ref{fig:decomposition_S} (the genus $2$ structure is unique, but there are more possibilities at higher genus).
Besides, since $\ell(S')=0$, these connected components $C_i$ also satisfy $\ell(C_i)=0$ for $i=1,\ldots,2g-2$. Hence, they must correspond to (non-rooted) melonic Feynman graphs. We now prove that they are in fact (non-rooted) cyclic graphs, which will achieve the proof.

\begin{figure}[htb]
    \centering
    \includegraphics[scale=0.7]{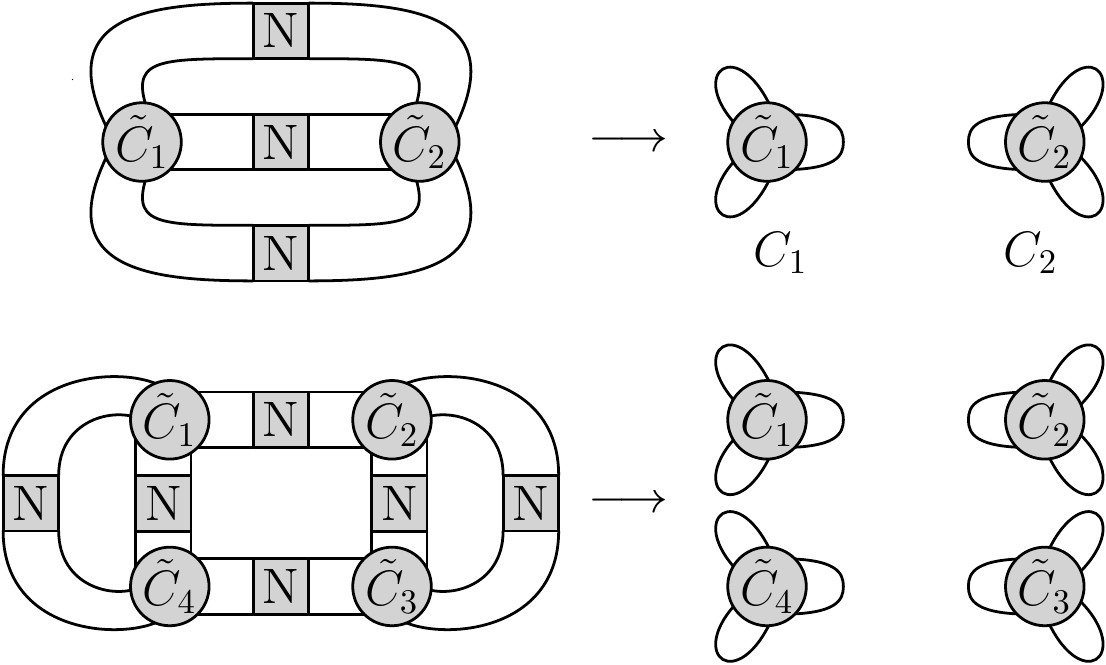}
    \caption{\small Examples of decomposition of the graph $S'$, with $g(S')=2$ (top) or $g'(S)=3$ (bottom). In both cases, we are left with $2g(S')-2$ connected components after all N-dipoles have been contracted. Note that more than one structure is allowed at genus $g>2$.}
    \label{fig:decomposition_S}
\end{figure}

Suppose that $C_i$ ($i=1,\ldots,2g-2$) is an arbitrary melonic Feynman graph. One can recover the $6$-point subgraph $\tilde{C}_i$ it originates from in $S'$ by performing three cuts on its edges (note that the same edge can be cut multiple times). Using the fact that the resulting $6$-point subgraph cannot contain a melonic subgraph (otherwise, $S'$ and $S$ would not be melon-free) nor a dipole (otherwise, $S'$ and $S$ would not be ladder-free), one can check that $C_i$ must be:
\begin{enumerate}[label=(\roman*)]
    \item the cycle graph, and in that case the corresponding $6$-point subgraph $\tilde{C}_i$ is the one represented in Figure \ref{fig:6-point};
    
    \item or the melonic Feynman graph with two standard vertices, and in that case the corresponding $6$-point subgraph $\tilde{C}_i$ is obtained by cutting three distinct edges.
\end{enumerate}
These configurations have already been encountered in Proposition \ref{sec:dominantSchemes}: situation (i) yields a planar 6-point interaction, while (ii) gives rise to a contact vertex. The only difference is that, in a 2PI-dominant scheme, they must connect three connecting N-vertices. It is straightforward to see that this condition disallows contact vertices (situation (ii)).
An example is given in Figure \ref{fig:6pointNotAllowed}, where one observes that the N-vertex $V_2$ is not connecting. We are therefore left with the planar $6$-point subgraph of Figure \ref{fig:6-point}, which concludes the proof.
\begin{figure}[htb]
    \centering
    \includegraphics[scale=0.7]{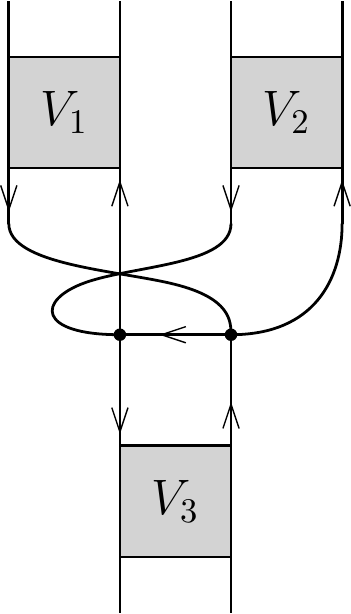}
    \caption{\small $6$-point subgraph obtained by cutting three distinct edges in a melonic Feynman graph with two vertices. It connects three N-vertices $V_1, V_2$ and $V_3$. However, at least one of them (here $V_2$) cannot be connecting because of the structure of the $\oD$-loops.}
    \label{fig:6pointNotAllowed}
\end{figure}
\end{proof}

An example of 2PI-dominant scheme of genus $3$ is provided in Figure \ref{fig:example_2PI_dominant}.
\begin{figure}[htb]
\centering
\includegraphics[scale=.5]{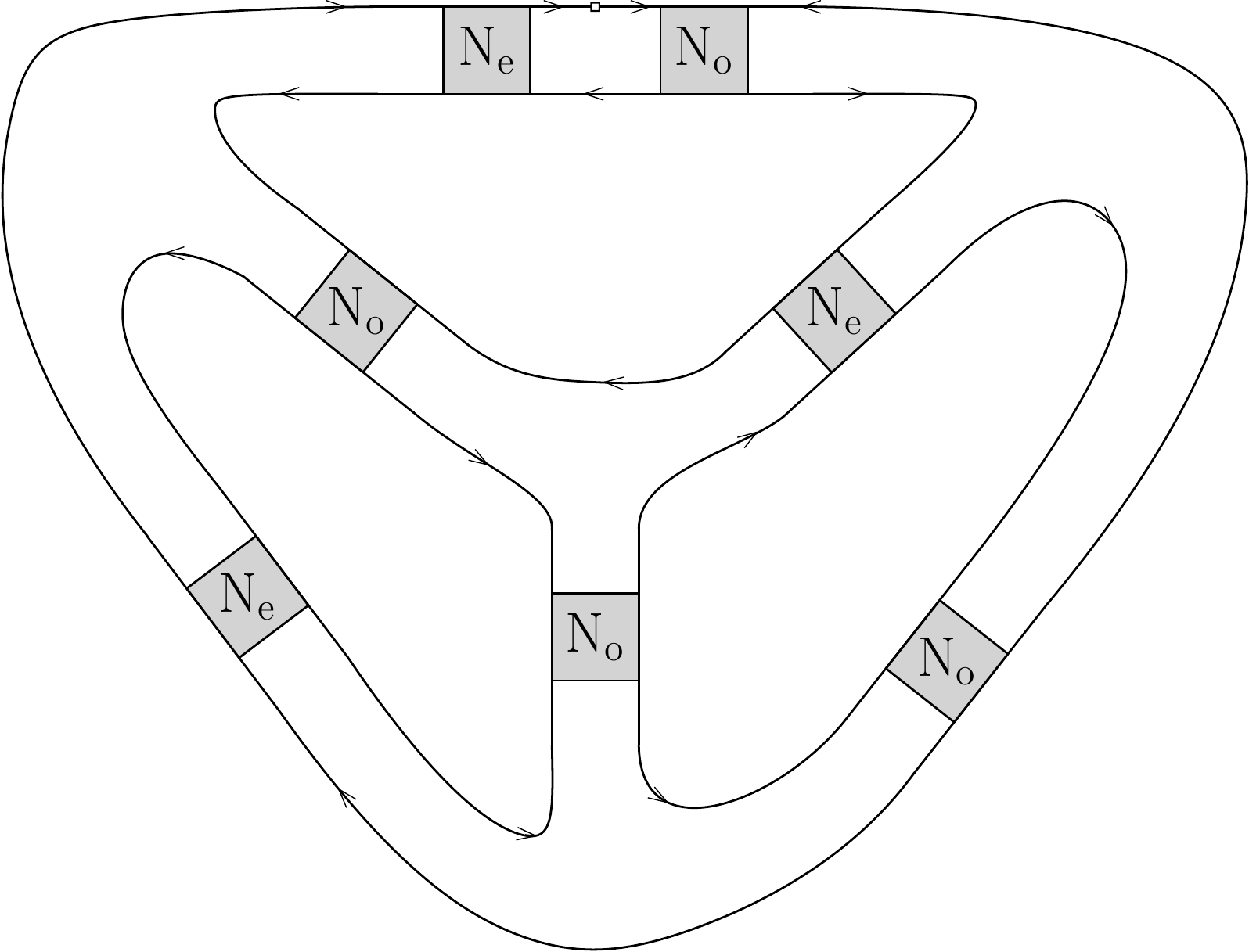}
\caption{\small A 2PI-dominant scheme of genus $g=3$. It has $3g-2 = 7$ N-vertices; all of them are pairwise non-homotopic except for the two N-vertices separated by the root.}\label{fig:example_2PI_dominant}
\end{figure}

\subsection{2PI-dominant schemes as Ising states on planar cubic maps}
\label{sec:Ising}

We are now going to describe a class of graphs that is a subset of all the graphs of our model, but which contains all the 2PI-dominant graphs. We will at first be slightly less rigorous and consider unrooted graphs. 

Let us view a N-ladder of a graph (or a N-vertex of a scheme) as a decorated edge of an auxiliary ribbon graph with  vertices of arbitrary order, corresponding to planar $2n$-point functions that generalize the one of Fig.~\ref{fig:6-point}.
Figure \ref{fig:example_2PI_dominant}, with the root removed as in the proof of Proposition \ref{propo:dominant-2PI}, provides an example of such a ribbon graph, with only tri-valent vertices. Notice that in Fig.~\ref{fig:example_2PI_dominant}, due to the arrows, the vertices have two possible orientations, clockwise and anti-clockwise, with $\Ne$-ladders joining same-orientation vertices and $\No$-ladders joining opposite-orientation vertices; in general this is only true for ribbon graphs without twisted edges, as otherwise we can for example join same-orientation vertices by means of a twisted $\No$-ladder. For a canonical description of all possible ribbon graphs, including non-orientable ones, it is convenient to choose the same orientation for all the vertices and allow twists of the edges; however, in order to more easily recognize planar graphs it is actually more convenient to allow both orientations of the vertices, so that planar graphs with $\No$-ladders can actually be drawn on a plane, as in Fig.~\ref{fig:example_2PI_dominant}.

A ladder $i$ has $r_i\geq 2$ rungs, and therefore the original graph corresponding to a given auxiliary graph has in total $\rho=\sum_{i\in \text{edges}} r_i$ rungs (or N-dipoles).
Remembering the structure of external faces of the N-ladders (see Fig.~\ref{fig:0chainvertexes}) and counting the number of their internal faces, we find the following mapping between the number of vertices, edges, and faces of the ribbon graph, denoted by  $\nu$, $\epsilon$, and $\phi$, respectively, and the number of faces, vertices and N-ladders in the original graph: 
\be
v = 2 \rho\,, \;\;\;  \eps=n\,, \;\;\; f_L = f_R = \rho -\epsilon + \nu \,, \;\;\; \varphi = \rho + \phi \,.
\ee
Consequently, we have also the following mapping:
\be \label{eq:grade=genus}
\ell = 4+4\rho -2 (\rho-\eps)-2\nu - 2\phi - 2\rho = 4 - 2\nu+2\eps-2\phi = 4 \eta\,,
\ee
\be \label{eq:genus=loops}
g = \eps - \nu + 1 = \m \,,
\ee
where $\eta$ is the genus of the ribbon graph, and $\m$ is its cyclomatic number (the  number of independent loops). 
For non-orientable surfaces, we have $\eta\in \f12 \mathbb{N}$, hence these graphs can only have even $\ell$: they are a strict subset of all the possible graphs, which have $\ell\in \f12 \mathbb{N}$.

It is however interesting that if we take all the planar ribbon graphs and decorate them by ladders as described above, we have a class of Feynman graphs of our multi-matrix model with $\ell=0$ and arbitrary $g$ (given by the number of independent loops of the planar ribbon graph).
Again, these are clearly not all the graphs with $\ell=0$ and arbitrary $g$, as made evident for example by Proposition~\ref{propo:g1}.\footnote{Notice that to actually capture at least the first graph ($S_1$) of Proposition~\ref{propo:g1}, we should allow for exactly one two-valent vertex in the ribbon graphs, corresponding to the root. From the point of view of the matrix model discussed in the following subsection, this would amount to studying the two-point function rather than the free energy.}
However, from Proposition~\ref{propo:dominant-2PI} we know that 2PI-dominant schemes contain only N-ladders, and that these are connected via planar six-point functions. Therefore, the 2PI-dominant schemes are contained in the class of graphs we have just described, by restriction to \emph{one-particle irreducible} (1PI) three-valent ribbon graphs. 
It is easy to verify that such graphs indeed saturate the bound of Lemma~\ref{lem:bound-cycles}.
In fact, we have $2\eps=3\nu$, and therefore from Eq.~\eqref{eq:genus=loops} we have $n(G)=\eps = 3g-3$, which is the right number of N-vertices in a dominant 2PI scheme, once we remove the root and join its two adjacent N-vertices.

\medskip

More rigorously, we can in fact construct an explicit bijection between 2PI-dominant schemes on the one hand, and Ising states on a certain family of rooted planar maps on the other hand. 

Let $S$ be a 2PI-dominant scheme of genus $g$. Since $\ell(S) = 0$, its underlying ribbon diagram has genus $\eta = 0$. We can therefore choose a planar embedding of the latter (i.e. without crossing or twist). By convention, we can furthermore require that the root is embedded in the plane in such a way that its R-face appears in between its two associated rungs, as shown in the first line of Figure \ref{fig:cubic_map}. 
By planarity, this completely fixes the local embedding of the other elements of the graph. Because each side of a N-vertex has exactly one incoming and one outgoing edge, there are exactly two possible embeddings of the vertices: clockwise, which we label $+$, and anti-clockwise, which we label $-$. Next, one finds two possible embeddings of the $\Ne$-vertices: those that connect two $+$ vertices, and those that connect two $-$ vertices. There is finally a unique embedding of $\No$-vertices, each such vertex always connecting two vertices of opposite orientations. 

As a result, we realize that $S$ encodes a unique Ising state on a planar map, obtained by performing the substitutions illustrated in Figure \ref{fig:cubic_map}. We emphasize that we have replaced the two rails associated to the root (together with its associated rung, when applicable) by two half-edges, connected to two univalent vertices. One of these vertices represents the R-face adjacent to the root in the initial scheme, and is distinguished by an outward-pointing arrow as the root-vertex in the planar map. By consistency with our construction, this vertex can only be in the $+$ state. Finally, we note that the 2PI character of the initial scheme translates into a 1PI condition in the colored map representation.

Examples of maps obtained in this manner are represented in Figure \ref{fig:ex_map}. 

\begin{figure}[H]
\centering
\includegraphics[scale=.8]{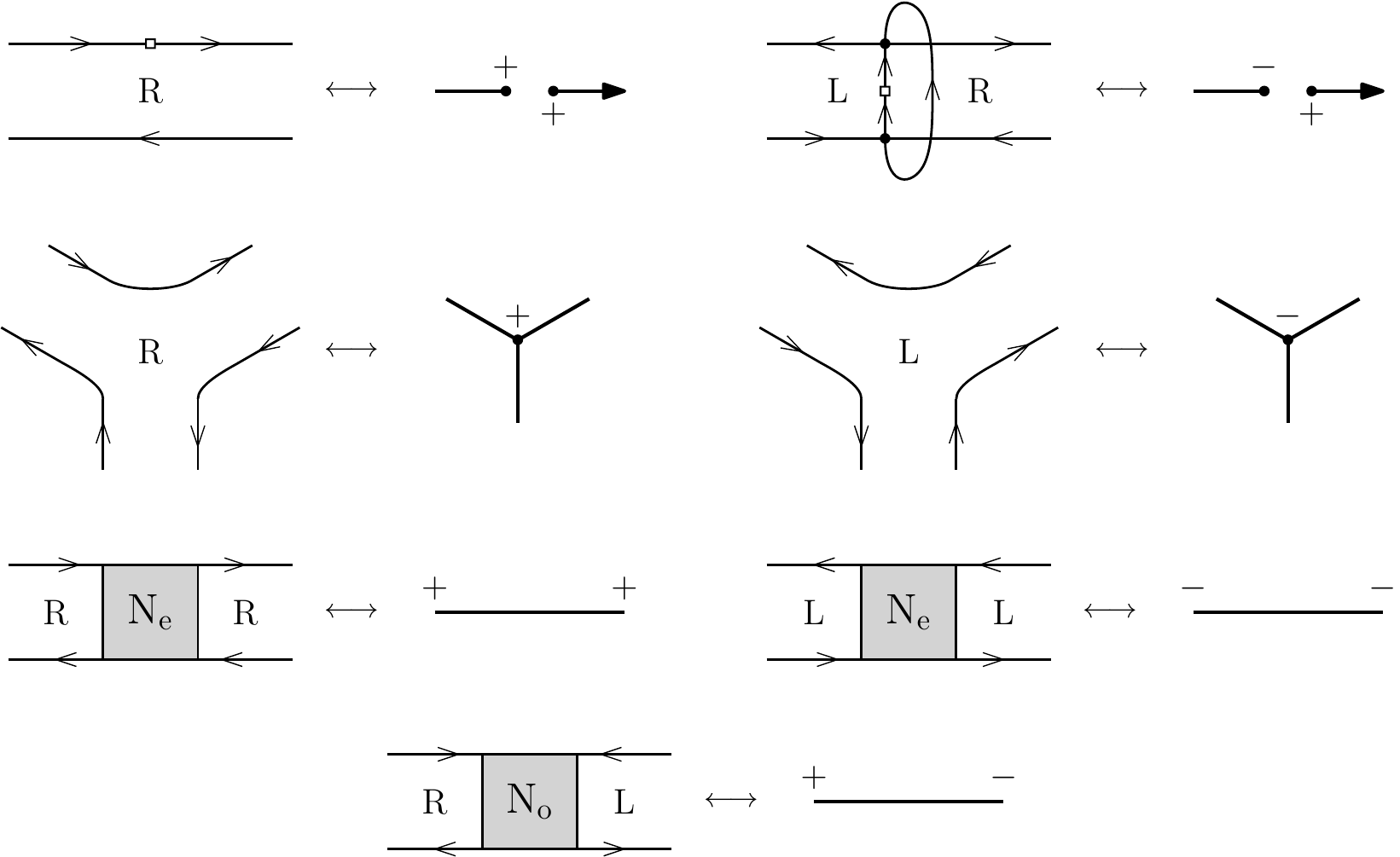}
\caption{\small Mapping between 2PI-dominant schemes and rooted planar maps. There are two types of root-edges (with one canonically distinguished strand, specifying a root-face), which we conveniently represent as pairs of univalent vertices; two types of 3-valent vertices, clockwise (+) and anti-clockwise (-); and three types of edges, distinguished by the signs of the vertices they connect.}\label{fig:cubic_map}
\end{figure}

To summarize, the decorated map $\mathcal{T}(S)$ we associate to a 2PI-dominant scheme $S$ is unique, and has the following properties:
\begin{enumerate}[label=(\roman*)]
    \item it is planar;
    \item it has $3g(S)-2$ edges;
    \item it has $2 g(S)$ vertices, each of which is decorated by a spin label $s \in \{ + , -\}$;
    \item $2 g(S) -2$ vertices are $3$-valent;
    \item $2$ vertices are univalent, one of which being distinguished has the root-vertex;
    \item the root-vertex is in the + state;
    \item the rooted planar map $\tilde{\cT}(S)$ obtained by joining the two univalent vertices (and keeping the arrow to specify the root-face) is 1PI.
\end{enumerate}
Reciprocally, it is easy to see that any decorated map verifying these conditions allows to reconstruct a unique 2PI-dominant scheme. The correspondence between 2PI-dominant schemes and rooted planar maps we have just described is therefore bijective. Finally, in the following, we will denote by $\cT_{++}$ (resp. $\cT_{+-}$) the set of decorated maps with boundary condition $++$ (resp. $+-$).

\begin{figure}[htb]
\centering
\includegraphics[scale=.6]{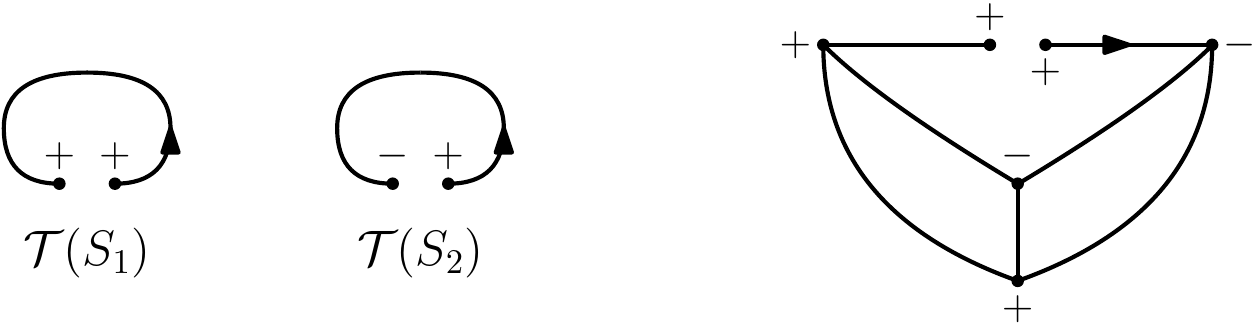}
\caption{\small Decorated rooted planar maps associated to: $S_1$ and $S_2$ of Proposition \ref{propo:g1} (left); the 2PI-dominant scheme of Figure \ref{fig:example_2PI_dominant} (right).}\label{fig:ex_map}
\end{figure}

\medskip

The previous construction is particularly interesting, because it allows to understand the 2PI-dominant schemes as encoding an Ising model on a family of random planar surfaces. To see this, let us introduce a generating function for the 2PI-dominant schemes of genus $g$ ($g\geq1$) 
\be
\cD_{g}^{\mathrm{2PI}}(\lambda) := \sum_{S \;\mathrm{2PI-dominant}}  \, \hat\cG_{S}(\lambda^2) \,,
\ee
and, similarly, a generating function for all 2PI-dominant schemes:
\be  \label{eq:D-2PI-dominant}
\cD^{\mathrm{2PI}}(\lambda, M) := \sum_{g \in \mathbb{N}} M^{2-2g} \, \cD_{g}^{\mathrm{2PI}}(\lambda)\,.
\ee
We can also introduce the two Ising generating functions:
\be
Z_{++} (t, x) := \sum_{T \in \cT_{++}} t^{\eps(T)} x^{m(T)}\,, \quad 
Z_{+-} (t, x) := \sum_{T \in \cT_{+-}} t^{\eps(T)} x^{m(T)}\,, 
\ee
where $\eps(T)$ is the total number of edges in a map $T$,\footnote{That is, $3g-2$, where $g$ is the genus of the scheme represented by $T$.} and $m(T)$ the number of \emph{monochromatic edges}.\footnote{Monochromatic edges are edges of type $++$ or $--$, and correspond to $\Ne$-vertices in the original scheme representation.} Notice that these are grand canonical partition functions, with $x=e^{2\beta}$ and $t=z e^{-2\beta}$, where $\beta$ is the inverse temperature and $z$ the fugacity.

We then have the remarkable formula:
\be\label{eq:Ising}
\cD^{\mathrm{2PI}}(\lambda, M) = M^{2/3} \left( Z_{++} (\cC_\No (\lambda^2) M^{-2/3} , \lambda^{-2}) + \lambda^2 Z_{+-} (\cC_\No (\lambda^2) M^{-2/3} , \lambda^{-2})  \right)\,,
\ee
where we have used \eqref{eq:generating-scheme} (with $r=l=b=0$, $p = 0$ or $1$) and $\cC_{\No} (u) = u \cC_{\Ne} (u)$.

\medskip

The Ising and Potts models on various families of random planar maps can be solved explicitly, for instance by means of matrix model techniques \cite{Kazakov:1986hu,Boulatov:1986sb,Kazakov:1987qg}, 
or of a general method based on \emph{Tutte equations with two catalytic variables}\footnote{The basic idea of this method is to introduce additional parameters, known as catalytic variables, that allow to derive tractable equations for the partition function. In the case of the Potts model on a random map, keeping track of the degree of the root-face and the degree of the root-vertex enables the use of deletion-contraction relations. This explains the need to introduction two catalytic variables.} \cite{bernardi2011counting, bernardi2017counting}. The specific case of the Ising model has also been solved by means of an exact mapping to a problem of map enumeration \cite{bousquet2002degree}, which shares similarities with our own construction. Indeed, the maps being enumerated in this work are bipartite, and are constructed by further decorations of ordinary Ising states on ordinary maps: given such a colored map, one simply adds an arbitrary odd (resp. even) number of bivalent vertices on a $++$ or $--$ (resp. $+-$) edge, and color them in the unique way that results in a bipartite map. It is illuminating to realize that we would obtain the exact same structure if we were to unfold the N-vertices as sums of N-ladders in our initial 2PI-dominant schemes (each rung in a ladder being now seen as a bivalent vertex). Hence, our construction can be seen as the inverse of the type of mapping considered in \cite{bousquet2002degree}. From this point of view, the relation with the Ising model established in equation \eqref{eq:Ising} is not so surprising.

\subsection{A matrix model for the 2PI-dominant graphs}\label{sec:effective-matrix}

The mapping between 2PI-dominant schemes and Ising model configurations on random planar 1PI maps can be encoded in the planar limit of a matrix model.
One simply needs to introduce two matrices, whose self-interactions represent the vertices with spin up and down, and with propagators carrying the spin-spin interaction \cite{Kazakov:1986hu,Boulatov:1986sb}.

We introduce two $L\times L$ Hermitian matrices $A$ and $B$, with the following free energy:
\be
W[u,M] = \lim_{L\to \infty} \f{1}{L^2} \ln \int [dA] [dB] e^{-S_{\rm eff}[A,B] } \, ,
\ee
where the action is:
\be \label{eq:Seff-1PI}
S_{\rm eff}[A,B]  = M^2 L \,\Tr\left[ \f{1}{2 u^2} \left(A^2+B^2\right) - \f{1}{u} A B -  \f{1}{3} \left(A^3+B^3\right) - j_1 A-j_2 B \right]  \,.
\ee
In order to impose the 1PI restriction on the Feynman graphs, we choose $j_1$ and $j_2$ to satisfy
\be
0 = \f{\p}{\p j_1} W[u,M] = \f{\p}{\p j_2} W[u,M]\,,
\ee
whose solution, with $j_1=j_2$ because of the symmetry under exchange of $A$ and $B$, we denote $\tilde{j}(u,M)$.

The global factor $L$ in the action is the standard one, required for achieving the usual topological expansion of matrix models (with genus $\eta$ related to the grade $\ell$ by Eq.~\eqref{eq:grade=genus}), while the factor $M^2$ is chosen so to attribute the correct scaling in $M$ to graphs with $\m=g$ loops (see \eqref{eq:genus=loops} and \eqref{eq:D-2PI-dominant}). 

The coefficients in the action have been chosen so that the free propagators (i.e.\ the two-point functions of the theory with neither cubic nor univalent vertices) match the generating functions of $\Ne$ and $\No$ ladders:
\begin{align} \label{eq:AA-prop}
&\la (A)_{ab} (A)_{cd} \ra_{\rm free} = \la (B)_{ab} (B)_{cd} \ra_{\rm free}   = \f{1}{M^2L}\d_{ad}\d_{bc}\, \cC_{\Ne}(u) \, , \\
 \label{eq:AB-prop}
& \la (A)_{ab} (B)_{cd} \ra_{\rm free} = \f{1}{M^2L} \d_{ad}\d_{bc}\, \cC_{\No}(u) \, ,
\end{align}
and as discussed below Eq.~\eqref{eq:2PI-G_g}, we have $u=\l^2$ because there are no melonic insertions.

For $j=0$, the action \eqref{eq:Seff-1PI}, is the same as in \cite{Boulatov:1986sb,Eynard:1992cn,Staudacher:1993xy}, with the mapping $c=e^{-2\beta}\to u$, $g\to u^6/M^2$, and a rescaling of $A$ and $B$ by $u^2$.
Therefore, in principle we could adapt their methods and results to our case, in order to compute $\cD^{\mathrm{2PI}}(\lambda, M)$. 
However, we are interested in the restriction to 1PI graphs, which they did not consider explicitly, and moreover, the precise relation to $\cD^{\mathrm{2PI}}(\lambda, M)$ requires to compute a combination of two-point functions, including $\la\Tr[AB]\ra$, which requires some extra work.
In fact, comparing to Eq.~\eqref{eq:Ising}, we have
\be \label{eq:D2PI-matrix}
\cD^{\mathrm{2PI}}(\lambda, M) = \lim_{L\to \infty} \f{M^2}{L} \left(  \left\la \Tr[A^2] \right\ra_{\tilde{j}} + \l^2 \left\la \Tr[AB] \right\ra_{\tilde{j}} \right)_{|_{u=\l^2}}\,.
\ee
Nevertheless, we expect the universal critical properties of the model not to be affected by the 1PI restriction or by the insertion of a special two-valent vertex, hence we can immediately anticipate some conclusions.

The free energy of the matrix model \eqref{eq:Seff-1PI} corresponds in the large-$L$ limit to the grand-canonical partition function of an Ising model on random planar 1PI graphs, at inverse temperature $\b=\f12 \ln (\f1u)$. The thermodynamic limit (with infinite graphs dominating the grand-canonical partition function) is obtained by tuning the coupling to its critical line, in our case $M\to M_c(u)$, with $u\in (0,1)$.
In the thermodynamic limit, the Ising model can also reach criticality at some critical point $u=u_c$, with $0<u_c<1$ ($u_c=(2\sqrt{7}-1)/27$ for the model with $j_1=j_2=0$ \cite{Boulatov:1986sb}), hence one can have a continuum limit describing matter coupled to quantum gravity in two dimensions by tuning $u\to u_c$ and $M\to M_c(u_c)$. However,  for $u\to 1$, which is the only relevant limit for us, we are in the high-temperature limit $\b\to 0$, and the Ising spins become completely uncorrelated, thus simply contributing a factor $2^\n$ ($\n$ being the number of vertices in the graph) to the pure gravity partition function.
Therefore, the limit $u\to 1$ corresponds to the high-temperature limit of the Ising model  on random planar graphs, which is in the universality class of pure two-dimensional quantum gravity.

That the latter is unaffected by the 1PI restriction can easily be verified from the solution of the one matrix model with action $S[Y]=L(\f12 \Tr[Y^2]-\f{\a}{3}\Tr[Y^3]-j\Tr[Y])$, where $j$ is chosen in such a way that the one-point function $\la \Tr[Y]\ra$ vanishes. Such a model was solved at large-$L$ in \cite{Brezin:1977sv}, from which one finds for the two-point function:
\be
\left\la \Tr[Y^2] \right\ra = \f{\t}{\a^2} (1-3\t) \,,
\ee
where $\t= \sum_{n\geq 1} 2^{n-1} \f{(3n-2)!}{n! (2n-1)!} \a^{2n}$. With a bit of work we find:
\be\label{eq:2point_matrix}
\left\la \Tr[Y^2] \right\ra = \sum_{n\geq 0} 2^{n} \f{(3n)!}{(n+1)! (2n+1)!} \a^{2n}\sim  \frac{1}{4} \sqrt{\frac{3}{\pi}} \sum_{n\geq 0} \left( \frac{27\, \a^2}{2} \right)^n n^{-5/2}  \,,
\ee
from which, taking into consideration the marked edge, one can read off the usual string susceptibility exponent  of pure two-dimensional quantum gravity, $\g_s=-1/2$.

We conclude this subsection by noticing that the limit $u\to 1$ should be taken carefully, as the propagators diverge in such limit. On the other hand, from \cite{Boulatov:1986sb} one finds that the critical line ends at zero coupling,\footnote{There is actually a typo in equation (42) of \cite{Boulatov:1986sb}, the term $c/4$ should be absent.} i.e.\ $M_c(u)\to +\infty$ for $u\to 1$. Moreover, the graphs with $g$ loops contributing to \eqref{eq:D2PI-matrix} have by construction a factor $M^{2-2g}$, and they have a factor $(1-u^2)^{2-3g}$ from the propagators. Therefore, we can approach the high-temperature and thermodynamic limits simultaneously and without problems by keeping $M(1-u^2)^{3/2}$ fixed. We will discuss such triple-scaling limit in more detail in the next subsection.

\subsection{Triple-scaling limit}\label{sec:2PI-enumeration}

We do not need to solve the effective Ising model described in subsection \ref{sec:Ising} to understand the critical properties of the 2PI-dominant generating function \eqref{eq:D-2PI-dominant}. Indeed, equation \eqref{eq:Ising}, together with $\cC_{\No}(u)=u^3/(1-u^2)$, suggest to define the triple-scaling limit in which $M \to +\infty$ and $\lambda \to 1^-$, while keeping the following quantity finite:
\be
\kappa^{-1} = M \left( 1 - \lambda\right)^{3/2} \,.
\ee
The generating function of 2PI-dominant schemes has then a nice limit:
\be\label{eq:D2PI-triple}
(1 - \lambda) \cD^{\mathrm{2PI}}(\lambda, M) \rightarrow  \frac{1}{\kappa^{2/3}} \left( Z_{++}\left(\frac{1}{4}\kappa^{2/3} , 1\right) + Z_{+-}\left(\frac{1}{4}\kappa^{2/3}, 1\right)  \right) =: \tilde\cD(\kappa)\,.
\ee
Note that this is a series in $\kappa$, since a genus $g$ contribution in $Z_{++}$ or $Z_{+-}$ behaves like $(\kappa^{2/3})^{3g-2} = \kappa^{2/3} \kappa^{2g-2}$. 

In fact, we can evaluate $\tilde\cD(\kappa)$ more explicitly. We first note that, since the Ising configurations of a given combinatorial map are weighted uniformly in this limit, $Z_{++}(\cdot,1)$ and $Z_{+-}(\cdot,1)$ both reduce to the same combinatorial sum. Furthermore, one can adopt a slightly simpler combinatorial description in terms of cubic maps (by gluing back the univalent vertices of a given map together, to form a root-edge). All in all, we obtain:
\be\label{eq:2PI-triple}
\tilde\cD(\kappa) = 2 \kappa^{-2/3} \sum_{g \in \mathbb{N}^*} \left(\frac{1}{4} \kappa^{2/3}\right)^{3g-2} 2^{2g-2} \cM_{g-1} = \frac{1}{2} \sum_{n \in \mathbb{N}} \left(\frac{\kappa^2}{16}\right)^n \cM_n\,,
\ee
where $\cM_n$ is the number of rooted bridgeless planar cubic maps with $2n$ vertices, and by convention $\cM_0 = 1$ (which corresponds to the cycle graph, and correctly counts the schemes $S_1$ and $S_2$). $\cM_n$ is the A000309 integer sequence of the OEIS classification\footnote{\url{https://oeis.org/A000309}.}, and is known in closed form \cite{tutte_1962}\footnote{This result can be derived from a Tutte equation involving a single catalytic variable, in contrast to the complete Ising model \eqref{eq:Ising}. Away from the triple-scaling regime, the latter falls into the more challenging class of problems which are governed by Tutte equations with two catalytic variables.}:
\be\label{eq:asymptotics}
\cM_n = \frac{2^{n}(3n)!}{(n+1)!(2n+1)!} \sim \frac{1}{4} \sqrt{\frac{3}{\pi}} \left( \frac{27}{2} \right)^n n^{-5/2} \,.
\ee
We have encountered it before, in the expansion of Eq.~\eqref{eq:2point_matrix}.
Comparing \eqref{eq:2PI-triple} to \eqref{eq:2point_matrix}, we notice the following differences: an overall factor $2\k^{-2/3}$, which comes from the overall factor in \eqref{eq:D2PI-triple} and from the two contributing configurations; a factor $2^{2g-2}=2^{\n}$, which accounts for the uncorrelated up and down spin configurations; instead of a weight $\a$ per vertex we have a weight $\frac{1}{4} \kappa^{2/3}$ per edge. 
As anticipated, we find the $n^{-5/2}$ term characteristic of the universality class of random planar maps. 

From equations \eqref{eq:2PI-triple} and \eqref{eq:asymptotics}, we infer that $\tilde\cD(\kappa)$ has a finite radius of convergence $\kappa_c = \frac{8}{3\sqrt{6}}$, and the following singular behaviour:
\be
\tilde\cD(\kappa) \underset{\kappa \to \kappa_c^-}{\sim} \frac{1}{2\sqrt{3}} \left( 1 - \frac{\kappa^2}{\kappa_c^2}\right)^{3/2} + \mathrm{more}\;\mathrm{regular}\;\mathrm{terms}\,.
\ee
In this regime, the expectation value of the genus $\langle g\rangle = \langle n+1 \rangle$ remains finite, but its variance diverges:
\be
\langle g^2 \rangle \underset{\kappa \to \kappa_c^-}{\sim} K \, \left( 1 - \frac{\kappa^2}{\kappa_c^2} \right)^{-1/2}\,,
\ee
for some constant $K >0$.

\section*{Acknowledgements}

SC acknowledges support from Perimeter Institute. GV and RT are grateful for the hospitality of Perimeter Institute where part of this work was carried out. Research at Perimeter Institute is supported in part by the Government of Canada through the Department of Innovation, Science and Economic Development Canada and by the Province of Ontario through the Ministry of Colleges and Universities. GV is a Research Fellow at the Belgian F.R.S.-FNRS.

The work of DB is supported by the European Research Council (ERC) under the European Union's Horizon 2020 research and innovation program (grant agreement No818066).
DB and GV are grateful for the hospitality of Okinawa Institute of Science and Technology where part of this work was carried out.

\
\appendix
 
\section{Algorithmic generation of graphs with vanishing grade: examples}
\label{app:algorithmex}
In this Appendix, we provide examples of graphs with vanishing grade, which we obtain by applying the algorithm outlined in section \ref{sec:FGVanishingGrade}. 
We note that all the graphs we have explicitly constructed in this way have a very simple structure: they are collections of ladders or dipoles glued together through effective vertices. Furthermore, we have empirically found that three types of effective vertices are sufficient to describe all such graphs: planar $2n$-point vertices (with $n\geq3$), as described in Theorem~\ref{thm:induction} (for $n=3$) and in section \ref{sec:effective-matrix}; contact $6$-point vertices, also defined in Theorem~\ref{thm:induction}; and $8$-point contact vertices. The latter $8$-point vertices do not make any apparition in the main text, because they contribute to neither of the two continuum limits we have investigated. An example will however appear below. Since it is tangential to the main objectives of this article, whether the three types of effective vertices we have encountered is sufficient to describe $\ell = 0$ graphs of arbitrary genus is left as an open question. 

\medskip

We now illustrate all the ways in which one can obtain a $\ell = 0$ graph of genus two, starting from the two schemes of genus one $S_1$ and $S_2$ described in Proposition \ref{propo:g1}. 
\begin{itemize}

\item
Inserting a connecting N-dipole or $\No$-ladder (Figure \ref{fig:algorithm1}) on the two edges of a rung increases the genus $g$ by one, and generates an $8$-point contact vertex in the resulting graph. If the original graph is 2PI, then the resulting graph remains 2PI.
\begin{figure}[H]
\centering
\includegraphics[scale=.5]{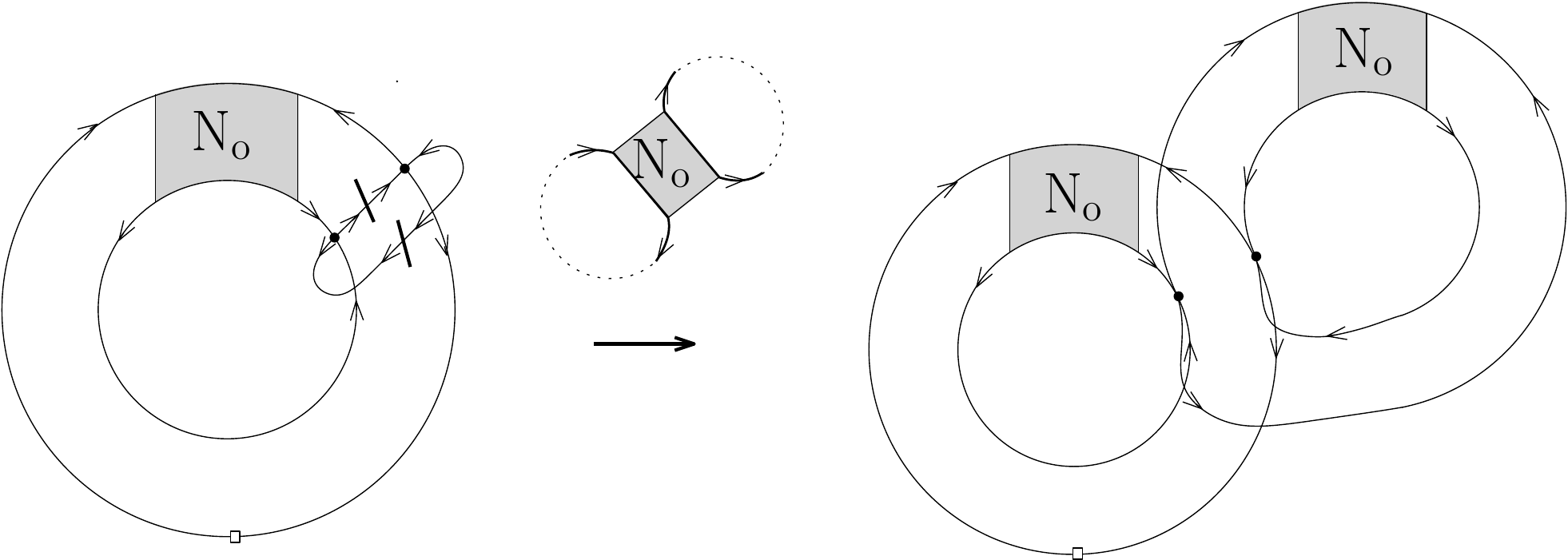}
\caption{\small Insertion of a connecting $\No$-ladder on the two edges of a rung in a 2PI $g=1$ graph (with scheme $S_1$, as defined in proposition \ref{propo:g1}). The resulting graph is 2PI, has genus $g+1=2$, and contains an $8$-point contact vertex.}
\label{fig:algorithm1}
\end{figure}

\item
Inserting a connecting $\Ne$-ladder between two edges of a rail increases the genus $g$ by one and generates planar effective vertices (in the sense of the effective matrix model of section~\ref{sec:effective-matrix}). If the original graph is 2PI, then the resulting graph remains 2PI. Both $\Ne$- and $\No$-ladders can be inserted in this manner, as illustrated on Figures \ref{fig:algorithm2} and \ref{fig:algorithm3}.
\begin{figure}[H]
\centering
\includegraphics[scale=.5]{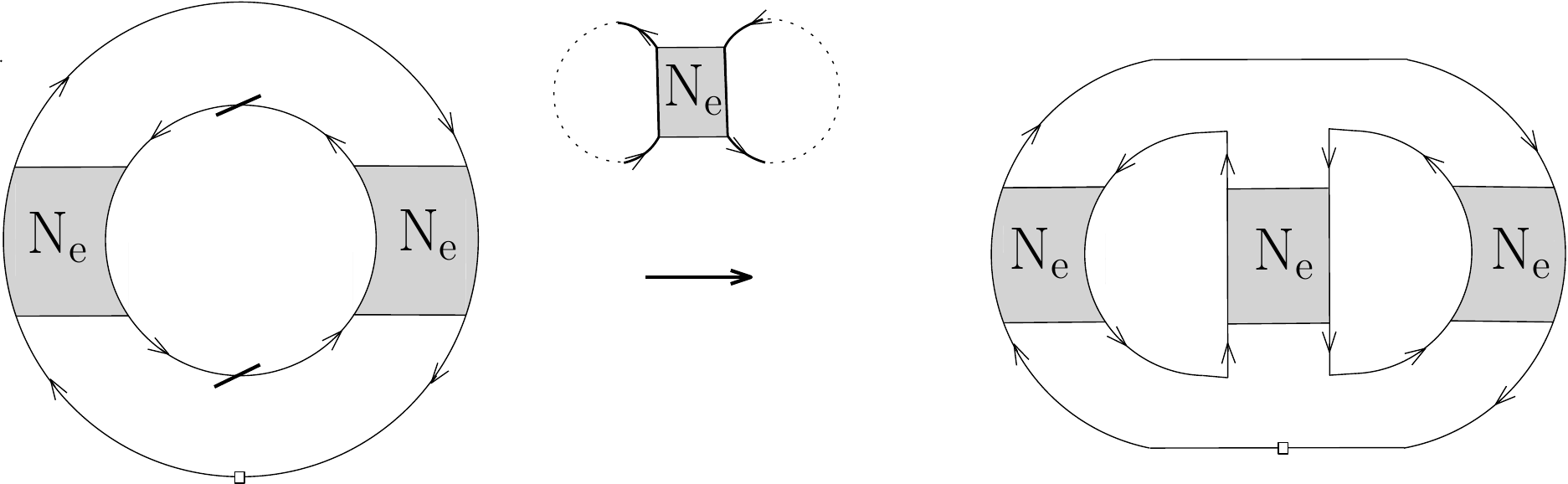}
\caption{\small Insertion of a connecting $\Ne$-ladder in-between two edges of a rung. Starting from a 2PI $g=1$ graph with scheme $S_1$ (as defined in proposition \ref{propo:g1}), we obtain a 2PI graph of genus $g+1 =2$. This graph has a planar structure, well captured by the effective matrix model of section \ref{sec:effective-matrix}.} 
\label{fig:algorithm2}
\end{figure}

\begin{figure}[H]
\centering
\includegraphics[scale=.5]{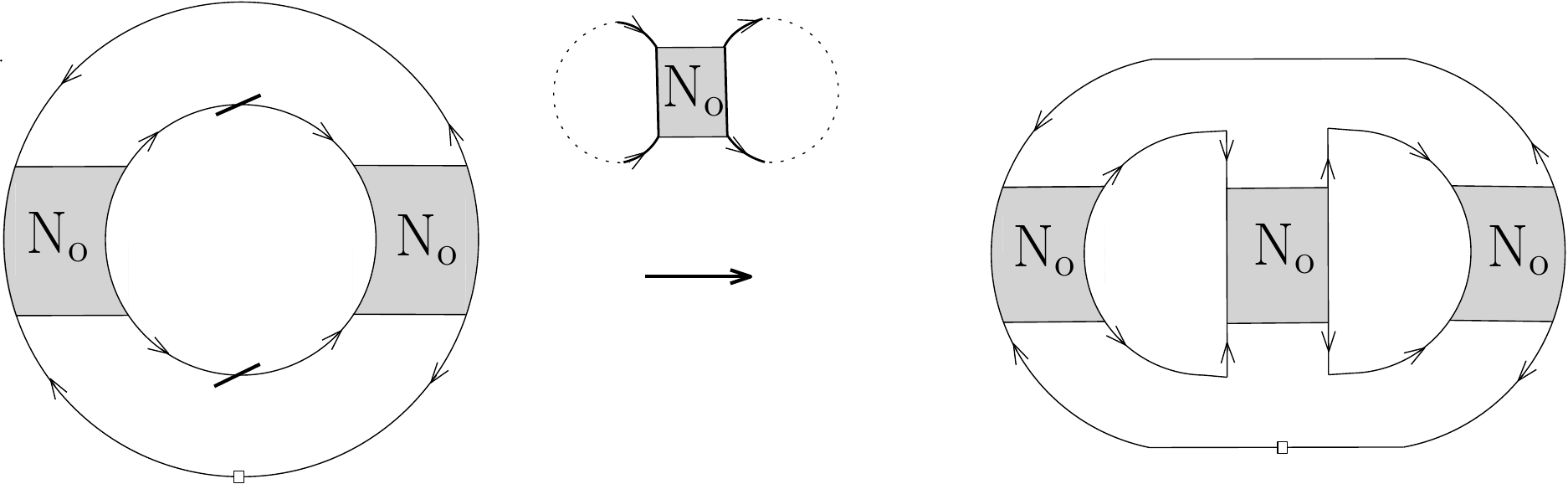}
\caption{\small Insertion of a connecting $\No$-ladder in-between two edges of a rung. Note that, in contrast to Figure \ref{fig:algorithm2}, we have chosen the two edges to be cut in such a way that they separate two $\No$-ladders. 
This contribution is again included in the planar limit of the effective matrix model of section \ref{sec:effective-matrix}.}
\label{fig:algorithm3}
\end{figure}

\item
Inserting a separating dipole, a separating ladder, or a two-edge connection in-between a rung of a graph of genus $g_1$, and a rung of a graph of genus $g_2$, yields a 2PR graph of genus $g_1 + g_2$.
The example given in Fig. \ref{fig:algorithm4} starts out with two 2PI graphs of genus one, and results in a 2PR graph of genus two. 
Two $6$-point contact vertices (as defined in Theorem \ref{thm:induction}) are generated in the process.
\begin{figure}[H]
\centering
\includegraphics[scale=.6]{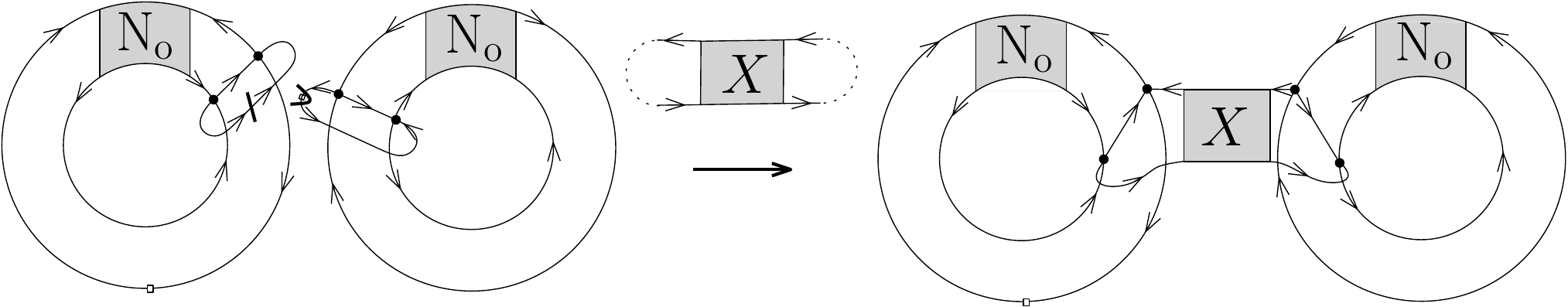}
\caption{\small Insertion of a separating dipole, a separating ladder, or a two-edge connection $X$, in-between two rungs. Starting from two genus one graphs with scheme $S_1$ (Proposition \ref{propo:g1}), we obtain a 2PR graph of genus 2 (in this particular example, $X \in \{\Ne, \; \mathrm{L}, \; \mathrm{R}, \; \mathrm{B}, \; \emptyset \}$).}
\label{fig:algorithm4}
\end{figure}

\item Inserting a separating dipole, a separating ladder, or a two-edge connection in-between a rail of a graph of genus $g_1$, and a rail of a graph of genus $g_2$, yields a 2PR graph of genus $g_1 + g_2$. Planar vertices of the type described in Theorem~\ref{thm:induction} or section \ref{sec:effective-matrix} are generated in the process.
The example given in Fig. \ref{fig:algorithm6} starts out with two 2PI graphs with genus one, and results in a 2PR graph of genus two. When $X$ is a B-ladder, this is well captured by the induction of Theorem~\ref{thm:induction}. 

\begin{figure}[H]
\centering
\includegraphics[scale=.6]{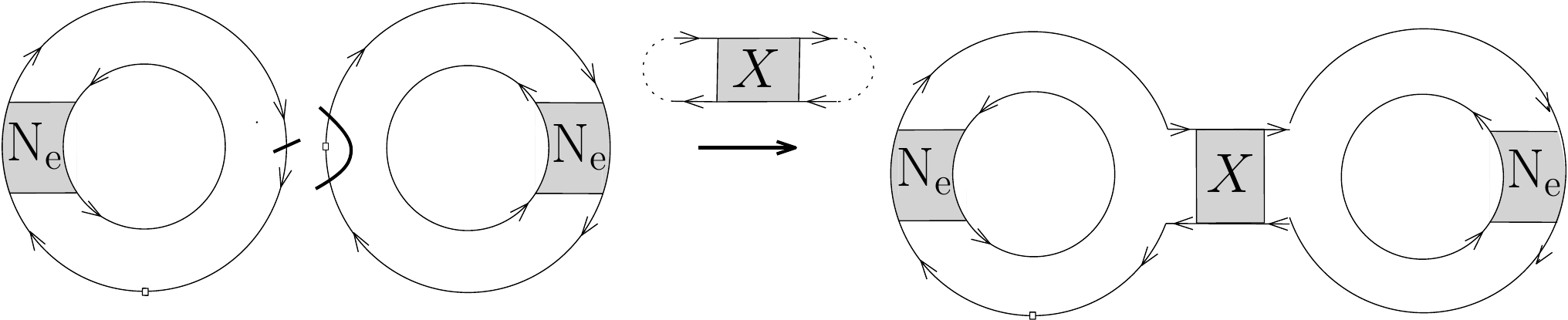}
\caption{\small Insertion of a separating dipole, a separating ladder, or a two-edge connection $X$, in-between two rails. Starting from two genus one graphs with scheme $S_1$ (Proposition \ref{propo:g1}), we obtain a 2PR graph of genus 2 (in this particular example, $X \in \{\Ne, \; \mathrm{L}, \; \mathrm{R}, \; \mathrm{B}, \; \emptyset \}$).}

\label{fig:algorithm6}
\end{figure}

\item Inserting a separating dipole, a separating ladder, or a two-edge connection in-between a rung of a graph of genus $g_1$, and a rail of a graph of genus $g_2$, yields a 2PR graph of genus $g_1 + g_2$. 
Following this procedure, we generate a $6$-point contact vertex and a $6$-point planar vertex. An example is provided in Fig. \ref{fig:algorithm8}. 

\begin{figure}[H]
\centering
\includegraphics[scale=.6]{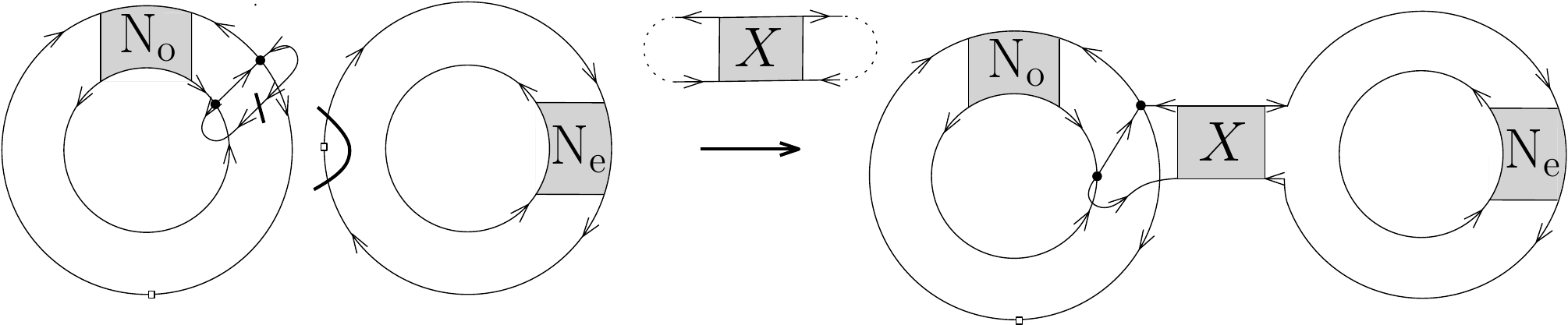}
\caption{\small Insertion of a separating dipole, a separating ladder, or a two-edge connection $X$, in-between a rung and a rail. Starting from two genus one graphs with scheme $S_1$ (Proposition \ref{propo:g1}), we obtain a 2PR graph of genus 2. In this particular example, $X$ must be a $\No$-ladder or a N-dipole.}
\label{fig:algorithm8}
\end{figure}

\end{itemize}


\setlength{\itemsep}{-1pt}
\providecommand{\href}[2]{#2}\begingroup\raggedright\endgroup

\end{document}